\documentclass[11pt, a4paper,onecolumn,allowtoday,nopdfoutputerror,accepted=2026-07-15]{quantumarticle}

\usepackage{rotating}
\usepackage{floatpag}
\usepackage{xifthen}

\rotfloatpagestyle{empty}
\renewcommand{\theenumi}{\arabic{enumi}.}
\ifthenelse{\isundefined{\showoptional}}{
  \newcommand{\showoptional}{1}
}{}

\ifthenelse{\isundefined{\ismain}}{
  \newcommand{\ismain}{0}
}{}

\ifthenelse{\isundefined{\lecturenotes}}{
  \newcommand{\lecturenotes}{0}
}{}

\usepackage{hyperref}

\usepackage{amsmath,amsthm,amssymb,mathabx}

\usepackage{xspace,enumerate,color,epsfig}
\usepackage{graphicx}
\graphicspath{{.}{./figures/}}

\usepackage{tikzfig}
\usepackage{stmaryrd}
\usepackage{docmute}
\usepackage{keycommand}

\newkeycommand{\wirelabel}[style=white label]{\,\tikz{\node[style=\commandkey{style}]{};}\,\xspace}

\newkeycommand{\pointmap}[style=point][1]{\,\tikz{\node[style=\commandkey{style}] (x) at (0,-0.3) {$#1$};\node [style=none] (3) at (0, 0.9) {};\node [style=none] (3) at (0, -0.9) {};\draw (x)--(0,0.7);}\,}

\newkeycommand{\copointmap}[style=copoint][1]{\,\tikz{\node[style=\commandkey{style}] (x) at (0,0.3) {$#1$};\draw (x)--(0,-0.7);}\,}

\newkeycommand{\dpoint}[style=point][1]{\,\tikz{\node[style=\commandkey{style},doubled] (x) at (0,-0.3) {$#1$};\draw[boldedge] (x)--(0,0.7);}\,}

\newkeycommand{\kpoint}[style=kpoint,estyle=][1]{\,\tikz{\node[style=\commandkey{style}] (x) at (0,-0.05) {$#1$};\draw [\commandkey{estyle}] (x)--(0,1.05);}\,}

\newkeycommand{\kpointgray}[style=gray kpoint,estyle=][1]{\,\tikz{\node[style=\commandkey{style}] (x) at (0,-0.05) {$#1$};\draw [\commandkey{estyle}] (x)--(0,1.05);}\,}

\newkeycommand{\typedkpoint}[style=kpoint,edgestyle=][2]{%
\,\begin{tikzpicture}
  \begin{pgfonlayer}{nodelayer}
    \node [style=none] (0) at (0, 1) {};
    \node [style=\commandkey{style}] (1) at (0, -0.75) {$#2$};
    \node [style=right label] (2) at (0.25, 0.5) {$#1$};
  \end{pgfonlayer}
  \begin{pgfonlayer}{edgelayer}
    \draw [\commandkey{edgestyle}] (1) to (0.center);
  \end{pgfonlayer}
\end{tikzpicture}\,}

\newkeycommand{\typedkpointdag}[style=kpoint adjoint,edgestyle=][2]{%
\,\begin{tikzpicture}
  \begin{pgfonlayer}{nodelayer}
    \node [style=none] (0) at (0, -1) {};
    \node [style=\commandkey{style}] (1) at (0, 0.75) {$#2$};
    \node [style=right label] (2) at (0.25, -0.5) {$#1$};
  \end{pgfonlayer}
  \begin{pgfonlayer}{edgelayer}
    \draw [\commandkey{edgestyle}] (0.center) to (1);
  \end{pgfonlayer}
\end{tikzpicture}\,}

\newkeycommand{\bistate}[style=kpoint][1]{\,\begin{tikzpicture}
  \begin{pgfonlayer}{nodelayer}
    \node [style=\commandkey{style}, minimum width=1 cm, inner sep=2pt] (0) at (0, -0.25) {$#1$};
    \node [style=none] (1) at (-0.75, 0) {};
    \node [style=none] (2) at (0.75, 0) {};
    \node [style=none] (3) at (-0.75, 0.88) {};
    \node [style=none] (4) at (0.75, 0.88) {};
  \end{pgfonlayer}
  \begin{pgfonlayer}{edgelayer}
    \draw (1.center) to (3.center);
    \draw (2.center) to (4.center);
  \end{pgfonlayer}
\end{tikzpicture}\,}

\newkeycommand{\bistatebig}[style=kpoint][1]{\,\begin{tikzpicture}
  \begin{pgfonlayer}{nodelayer}
    \node [style=\commandkey{style}, minimum width=, minimum width=1.26 cm, inner sep=2pt] (0) at (0, -0.25) {$#1$};
    \node [style=none] (1) at (-1, 0) {};
    \node [style=none] (2) at (1, 0) {};
    \node [style=none] (3) at (-1, 0.88) {};
    \node [style=none] (4) at (1, 0.88) {};
  \end{pgfonlayer}
  \begin{pgfonlayer}{edgelayer}
    \draw (1.center) to (3.center);
    \draw (2.center) to (4.center);
  \end{pgfonlayer}
\end{tikzpicture}\,}

\newkeycommand{\dbistate}[style=dkpoint][1]{\,\begin{tikzpicture}
  \begin{pgfonlayer}{nodelayer}
    \node [style=\commandkey{style}, minimum width=1 cm, inner sep=2pt] (0) at (0, -0.25) {$#1$};
    \node [style=none] (1) at (-0.75, 0) {};
    \node [style=none] (2) at (0.75, 0) {};
    \node [style=none] (3) at (-0.75, 1.0) {};
    \node [style=none] (4) at (0.75, 1.0) {};
  \end{pgfonlayer}
  \begin{pgfonlayer}{edgelayer}
    \draw [boldedge] (1.center) to (3.center);
    \draw [boldedge] (2.center) to (4.center);
  \end{pgfonlayer}
\end{tikzpicture}\,}

\newkeycommand{\bistatebraket}[style1=kpoint,style2=kpointadj][2]{\,\begin{tikzpicture}
  \begin{pgfonlayer}{nodelayer}
    \node [style=\commandkey{style1}, minimum width=1 cm, inner sep=2pt] (0) at (0, -0.75) {$#1$};
    \node [style=none] (1) at (-0.75, -0.5) {};
    \node [style=none] (2) at (0.75, -0.5) {};
    \node [style=none] (3) at (-0.75, 0.5) {};
    \node [style=none] (4) at (0.75, 0.5) {};
    \node [style=\commandkey{style2}, minimum width=1 cm, inner sep=2pt] (5) at (0, 0.75) {$#2$};
  \end{pgfonlayer}
  \begin{pgfonlayer}{edgelayer}
    \draw (1.center) to (3.center);
    \draw (2.center) to (4.center);
  \end{pgfonlayer}
\end{tikzpicture}\,}

\newkeycommand{\dkpoint}[style=dkpoint][1]{\,\tikz{\node[style=\commandkey{style}] (x) at (0,-0.1) {$#1$};\draw[boldedge] (x)--(0,1);}\,}
\newkeycommand{\dkpointadj}[style=dkpointadj][1]{\,\tikz{\node[style=\commandkey{style}] (x) at (0,0.1) {$#1$};\draw[boldedge] (x)--(0,-1);}\,}
\newkeycommand{\dkpointtrans}[style=dkpointtrans][1]{\,\tikz{\node[style=\commandkey{style}] (x) at (0,0.1) {$#1$};\draw[boldedge] (x)--(0,-1);}\,}
\newkeycommand{\dkpointconj}[style=dkpointconj][1]{\,\tikz{\node[style=\commandkey{style}] (x) at (0,-0.1) {$#1$};\draw[boldedge] (x)--(0,1);}\,}

\newcommand{\trace}{\,\begin{tikzpicture}
  \begin{pgfonlayer}{nodelayer}
    \node [style=none] (0) at (0, -0.60) {};
    \node [style=upground] (1) at (0, 0.40) {};
  \end{pgfonlayer}
  \begin{pgfonlayer}{edgelayer}
    \draw [style=boldedge] (0.center) to (1);
  \end{pgfonlayer}
\end{tikzpicture}\,}

\newcommand\discard\trace

\newkeycommand{\pointketbra}[style1=copoint,style2=point][2]{\,%
\begin{tikzpicture}
  \begin{pgfonlayer}{nodelayer}
    \node [style=none] (0) at (0, -1.75) {};
    \node [style=\commandkey{style1}] (1) at (0, -0.75) {$#1$};
    \node [style=\commandkey{style2}] (2) at (0, 0.75) {$#2$};
    \node [style=none] (3) at (0, 1.75) {};
  \end{pgfonlayer}
  \begin{pgfonlayer}{edgelayer}
    \draw (0.center) to (1);
    \draw (2) to (3.center);
  \end{pgfonlayer}
\end{tikzpicture}\,}

\newkeycommand{\twocopointketbra}[style1=copoint,style2=copoint,style3=point][3]{\,%
\begin{tikzpicture}
  \begin{pgfonlayer}{nodelayer}
    \node [style=none] (0) at (-0.75, -1.75) {};
    \node [style=none] (1) at (0.75, -1.75) {};
    \node [style=\commandkey{style1}] (2) at (-0.75, -0.75) {$#1$};
    \node [style=\commandkey{style2}] (3) at (0.75, -0.75) {$#2$};
    \node [style=\commandkey{style3}] (4) at (0, 0.75) {$#3$};
    \node [style=none] (5) at (0, 1.75) {};
  \end{pgfonlayer}
  \begin{pgfonlayer}{edgelayer}
    \draw (0.center) to (2);
    \draw (1.center) to (3);
    \draw (4) to (5.center);
  \end{pgfonlayer}
\end{tikzpicture}\,}

\newkeycommand{\twopointketbra}[style1=copoint,style2=point,style3=point][3]{\,%
\begin{tikzpicture}
  \begin{pgfonlayer}{nodelayer}
    \node [style=none] (0) at (0, -1.75) {};
    \node [style=none] (1) at (-0.75, 1.75) {};
    \node [style=\commandkey{style1}] (2) at (0, -0.75) {$#1$};
    \node [style=\commandkey{style2}] (3) at (-0.75, 0.75) {$#2$};
    \node [style=\commandkey{style3}] (4) at (0.75, 0.75) {$#3$};
    \node [style=none] (5) at (0.75, 1.75) {};
  \end{pgfonlayer}
  \begin{pgfonlayer}{edgelayer}
    \draw (0.center) to (2);
    \draw (1.center) to (3);
    \draw (4) to (5.center);
  \end{pgfonlayer}
\end{tikzpicture}\,}

\newkeycommand{\pointbraket}[style1=point,style2=copoint,estyle=][2]{\,%
\begin{tikzpicture}
  \begin{pgfonlayer}{nodelayer}
    \node [style=\commandkey{style1}] (0) at (0, -0.625) {$#1$};
    \node [style=\commandkey{style2}] (1) at (0, 0.625) {$#2$};
  \end{pgfonlayer}
  \begin{pgfonlayer}{edgelayer}
    \draw [\commandkey{estyle}] (0) to (1);
  \end{pgfonlayer}
\end{tikzpicture}\,}

\newkeycommand{\kpointketbra}[style1=kpointdag,style2=kpoint,estyle=][2]{\,%
\begin{tikzpicture}
  \begin{pgfonlayer}{nodelayer}
    \node [style=none] (0) at (0, -1.9) {};
    \node [style=\commandkey{style1}] (1) at (0, -0.95) {$#1$};
    \node [style=\commandkey{style2}] (2) at (0, 0.95) {$#2$};
    \node [style=none] (3) at (0, 1.9) {};
  \end{pgfonlayer}
  \begin{pgfonlayer}{edgelayer}
    \draw [\commandkey{estyle}] (0.center) to (1);
    \draw [\commandkey{estyle}] (2) to (3.center);
  \end{pgfonlayer}
\end{tikzpicture}\,}

\newkeycommand{\kpointmap}[style1=kpoint,style2=map][2]{\,%
\begin{tikzpicture}
  \begin{pgfonlayer}{nodelayer}
    \node [style=\commandkey{style1}] (0) at (0, -1.1) {$#1$};
    \node [style=\commandkey{style2}] (1) at (0, 0.5) {$#2$};
    \node [style=none] (2) at (0, 1.75) {};
  \end{pgfonlayer}
  \begin{pgfonlayer}{edgelayer}
    \draw (0) to (1);
    \draw (1) to (2.center);
  \end{pgfonlayer}
\end{tikzpicture}%
\,}

\newkeycommand{\dkpointmap}[style1=dkpoint,style2=dmap][2]{\,%
\begin{tikzpicture}
  \begin{pgfonlayer}{nodelayer}
    \node [style=\commandkey{style1}] (0) at (0, -1.2) {$#1$};
    \node [style=\commandkey{style2}] (1) at (0, 0.4) {$#2$};
    \node [style=none] (2) at (0, 1.7) {};
  \end{pgfonlayer}
  \begin{pgfonlayer}{edgelayer}
    \draw[style=boldedge] (0) to (1);
    \draw[style=boldedge] (1) to (2.center);
  \end{pgfonlayer}
\end{tikzpicture}%
\,}

\newkeycommand{\dkpointmapx}[style1=dkpoint,style2=dmap][2]{\,%
\begin{tikzpicture}
  \begin{pgfonlayer}{nodelayer}
    \node [style=\commandkey{style1}] (0) at (0, -1.4) {$#1$};
    \node [style=\commandkey{style2}] (1) at (0, 0.6) {$#2$};
    \node [style=none] (2) at (0, 1.9) {};
  \end{pgfonlayer}
  \begin{pgfonlayer}{edgelayer}
    \draw[style=boldedge] (0) to (1);
    \draw[style=boldedge] (1) to (2.center);
  \end{pgfonlayer}
\end{tikzpicture}%
\,}

\newkeycommand{\kpointsandwichmap}[style1=kpoint,style2=map,style3=kpointdag][3]{\,%
\begin{tikzpicture}
  \begin{pgfonlayer}{nodelayer}
    \node [style=\commandkey{style1}] (0) at (0, -1.5) {$#1$};
    \node [style=\commandkey{style2}] (1) at (0, 0) {$#2$};
    \node [style=\commandkey{style3}] (2) at (0, 1.5) {$#3$};
  \end{pgfonlayer}
  \begin{pgfonlayer}{edgelayer}
    \draw (0) to (1);
    \draw (1) to (2);
  \end{pgfonlayer}
\end{tikzpicture}%
}

\newkeycommand{\sandwichtwo}[style1=point,style2=map,style3=map,style4=copoint][4]{\,%
\begin{tikzpicture}
  \begin{pgfonlayer}{nodelayer}
    \node [style=\commandkey{style2}] (0) at (0, -1) {$#2$};
    \node [style=\commandkey{style3}] (1) at (0, 1) {$#3$};
    \node [style=\commandkey{style1}] (2) at (0, -2.6) {$#1$};
    \node [style=\commandkey{style4}] (3) at (0, 2.6) {$#4$};
  \end{pgfonlayer}
  \begin{pgfonlayer}{edgelayer}
    \draw (2) to (0);
    \draw (0) to (1);
    \draw (1) to (3);
  \end{pgfonlayer}
\end{tikzpicture}\,}

\newkeycommand{\longbraket}[style1=point,style2=copoint][2]{\,%
\begin{tikzpicture}
  \begin{pgfonlayer}{nodelayer}
    \node [style=\commandkey{style1}] (0) at (0, -2.6) {$#1$};
    \node [style=\commandkey{style2}] (1) at (0, 2.6) {$#2$};
  \end{pgfonlayer}
  \begin{pgfonlayer}{edgelayer}
    \draw (0) to (1);
  \end{pgfonlayer}
\end{tikzpicture}\,}

\newkeycommand{\onb}[style=point]{\ensuremath{\{\,\tikz{\node[style=\commandkey{style}] (x) at (0,-0.3) {$j$};\draw (x)--(0,0.7);}\,\}}\xspace}

\newkeycommand{\phase}[style=white dot][1]{\,\begin{tikzpicture}
    \begin{pgfonlayer}{nodelayer}
        \node [style=none] (0) at (0, 1) {};
        \node [style=\commandkey{style}] (2) at (0, -0) {$#1$}; 
        \node [style=none] (3) at (0, -1) {};
    \end{pgfonlayer}
    \begin{pgfonlayer}{edgelayer}
        \draw (2) to (0.center);
        \draw (3.center) to (2);
    \end{pgfonlayer}
\end{tikzpicture}\,}

\newkeycommand{\dphase}[style=white phase ddot][1]{\,\begin{tikzpicture}
    \begin{pgfonlayer}{nodelayer}
        \node [style=none] (0) at (0, 1.25) {};
        \node [style=\commandkey{style}] (2) at (0, -0) {$#1$};
        \node [style=none] (3) at (0, -1.25) {};
    \end{pgfonlayer}
    \begin{pgfonlayer}{edgelayer}
        \draw [boldedge] (2) to (0.center);
        \draw [boldedge] (3.center) to (2);
    \end{pgfonlayer}
\end{tikzpicture}\,}

\newkeycommand{\dphasepoint}[style=white phase ddot][1]{\,\begin{tikzpicture}[yshift=-3mm]
  \begin{pgfonlayer}{nodelayer}
    \node [style=none] (0) at (0, 1) {};
    \node [style=\commandkey{style}] (1) at (0, 0) {$#1$};
  \end{pgfonlayer}
  \begin{pgfonlayer}{edgelayer}
    \draw [style=boldedge] (0.center) to (1);
  \end{pgfonlayer}
\end{tikzpicture}\,}

\newkeycommand{\dphasepointgray}[style=gray phase ddot][1]{\,\begin{tikzpicture}[yshift=-3mm]
  \begin{pgfonlayer}{nodelayer}
    \node [style=none] (0) at (0, 1) {};
    \node [style=\commandkey{style}] (1) at (0, 0) {$#1$};
  \end{pgfonlayer}
  \begin{pgfonlayer}{edgelayer}
    \draw [style=boldedge] (0.center) to (1);
  \end{pgfonlayer}
\end{tikzpicture}\,}

\newkeycommand{\dphasecopoint}[style=white phase ddot][1]{\,\begin{tikzpicture}[yshift=3mm,yscale=-1]
  \begin{pgfonlayer}{nodelayer}
    \node [style=none] (0) at (0, 1) {};
    \node [style=\commandkey{style}] (1) at (0, 0) {$#1$};
  \end{pgfonlayer}
  \begin{pgfonlayer}{edgelayer}
    \draw [style=boldedge] (0.center) to (1);
  \end{pgfonlayer}
\end{tikzpicture}\,}

\newkeycommand{\dphasepointsm}[style=white phase ddot][1]{\,\begin{tikzpicture}[yshift=-3mm]
  \begin{pgfonlayer}{nodelayer}
    \node [style=none] (0) at (0, 1) {};
    \node [style=\commandkey{style}] (1) at (0, 0) {\footnotesize\!$#1$\!};
  \end{pgfonlayer}
  \begin{pgfonlayer}{edgelayer}
    \draw [style=boldedge] (0.center) to (1);
  \end{pgfonlayer}
\end{tikzpicture}\,}

\newkeycommand{\phasepoint}[style=white phase dot][1]{\,\begin{tikzpicture}[yshift=-3mm]
  \begin{pgfonlayer}{nodelayer}
    \node [style=none] (0) at (0, 1) {};
    \node [style=\commandkey{style}] (1) at (0, 0) {$#1$};
  \end{pgfonlayer}
  \begin{pgfonlayer}{edgelayer}
    \draw (0.center) to (1);
  \end{pgfonlayer}
\end{tikzpicture}\,}

\newkeycommand{\kpointbraket}[style1=kpoint,style2=kpoint adjoint,estyle=][2]{\,%
\begin{tikzpicture}
  \begin{pgfonlayer}{nodelayer}
    \node [style=\commandkey{style1}] (0) at (0, -0.735) {$#1$};
    \node [style=\commandkey{style2}] (1) at (0, 0.735) {$#2$};
  \end{pgfonlayer}
  \begin{pgfonlayer}{edgelayer}
    \draw [\commandkey{estyle}] (0) to (1);
  \end{pgfonlayer}
\end{tikzpicture}\,}

\newkeycommand{\kpointbraketx}[style1=kpoint,style2=kpoint adjoint,estyle=][2]{\,%
\begin{tikzpicture}
  \begin{pgfonlayer}{nodelayer}
    \node [style=\commandkey{style1}] (0) at (0, -0.885) {$#1$};
    \node [style=\commandkey{style2}] (1) at (0, 0.885) {$#2$};
  \end{pgfonlayer}
  \begin{pgfonlayer}{edgelayer}
    \draw [\commandkey{estyle}] (0) to (1);
  \end{pgfonlayer}
\end{tikzpicture}\,}

\newcommand{\bra}[1]{\ensuremath{\left\langle #1 \right|}}
\newcommand{\ket}[1]{\ensuremath{\left|  #1 \right\rangle}}

\newcommand{\ketbra}[2]{\ensuremath{\ket{#1}\!\bra{#2}}}

\tikzstyle{cdiag}=[matrix of math nodes, row sep=3em, column sep=3em, text height=1.5ex, text depth=0.25ex,inner sep=0.5em]
\tikzstyle{arrow above}=[transform canvas={yshift=0.5ex}]
\tikzstyle{arrow below}=[transform canvas={yshift=-0.5ex}]

\newcommand{\NWbracket}[1]{%
\draw #1 +(-0.25,0.5) -- +(-0.5,0.5) -- +(-0.5,0.25);}
\newcommand{\NEbracket}[1]{%
\draw #1 +(0.25,0.5) -- +(0.5,0.5) -- +(0.5,0.25);}

\usetikzlibrary{shapes.multipart}
\usetikzlibrary{decorations.markings}

\tikzset{->-/.style={decoration={
  markings,
  mark=at position .5 with {\arrow{>}}},postaction={decorate}}}
\tikzset{-<-/.style={decoration={
  markings,
  mark=at position .5 with {\arrow{<}}},postaction={decorate}}}
  
\tikzstyle{bwSpider}=[
       rectangle split,
       rectangle split parts=2,
       rectangle split part fill={black,white},
 minimum size=3.6 mm, inner sep=-2mm, draw=black,scale=0.5,rounded corners=0.8 mm
       ]
 \tikzstyle{wbSpider}=[
       rectangle split,
       rectangle split parts=2,
       rectangle split part fill={white,black},
 minimum size=3.6 mm, inner sep=-2mm, draw=black,scale=0.5,rounded corners=0.8 mm
       ]
\tikzstyle{oWire}=[line width = 1pt, color=black]
\tikzstyle{QWire}=[line width = 1pt, color=quantumWireviolet]
\tikzstyle{qWire}=[line width = 1pt, color=compositionalitydarkpink]
\tikzstyle{cWire}=[color=darkgray,line width = .75pt]
\tikzstyle{GWire}=[color=green!60!black,line width = .75pt]
\tikzstyle{CqWire}=[line width = 1pt, color=compositionalitydarkpink,->-]
\tikzstyle{CQWire}=[line width = 1pt, color=quantumWireviolet,->-]
\tikzstyle{CcWire}=[color=darkgray,line width = .75pt,->-]
\tikzstyle{RqWire}=[line width = 1pt, color=compositionalitydarkpink,-<-]
\tikzstyle{RQWire}=[line width = 1pt, color=quantumWireviolet,-<-]
\tikzstyle{RcWire}=[color=darkgray,line width = .75pt,-<-]

\tikzstyle{epiCopoint}=[regular polygon,regular polygon sides=3,draw,scale=0.8,inner sep=-0.5pt,minimum width=5mm,fill=white,regular polygon rotate=0,line width=.5pt]
\tikzstyle{epiPoint}=[regular polygon,regular polygon sides=3,draw,scale=0.8,inner sep=-0.5pt,minimum width=5mm,fill=white,regular polygon rotate=180,line width=.5pt]
\tikzstyle{epiPointWide}=[regular polygon,regular polygon sides=3,draw,xscale=0.75,yscale=.55,inner sep=-0.5pt,minimum width=8mm,fill=white,regular polygon rotate=180,line width=.5pt]
\tikzstyle{epiBox}=[fill=white,draw, line width = .5pt,inner sep=0.6mm,font=\footnotesize,minimum height=3mm,minimum width=3mm]
\tikzstyle{epiBoxWide}=[fill=white,draw, line width = .5pt,inner sep=0.6mm,font=\footnotesize,minimum height=3.5mm,minimum width=5mm]
\tikzstyle{epiBoxVeryWide}=[fill=white,draw, line width = .5pt,inner sep=0.6mm,font=\footnotesize,minimum height=3mm,minimum width=7mm]

\tikzstyle{env}=[copoint,regular polygon rotate=0,minimum width=0.2cm, fill=black]

\tikzstyle{probs}=[shape=semicircle,fill=white,draw=black,shape border rotate=180,minimum width=1.2cm]

\tikzstyle{every picture}=[baseline=-0.25em,scale=0.5]
\tikzstyle{dotpic}=[] 
\tikzstyle{diredges}=[every to/.style={diredge}]
\tikzstyle{math matrix}=[matrix of math nodes,left delimiter=(,right delimiter=),inner sep=2pt,column sep=1em,row sep=0.5em,nodes={inner sep=0pt},text height=1.5ex, text depth=0.25ex]

\tikzstyle{inline text}=[text height=1.5ex, text depth=0.25ex,yshift=0.5mm]
\tikzstyle{label}=[font=\footnotesize,text height=1.5ex, text depth=0.25ex,yshift=0.5mm]
\tikzstyle{left label}=[label,anchor=east,xshift=1.5mm]
\tikzstyle{right label}=[label,anchor=west,xshift=-1mm]
\tikzstyle{up label}=[label,anchor=south,yshift=-1mm]

\tikzstyle{rght label}=[label,anchor=west,xshift=-1mm]

\tikzstyle{braceedge}=[decorate,decoration={brace,amplitude=2mm,raise=-1mm}]
\tikzstyle{small braceedge}=[decorate,decoration={brace,amplitude=1mm,raise=-1mm}]

\tikzstyle{doubled}=[line width=1.6pt] 
\tikzstyle{boldedge}=[doubled,shorten <=-0.17mm,shorten >=-0.17mm]
\tikzstyle{boldedgegray}=[doubled,gray,shorten <=-0.17mm,shorten >=-0.17mm]
\tikzstyle{singleedgegray}=[gray]

\tikzstyle{semidoubled}=[line width=1.4pt] 
\tikzstyle{semiboldedgegray}=[semidoubled,gray,shorten <=-0.17mm,shorten >=-0.17mm]

\tikzstyle{boxedge}=[semiboldedgegray]

\tikzstyle{boldedgedashed}=[very thick,dashed,shorten <=-0.17mm,shorten >=-0.17mm]
\tikzstyle{vboldedgedashed}=[doubled,dashed,shorten <=-0.17mm,shorten >=-0.17mm]
\tikzstyle{left hook arrow}=[left hook-latex]
\tikzstyle{right hook arrow}=[right hook-latex]
\tikzstyle{sembracket}=[line width=0.5pt,shorten <=-0.07mm,shorten >=-0.07mm]

\tikzstyle{causal edge}=[->,thick,gray]
\tikzstyle{causal nondir}=[thick,gray]
\tikzstyle{timeline}=[thick,gray, dashed]

\tikzstyle{cedge}=[<->,thick,gray!70!white]

\tikzstyle{empty diagram}=[draw=gray!40!white,dashed,shape=rectangle,minimum width=1cm,minimum height=1cm]
\tikzstyle{empty diagram small}=[draw=gray!50!white,dashed,shape=rectangle,minimum width=0.6cm,minimum height=0.5cm]

\tikzstyle{dot}=[inner sep=0mm,minimum width=2mm,minimum height=2mm,draw,shape=circle]
\tikzstyle{bigdot}=[inner sep=0mm,minimum width=5mm,minimum height=5mm,draw,shape=circle]
\tikzstyle{leak}=[white dot, shape=regular polygon, minimum size=3.3 mm, regular polygon sides=3, outer sep=-0.2mm, regular polygon rotate=270]
\tikzstyle{proj}=[regular polygon,regular polygon sides=4,draw,scale=0.75,inner sep=-0.5pt,minimum width=6mm,fill=white]
\tikzstyle{projOut}=[regular polygon,regular polygon sides=3,draw,scale=0.75,inner sep=-0.5pt,minimum width=7.5mm,fill=white,regular polygon rotate=180]
\tikzstyle{projIn}=[regular polygon,regular polygon sides=3,draw,scale=0.75,inner sep=-0.5pt,minimum width=7.5mm,fill=white]
\tikzstyle{Vleak}=[white dot, shape=regular polygon, minimum size=3.3 mm, regular polygon sides=3, outer sep=-0.2mm, regular polygon rotate=90]
\tikzstyle{dleak}=[white dot, line width=1.6pt, shape=regular polygon, minimum size=3.3 mm, regular polygon sides=3, outer sep=-0.2mm, regular polygon rotate=270]

\tikzstyle{Wsquare}=[white dot, shape=regular polygon, rounded corners=0.8 mm, minimum size=3.3 mm, regular polygon sides=3, outer sep=-0.2mm]
\tikzstyle{Wsquareadj}=[white dot, shape=regular polygon, rounded corners=0.8 mm, minimum size=3.3 mm, regular polygon sides=3, outer sep=-0.2mm, regular polygon rotate=180]
\tikzstyle{ddot}=[inner sep=0mm, doubled, minimum width=2.5mm,minimum height=2.5mm,draw,shape=circle]

\tikzstyle{black dot}=[dot,fill=black]
\tikzstyle{small black dot}=[inner sep=0mm,minimum width=.5mm,minimum height=.5mm,draw,shape=circle,fill=black]
\tikzstyle{white dot}=[dot,fill=white,text depth=-0.2mm]
\tikzstyle{big white dot}=[dot,fill=white,text depth=-0.2mm, minimum width = 2.5mm]
\tikzstyle{white Wsquare}=[Wsquare,fill=gray,,text depth=-0.2mm]
\tikzstyle{white Wsquareadj}=[Wsquareadj,fill=white,,text depth=-0.2mm]
\tikzstyle{green dot}=[white dot] 
\tikzstyle{gray dot}=[dot,fill=gray!40!white,,text depth=-0.2mm]
\tikzstyle{red dot}=[gray dot] 

\tikzstyle{black ddot}=[ddot,fill=black]
\tikzstyle{white ddot}=[ddot,fill=white]
\tikzstyle{gray ddot}=[ddot,fill=gray!40!white]

\tikzstyle{gray edge}=[gray!60!white]

\tikzstyle{small dot}=[inner sep=0.5mm,minimum width=0pt,minimum height=0pt,draw,shape=circle]

\tikzstyle{small white dot}=[small dot,fill=white]
\tikzstyle{small gray dot}=[small dot,fill=gray!40!white]

\tikzstyle{causal dot}=[inner sep=0.4mm,minimum width=0pt,minimum height=0pt,draw=white,shape=circle,fill=gray!40!white]

\tikzstyle{phase dimensions}=[minimum size=5mm,font=\footnotesize,rectangle,rounded corners=2.5mm,inner sep=0.2mm,outer sep=-2mm]
\tikzstyle{dphase dimensions}=[minimum size=5mm,font=\footnotesize,rectangle,rounded corners=2.5mm,inner sep=0.2mm,outer sep=-2mm]

\tikzstyle{white phase dot}=[dot,fill=white,phase dimensions]
\tikzstyle{white phase ddot}=[ddot,fill=white,dphase dimensions]

\tikzstyle{white rect ddot}=[draw=black,fill=white,doubled,minimum size=5mm,font=\footnotesize,rectangle,rounded corners=2.5mm,inner sep=0.2mm]
\tikzstyle{gray rect ddot}=[draw=black,fill=gray!40!white,doubled,minimum size=6mm,font=\footnotesize,rectangle,rounded corners=3mm]

\tikzstyle{gray phase dot}=[dot,fill=gray!40!white,phase dimensions]
\tikzstyle{gray phase ddot}=[ddot,fill=gray!40!white,dphase dimensions]
\tikzstyle{grey phase dot}=[gray phase dot]
\tikzstyle{grey phase ddot}=[gray phase ddot]

\tikzstyle{small phase dimensions}=[minimum size=4mm,font=\tiny,rectangle,rounded corners=2mm,inner sep=0.2mm,outer sep=-2mm]
\tikzstyle{small dphase dimensions}=[minimum size=4mm,font=\tiny,rectangle,rounded corners=2mm,inner sep=0.2mm,outer sep=-2mm]

\tikzstyle{small gray phase dot}=[dot,fill=gray!40!white,small phase dimensions]
\tikzstyle{small gray phase ddot}=[ddot,fill=gray!40!white,small dphase dimensions]

\tikzstyle{small map}=[draw,shape=rectangle,minimum height=4mm,minimum width=4mm,fill=white]

\tikzstyle{cnot}=[fill=white,shape=circle,inner sep=-1.4pt]

\tikzstyle{asym hadamard}=[fill=white,draw,shape=NEbox,inner sep=0.6mm,font=\footnotesize,minimum height=4mm]
\tikzstyle{asym hadamard conj}=[fill=white,draw,shape=NWbox,inner sep=0.6mm,font=\footnotesize,minimum height=4mm]
\tikzstyle{asym hadamard dag}=[fill=white,draw,shape=SEbox,inner sep=0.6mm,font=\footnotesize,minimum height=4mm]

\tikzstyle{hadamard}=[fill=white,draw,inner sep=0.6mm,font=\footnotesize,minimum height=4mm,minimum width=4mm]
\tikzstyle{black square}=[fill=black,draw,inner sep=0.6mm,font=\footnotesize,minimum height=2mm,minimum width=2mm]
\tikzstyle{small hadamard}=[fill=white,draw,inner sep=0.6mm,minimum height=1.5mm,minimum width=1.5mm]
\tikzstyle{small hadamard rotate}=[small hadamard,rotate=45]
\tikzstyle{dhadamard}=[hadamard,doubled]
\tikzstyle{small dhadamard}=[small hadamard,doubled]
\tikzstyle{small dhadamard rotate}=[small hadamard rotate,doubled]
\tikzstyle{antipode}=[white dot,inner sep=0.3mm,font=\footnotesize]

\tikzstyle{scalar}=[diamond,draw,inner sep=0.5pt,font=\small]
\tikzstyle{dscalar}=[diamond,doubled, draw,inner sep=0.5pt,font=\small]

\tikzstyle{small box}=[rectangle,inline text,fill=white,draw,minimum height=5mm,yshift=-0.5mm,minimum width=5mm,font=\small]
\tikzstyle{small gray box}=[small box,fill=gray!30]
\tikzstyle{medium box}=[rectangle,inline text,fill=white,draw,minimum height=5mm,yshift=-0.5mm,minimum width=10mm,font=\small]
\tikzstyle{square box}=[small box] 
\tikzstyle{medium gray box}=[small box,fill=gray!30]
\tikzstyle{semilarge box}=[rectangle,inline text,fill=white,draw,minimum height=5mm,yshift=-0.5mm,minimum width=12.5mm,font=\small]
\tikzstyle{large box}=[rectangle,inline text,fill=white,draw,minimum height=5mm,yshift=-0.5mm,minimum width=15mm,font=\small]
\tikzstyle{large gray box}=[small box,fill=gray!30]

\tikzstyle{Bayes box}=[rectangle,fill=black,draw, minimum height=3mm, minimum width=3mm]

\tikzstyle{gray square point}=[small box,fill=gray!50]

\tikzstyle{dphase box white}=[dhadamard]
\tikzstyle{dphase box gray}=[dhadamard,fill=gray!50!white]
\tikzstyle{phase box white}=[hadamard]
\tikzstyle{phase box gray}=[hadamard,fill=gray!50!white]

\tikzstyle{point nosep}=[regular polygon,regular polygon sides=3,draw,scale=0.75,inner sep=-2pt,minimum width=9mm,fill=white,regular polygon rotate=180]

\tikzstyle{point}=[regular polygon,regular polygon sides=3,draw,scale=0.75,inner sep=-0.5pt,minimum width=9mm,fill=white,regular polygon rotate=180]

\tikzstyle{copoint}=[regular polygon,regular polygon sides=3,draw,scale=0.75,inner sep=-0.5pt,minimum width=9mm,fill=white]

\tikzstyle{dpoint}=[point,doubled]
\tikzstyle{dcopoint}=[copoint,doubled]

\tikzstyle{pointgrow}=[shape=cornerpoint,kpoint common,scale=0.75,inner sep=3pt]
\tikzstyle{pointgrow dag}=[shape=cornercopoint,kpoint common,scale=0.75,inner sep=3pt]

\tikzstyle{wide copoint}=[fill=white,draw,shape=isosceles triangle,shape border rotate=90,isosceles triangle stretches=true,inner sep=0pt,minimum width=1.5cm,minimum height=6.12mm]
\tikzstyle{wide point}=[fill=white,draw,shape=isosceles triangle,shape border rotate=-90,isosceles triangle stretches=true,inner sep=0pt,minimum width=1.5cm,minimum height=6.12mm,yshift=-0.0mm]
\tikzstyle{wide point plus}=[fill=white,draw,shape=isosceles triangle,shape border rotate=-90,isosceles triangle stretches=true,inner sep=0pt,minimum width=1.74cm,minimum height=7mm,yshift=-0.0mm]

\tikzstyle{wide dpoint}=[fill=white,doubled,draw,shape=isosceles triangle,shape border rotate=-90,isosceles triangle stretches=true,inner sep=0pt,minimum width=1.5cm,minimum height=6.12mm,yshift=-0.0mm]

\tikzstyle{tinypoint}=[regular polygon,regular polygon sides=3,draw,scale=0.55,inner sep=-0.15pt,minimum width=6mm,fill=white,regular polygon rotate=180]

\tikzstyle{white point}=[point]
\tikzstyle{white dpoint}=[dpoint]
\tikzstyle{green point}=[white point] 
\tikzstyle{white copoint}=[copoint]
\tikzstyle{gray point}=[point,fill=gray!40!white]
\tikzstyle{gray dpoint}=[gray point,doubled]
\tikzstyle{red point}=[gray point] 
\tikzstyle{gray copoint}=[copoint,fill=gray!40!white]
\tikzstyle{gray dcopoint}=[gray copoint,doubled]

\tikzstyle{white point guide}=[regular polygon,regular polygon sides=3,font=\scriptsize,draw,scale=0.65,inner sep=-0.5pt,minimum width=9mm,fill=white,regular polygon rotate=180]

\tikzstyle{black point}=[point,fill=black,font=\color{white}]
\tikzstyle{black copoint}=[copoint,fill=black,font=\color{white}]

\tikzstyle{tiny gray point}=[tinypoint,fill=gray!40!white]

\tikzstyle{diredge}=[->]
\tikzstyle{ddiredge}=[<->]
\tikzstyle{rdiredge}=[<-]
\tikzstyle{thickdiredge}=[->, very thick]
\tikzstyle{pointer edge}=[->,very thick,gray]
\tikzstyle{pointer edge part}=[very thick,gray]
\tikzstyle{dashed edge}=[dashed]
\tikzstyle{thick dashed edge}=[very thick,dashed]
\tikzstyle{thick gray dashed edge}=[thick dashed edge,gray!40]
\tikzstyle{thick map edge}=[very thick,|->]

\makeatletter
\newcommand{\boxshape}[3]{%
\pgfdeclareshape{#1}{
\inheritsavedanchors[from=rectangle] 
\inheritanchorborder[from=rectangle]
\inheritanchor[from=rectangle]{center}
\inheritanchor[from=rectangle]{north}
\inheritanchor[from=rectangle]{south}
\inheritanchor[from=rectangle]{west}
\inheritanchor[from=rectangle]{east}
\backgroundpath{
\southwest \pgf@xa=\pgf@x \pgf@ya=\pgf@y
\northeast \pgf@xb=\pgf@x \pgf@yb=\pgf@y

\@tempdima=#2
\@tempdimb=#3

\pgfpathmoveto{\pgfpoint{\pgf@xa - 5pt + \@tempdima}{\pgf@ya}}
\pgfpathlineto{\pgfpoint{\pgf@xa - 5pt - \@tempdima}{\pgf@yb}}
\pgfpathlineto{\pgfpoint{\pgf@xb + 5pt + \@tempdimb}{\pgf@yb}}
\pgfpathlineto{\pgfpoint{\pgf@xb + 5pt - \@tempdimb}{\pgf@ya}}
\pgfpathlineto{\pgfpoint{\pgf@xa - 5pt + \@tempdima}{\pgf@ya}}
\pgfpathclose
}
}}

\boxshape{NEbox}{0pt}{3pt}
\boxshape{SEbox}{0pt}{-3pt}
\boxshape{NWbox}{5pt}{0pt}
\boxshape{SWbox}{-5pt}{0pt}
\boxshape{EBox}{-3pt}{3pt}
\boxshape{WBox}{3pt}{-3pt}
\makeatother

\tikzstyle{cloud}=[shape=cloud,draw,minimum width=1.5cm,minimum height=1.5cm]

\tikzstyle{map}=[draw,shape=NEbox,inner sep=1pt,minimum height=4mm,fill=white]
\tikzstyle{dashedmap}=[draw,dashed,shape=NEbox,inner sep=2pt,minimum height=6mm,fill=white]
\tikzstyle{mapdag}=[draw,shape=SEbox,inner sep=1pt,minimum height=4mm,fill=white]
\tikzstyle{mapadj}=[draw,shape=SEbox,inner sep=2pt,minimum height=6mm,fill=white]
\tikzstyle{maptrans}=[draw,shape=SWbox,inner sep=2pt,minimum height=6mm,fill=white]
\tikzstyle{mapconj}=[draw,shape=NWbox,inner sep=2pt,minimum height=6mm,fill=white]

\tikzstyle{medium map}=[draw,shape=NEbox,inner sep=2pt,minimum height=6mm,fill=white,minimum width=7mm]
\tikzstyle{medium map dag}=[draw,shape=SEbox,inner sep=2pt,minimum height=6mm,fill=white,minimum width=7mm]
\tikzstyle{medium map adj}=[draw,shape=SEbox,inner sep=2pt,minimum height=6mm,fill=white,minimum width=7mm]
\tikzstyle{medium map trans}=[draw,shape=SWbox,inner sep=2pt,minimum height=6mm,fill=white,minimum width=7mm]
\tikzstyle{medium map conj}=[draw,shape=NWbox,inner sep=2pt,minimum height=6mm,fill=white,minimum width=7mm]
\tikzstyle{semilarge map}=[draw,shape=NEbox,inner sep=2pt,minimum height=6mm,fill=white,minimum width=9.5mm]
\tikzstyle{semilarge map trans}=[draw,shape=SWbox,inner sep=2pt,minimum height=6mm,fill=white,minimum width=9.5mm]
\tikzstyle{semilarge map adj}=[draw,shape=SEbox,inner sep=2pt,minimum height=6mm,fill=white,minimum width=9.5mm]
\tikzstyle{semilarge map dag}=[draw,shape=SEbox,inner sep=2pt,minimum height=6mm,fill=white,minimum width=9.5mm]
\tikzstyle{semilarge map conj}=[draw,shape=NWbox,inner sep=2pt,minimum height=6mm,fill=white,minimum width=9.5mm]
\tikzstyle{large map}=[draw,shape=NEbox,inner sep=2pt,minimum height=6mm,fill=white,minimum width=12mm]
\tikzstyle{large map conj}=[draw,shape=NWbox,inner sep=2pt,minimum height=6mm,fill=white,minimum width=12mm]
\tikzstyle{very large map}=[draw,shape=NEbox,inner sep=2pt,minimum height=6mm,fill=white,minimum width=17mm]

\tikzstyle{medium dmap}=[draw,doubled,shape=NEbox,inner sep=2pt,minimum height=6mm,fill=white,minimum width=7mm]
\tikzstyle{medium dmap dag}=[draw,doubled,shape=SEbox,inner sep=2pt,minimum height=6mm,fill=white,minimum width=7mm]
\tikzstyle{medium dmap adj}=[draw,doubled,shape=SEbox,inner sep=2pt,minimum height=6mm,fill=white,minimum width=7mm]
\tikzstyle{medium dmap trans}=[draw,doubled,shape=SWbox,inner sep=2pt,minimum height=6mm,fill=white,minimum width=7mm]
\tikzstyle{medium dmap conj}=[draw,doubled,shape=NWbox,inner sep=2pt,minimum height=6mm,fill=white,minimum width=7mm]
\tikzstyle{semilarge dmap}=[draw,doubled,shape=NEbox,inner sep=2pt,minimum height=6mm,fill=white,minimum width=9.5mm]
\tikzstyle{semilarge dmap trans}=[draw,doubled,shape=SWbox,inner sep=2pt,minimum height=6mm,fill=white,minimum width=9.5mm]
\tikzstyle{semilarge dmap adj}=[draw,doubled,shape=SEbox,inner sep=2pt,minimum height=6mm,fill=white,minimum width=9.5mm]
\tikzstyle{semilarge dmap dag}=[draw,doubled,shape=SEbox,inner sep=2pt,minimum height=6mm,fill=white,minimum width=9.5mm]
\tikzstyle{semilarge dmap conj}=[draw,doubled,shape=NWbox,inner sep=2pt,minimum height=6mm,fill=white,minimum width=9.5mm]
\tikzstyle{large dmap}=[draw,doubled,shape=NEbox,inner sep=2pt,minimum height=6mm,fill=white,minimum width=12mm]
\tikzstyle{large dmap conj}=[draw,doubled,shape=NWbox,inner sep=2pt,minimum height=6mm,fill=white,minimum width=12mm]
\tikzstyle{large dmap trans}=[draw,doubled,shape=SWbox,inner sep=2pt,minimum height=6mm,fill=white,minimum width=12mm]
\tikzstyle{large dmap adj}=[draw,doubled,shape=SEbox,inner sep=2pt,minimum height=6mm,fill=white,minimum width=12mm]
\tikzstyle{large dmap dag}=[draw,doubled,shape=SEbox,inner sep=2pt,minimum height=6mm,fill=white,minimum width=12mm]
\tikzstyle{very large dmap}=[draw,doubled,shape=NEbox,inner sep=2pt,minimum height=6mm,fill=white,minimum width=19.5mm]

\tikzstyle{muxbox}=[draw,shape=rectangle,minimum height=3mm,minimum width=3mm,fill=white]
\tikzstyle{dmuxbox}=[muxbox,doubled]

\tikzstyle{box}=[draw,shape=rectangle,inner sep=2pt,minimum height=6mm,minimum width=6mm,fill=white]
\tikzstyle{dbox}=[draw,doubled,shape=rectangle,inner sep=2pt,minimum height=6mm,minimum width=6mm,fill=white]
\tikzstyle{dmap}=[draw,doubled,shape=NEbox,inner sep=2pt,minimum height=6mm,fill=white]
\tikzstyle{dmapdag}=[draw,doubled,shape=SEbox,inner sep=2pt,minimum height=6mm,fill=white]
\tikzstyle{dmapadj}=[draw,doubled,shape=SEbox,inner sep=2pt,minimum height=6mm,fill=white]
\tikzstyle{dmaptrans}=[draw,doubled,shape=SWbox,inner sep=2pt,minimum height=6mm,fill=white]
\tikzstyle{dmapconj}=[draw,doubled,shape=NWbox,inner sep=2pt,minimum height=6mm,fill=white]

\tikzstyle{ddmap}=[draw,doubled,dashed,shape=NEbox,inner sep=2pt,minimum height=6mm,fill=white]
\tikzstyle{ddmapdag}=[draw,doubled,dashed,shape=SEbox,inner sep=2pt,minimum height=6mm,fill=white]
\tikzstyle{ddmapadj}=[draw,doubled,dashed,shape=SEbox,inner sep=2pt,minimum height=6mm,fill=white]
\tikzstyle{ddmaptrans}=[draw,doubled,dashed,shape=SWbox,inner sep=2pt,minimum height=6mm,fill=white]
\tikzstyle{ddmapconj}=[draw,doubled,dashed,shape=NWbox,inner sep=2pt,minimum height=6mm,fill=white]

\boxshape{sNEbox}{0pt}{3pt}
\boxshape{sSEbox}{0pt}{-3pt}
\boxshape{sNWbox}{3pt}{0pt}
\boxshape{sSWbox}{-3pt}{0pt}
\tikzstyle{smap}=[draw,shape=sNEbox,fill=white]
\tikzstyle{smapdag}=[draw,shape=sSEbox,fill=white]
\tikzstyle{smapadj}=[draw,shape=sSEbox,fill=white]
\tikzstyle{smaptrans}=[draw,shape=sSWbox,fill=white]
\tikzstyle{smapconj}=[draw,shape=sNWbox,fill=white]

\tikzstyle{dsmap}=[draw,dashed,shape=sNEbox,fill=white]
\tikzstyle{dsmapdag}=[draw,dashed,shape=sSEbox,fill=white]
\tikzstyle{dsmaptrans}=[draw,dashed,shape=sSWbox,fill=white]
\tikzstyle{dsmapconj}=[draw,dashed,shape=sNWbox,fill=white]

\boxshape{mNEbox}{0pt}{10pt}
\boxshape{mSEbox}{0pt}{-10pt}
\boxshape{mNWbox}{10pt}{0pt}
\boxshape{mSWbox}{-10pt}{0pt}
\tikzstyle{mmap}=[draw,shape=mNEbox]
\tikzstyle{mmapdag}=[draw,shape=mSEbox]
\tikzstyle{mmaptrans}=[draw,shape=mSWbox]
\tikzstyle{mmapconj}=[draw,shape=mNWbox]

\tikzstyle{mmapgray}=[draw,fill=gray!40!white,shape=mNEbox]
\tikzstyle{smapgray}=[draw,fill=gray!40!white,shape=sNEbox]

\makeatletter

\pgfdeclareshape{cornerpoint}{
\inheritsavedanchors[from=rectangle] 
\inheritanchorborder[from=rectangle]
\inheritanchor[from=rectangle]{center}
\inheritanchor[from=rectangle]{north}
\inheritanchor[from=rectangle]{south}
\inheritanchor[from=rectangle]{west}
\inheritanchor[from=rectangle]{east}
\backgroundpath{
\southwest \pgf@xa=\pgf@x \pgf@ya=\pgf@y
\northeast \pgf@xb=\pgf@x \pgf@yb=\pgf@y

\pgfmathsetmacro{\pgf@shorten@left}{\pgfkeysvalueof{/tikz/shorten left}}
\pgfmathsetmacro{\pgf@shorten@right}{\pgfkeysvalueof{/tikz/shorten right}}

\pgfpathmoveto{\pgfpoint{0.5 * (\pgf@xa + \pgf@xb)}{\pgf@ya - 5pt}}
\pgfpathlineto{\pgfpoint{\pgf@xa - 8pt + \pgf@shorten@left}{\pgf@yb - 1.5 * \pgf@shorten@left}}
\pgfpathlineto{\pgfpoint{\pgf@xa - 8pt + \pgf@shorten@left}{\pgf@yb}}
\pgfpathlineto{\pgfpoint{\pgf@xb + 8pt - \pgf@shorten@right}{\pgf@yb}}
\pgfpathlineto{\pgfpoint{\pgf@xb + 8pt - \pgf@shorten@right}{\pgf@yb - 1.5 * \pgf@shorten@right}}
\pgfpathclose
}
}

\pgfdeclareshape{cornercopoint}{
\inheritsavedanchors[from=rectangle] 
\inheritanchorborder[from=rectangle]
\inheritanchor[from=rectangle]{center}
\inheritanchor[from=rectangle]{north}
\inheritanchor[from=rectangle]{south}
\inheritanchor[from=rectangle]{west}
\inheritanchor[from=rectangle]{east}
\backgroundpath{
\southwest \pgf@xa=\pgf@x \pgf@ya=\pgf@y
\northeast \pgf@xb=\pgf@x \pgf@yb=\pgf@y

\pgfmathsetmacro{\pgf@shorten@left}{\pgfkeysvalueof{/tikz/shorten left}}
\pgfmathsetmacro{\pgf@shorten@right}{\pgfkeysvalueof{/tikz/shorten right}}

\pgfpathmoveto{\pgfpoint{0.5 * (\pgf@xa + \pgf@xb)}{\pgf@yb + 5pt}}
\pgfpathlineto{\pgfpoint{\pgf@xa - 8pt + \pgf@shorten@left}{\pgf@ya + 1.5 * \pgf@shorten@left}}
\pgfpathlineto{\pgfpoint{\pgf@xa - 8pt + \pgf@shorten@left}{\pgf@ya}}
\pgfpathlineto{\pgfpoint{\pgf@xb + 8pt - \pgf@shorten@right}{\pgf@ya}}
\pgfpathlineto{\pgfpoint{\pgf@xb + 8pt - \pgf@shorten@right}{\pgf@ya + 1.5 * \pgf@shorten@right}}
\pgfpathclose
}
}

\makeatother

\pgfkeyssetvalue{/tikz/shorten left}{0pt}
\pgfkeyssetvalue{/tikz/shorten right}{0pt}

\tikzstyle{kpoint common}=[draw,fill=white,inner sep=1pt,minimum height=4mm]
\tikzstyle{kpoint sc}=[shape=cornerpoint,kpoint common]
\tikzstyle{kpoint adjoint sc}=[shape=cornercopoint,kpoint common]
\tikzstyle{kpoint}=[shape=cornerpoint,shorten left=5pt,kpoint common]
\tikzstyle{kpoint adjoint}=[shape=cornercopoint,shorten left=5pt,kpoint common]
\tikzstyle{kpoint conjugate}=[shape=cornerpoint,shorten right=5pt,kpoint common]
\tikzstyle{kpoint transpose}=[shape=cornercopoint,shorten right=5pt,kpoint common]
\tikzstyle{kpoint symm}=[shape=cornerpoint,shorten left=5pt,shorten right=5pt,kpoint common]

\tikzstyle{wide kpoint sc}=[shape=cornerpoint,kpoint common, minimum width=1 cm]
\tikzstyle{wide kpointdag sc}=[shape=cornercopoint,kpoint common, minimum width=1 cm]

\tikzstyle{black kpoint}=[shape=cornerpoint,shorten left=5pt,kpoint common,fill=black,font=\color{white}]

\tikzstyle{black kpoint sm}=[shape=cornerpoint,shorten left=5pt,kpoint common,fill=black,font=\color{white},scale=0.75]

\tikzstyle{black kpoint adjoint}=[shape=cornercopoint,shorten left=5pt,kpoint common,fill=black,font=\color{white}]
\tikzstyle{black kpointadj}=[shape=cornercopoint,shorten left=5pt,kpoint common,fill=black,font=\color{white}]

\tikzstyle{black kpointadj sm}=[shape=cornercopoint,shorten left=5pt,kpoint common,fill=black,font=\color{white},scale=0.75]

\tikzstyle{black dkpoint}=[shape=cornerpoint,shorten left=5pt,kpoint common,fill=black, doubled,font=\color{white}]
\tikzstyle{black dkpoint adjoint}=[shape=cornercopoint,shorten left=5pt,kpoint common,fill=black, doubled,font=\color{white}]
\tikzstyle{black dkpointadj}=[shape=cornercopoint,shorten left=5pt,kpoint common,fill=black, doubled,font=\color{white}]

\tikzstyle{black dkpoint sm}=[shape=cornerpoint,shorten left=5pt,kpoint common,fill=black, doubled,font=\color{white},scale=0.75]
\tikzstyle{black dkpointadj sm}=[shape=cornercopoint,shorten left=5pt,kpoint common,fill=black, doubled,font=\color{white},scale=0.75]

\tikzstyle{kpointdag}=[kpoint adjoint]
\tikzstyle{kpointadj}=[kpoint adjoint]
\tikzstyle{kpointconj}=[kpoint conjugate]
\tikzstyle{kpointtrans}=[kpoint transpose]

\tikzstyle{big kpoint}=[kpoint, minimum width=1.2 cm, minimum height=8mm, inner sep=4pt, text depth=3mm]

\tikzstyle{wide kpoint}=[kpoint, minimum width=1 cm, inner sep=2pt]
\tikzstyle{wide kpointdag}=[kpointdag, minimum width=1 cm, inner sep=2pt]
\tikzstyle{wide kpointconj}=[kpointconj, minimum width=1 cm, inner sep=2pt]
\tikzstyle{wide kpointtrans}=[kpointtrans, minimum width=1 cm, inner sep=2pt]

\tikzstyle{wider kpoint}=[kpoint, minimum width=1.25 cm, inner sep=2pt]
\tikzstyle{wider kpointdag}=[kpointdag, minimum width=1.25 cm, inner sep=2pt]
\tikzstyle{wider kpointconj}=[kpointconj, minimum width=1.25 cm, inner sep=2pt]
\tikzstyle{wider kpointtrans}=[kpointtrans, minimum width=1.25 cm, inner sep=2pt]

\tikzstyle{gray kpoint}=[kpoint,fill=gray!50!white]
\tikzstyle{gray kpointdag}=[kpointdag,fill=gray!50!white]
\tikzstyle{gray kpointadj}=[kpointadj,fill=gray!50!white]
\tikzstyle{gray kpointconj}=[kpointconj,fill=gray!50!white]
\tikzstyle{gray kpointtrans}=[kpointtrans,fill=gray!50!white]

\tikzstyle{gray dkpoint}=[kpoint,fill=gray!50!white,doubled]
\tikzstyle{gray dkpointdag}=[kpointdag,fill=gray!50!white,doubled]
\tikzstyle{gray dkpointadj}=[kpointadj,fill=gray!50!white,doubled]
\tikzstyle{gray dkpointconj}=[kpointconj,fill=gray!50!white,doubled]
\tikzstyle{gray dkpointtrans}=[kpointtrans,fill=gray!50!white,doubled]

\tikzstyle{white label}=[draw,fill=white,rectangle,inner sep=0.7 mm]
\tikzstyle{gray label}=[draw,fill=gray!50!white,rectangle,inner sep=0.7 mm]
\tikzstyle{black label}=[draw,fill=black,rectangle,inner sep=0.7 mm]

\tikzstyle{dkpoint}=[kpoint,doubled]
\tikzstyle{wide dkpoint}=[wide kpoint,doubled]
\tikzstyle{dkpointdag}=[kpoint adjoint,doubled]
\tikzstyle{wide dkpointdag}=[wide kpointdag,doubled]
\tikzstyle{dkcopoint}=[kpoint adjoint,doubled]
\tikzstyle{dkpointadj}=[kpoint adjoint,doubled]
\tikzstyle{dkpointconj}=[kpoint conjugate,doubled]
\tikzstyle{dkpointtrans}=[kpoint transpose,doubled]

\tikzstyle{kscalar}=[kpoint common, shape=EBox, inner xsep=-1pt, inner ysep=3pt,font=\small]
\tikzstyle{kscalarconj}=[kpoint common, shape=WBox, inner xsep=-1pt, inner ysep=3pt,font=\small]

\tikzstyle{spekpoint}=[kpoint sc,minimum height=5mm,inner sep=3pt]
\tikzstyle{spekcopoint}=[kpoint adjoint sc,minimum height=5mm,inner sep=3pt]

\tikzstyle{dspekpoint}=[spekpoint,doubled]
\tikzstyle{dspekcopoint}=[spekcopoint,doubled]

 \tikzstyle{upground}=[circuit ee IEC,thick,ground,rotate=90,scale=2.5]
 \tikzstyle{downground}=[circuit ee IEC,thick,ground,rotate=-90,scale=2.5]
 \tikzstyle{bigground}=[regular polygon,regular polygon sides=3,draw=gray,scale=0.50,inner sep=-0.5pt,minimum width=10mm,fill=gray]

\tikzstyle{arrs}=[-latex,font=\small,auto]
\tikzstyle{arrow plain}=[arrs]
\tikzstyle{arrow dashed}=[dashed,arrs]
\tikzstyle{arrow bold}=[very thick,arrs]
\tikzstyle{arrow hide}=[draw=white!0,-]
\tikzstyle{arrow reverse}=[latex-]
\tikzstyle{cdnode}=[]

\let\olddagger\dagger
\renewcommand{\dagger}{\ensuremath{\olddagger}\xspace}

\usepackage{makeidx}
\makeindex

\theoremstyle{definition}
\newtheorem{theorem}{Theorem}[section]

\newtheorem{proposition}[theorem]{Proposition}
\newtheorem{definition}[theorem]{Definition}
\newtheorem{remark}[theorem]{Remark}
\newtheorem{problem}[]{Open Problem}

\newcommand{\TODO}[1]{\marginpar{\scriptsize\bB \textbf{TODO:} #1\e}}

\newcommand{\TODOa}[1]{\marginpar{\scriptsize\bM \textbf{TODO:} #1\e}}
\newcommand{\TODOb}[1]{\marginpar{\scriptsize\bB \textbf{TODO:} #1\e}}

\newcommand{\COMMa}[1]{\marginpar{\scriptsize\bM \textbf{COMM:} #1\e}}
\newcommand{\COMMb}[1]{\marginpar{\scriptsize\bB \textbf{COMM:} #1\e}}

\newcommand{\CHECK}[1]{\marginpar{\scriptsize\bR \textbf{CHECK:} #1\e}}

\usepackage[color]{changebar}

\usepackage{color}
\def\bR{\begin{color}{red}}
\def\bB{\begin{color}{blue}}
\def\bM{\begin{color}{magenta}}
\def\bC{\begin{color}{cyan}}
\def\bW{\begin{color}{white}}
\def\bBl{\begin{color}{black}}
\def\bG{\begin{color}{green}}
\def\bY{\begin{color}{yellow}}
\def\e{\end{color}\xspace}
\newcommand{\bit}{\begin{itemize}}
\newcommand{\eit}{\end{itemize}\par\noindent}
\newcommand{\ben}{\begin{enumerate}}
\newcommand{\een}{\end{enumerate}\par\noindent}
\newcommand{\beq}{\begin{equation}}
\newcommand{\eeq}{\end{equation}\par\noindent}
\newcommand{\beqa}{\begin{eqnarray*}}
\newcommand{\eeqa}{\end{eqnarray*}\par\noindent}
\newcommand{\beqn}{\begin{eqnarray}}
\newcommand{\eeqn}{\end{eqnarray}\par\noindent}

\if\lecturenotes1

\renewcommand{\TODO}[1]{}
\renewcommand{\TODOa}[1]{}
\renewcommand{\TODOb}[1]{}
\renewcommand{\COMMa}[1]{}
\renewcommand{\COMMb}[1]{}
\renewcommand{\CHECK}[1]{}

\def\bR{\begin{color}{black}}
\def\bB{\begin{color}{black}}
\def\bM{\begin{color}{black}}
\def\bC{\begin{color}{black}}
\def\bW{\begin{color}{black}}
\def\bG{\begin{color}{black}}
\def\bY{\begin{color}{black}}

\fi

\usepackage[numbers,sort&compress]{natbib}
\usepackage[english]{babel}
\usepackage{amsmath}
\usepackage{hyperref}
\usepackage{tikz}
\usepackage{dsfont}
\usepackage{subcaption}
\usepackage{xcolor}

\newcommand{\CPMTN}{\QCalc/\!\!\sim} 
\newcommand{\QPhys}{\mathbf{QPhys}}
\newcommand{\QCalc}{\mathbf{QCalc}}
\newcommand{\QRep}{\mathbf{QRep}_G}
\newcommand{\QPart}{\mathbf{QPart}_G}
\newcommand{\QUnital}{\mathbf{QPhys}|_{unital}}
\newcommand{\QSym}{\mathbf{QNeut}}

\definecolor{compositionalitypink}{RGB}{255,90,95}
	\definecolor{compositionalitydarkpink}{RGB}{255,70,76}
	\definecolor{compositionalitygray}{RGB}{85, 85, 85}  
\definecolor{quantumWireviolet}{RGB}{83, 37, 127} 

\begin{document}

\title{Time symmetry in quantum theories and beyond}
\date{\today}

\author{John H. Selby}
\affiliation{International Centre for Theory of Quantum Technologies, University of Gda\'nsk, 80-308 Gda\'nsk, Poland}
\email{john.h.selby@gmail.com}
\author{Maria E. Stasinou}
\affiliation{Research Center `ATHENA', Patras, Greece}
\affiliation{Quantum Neural Technologies S.A.}
\email{stasinoumar@gmail.com}
\author{Stefano Gogioso}
\affiliation{Spooqy Ltd, London, UK}
\affiliation{University of Oxford, Oxford, UK}
\author{Bob Coecke}
\affiliation{Relational Intelligence Ltd, Oxford, UK}
\affiliation{Wolfson College, University of Oxford, Oxford, UK}
\affiliation{Perimeter Institute, Waterloo, Ontario, Canada}
\email{bob@rel-int.ai}
\maketitle

 \begin{abstract}
There is a stark tension among different formulations of quantum theory in that some are fundamentally time symmetric and others are radically time asymmetric. This tension is crisply captured when thinking of physical theories as theories of processes. We review process theories and their diagrammatic representation, and show how quantum theory can be described in this language. The tension between time symmetry and time asymmetry is then captured by the tension between two of the key structures that are used in this framework. On the one hand, the symmetry is captured by a \emph{dagger structure}, which is represented by a reflection of diagrams. On the other hand, the asymmetry is captured by a condition involving \emph{discarding}, which is ultimately responsible for the theory being compatible with relativistic causality.   Other formulations of quantum theory go a step further than being time symmetric and are ``time-neutral'' -- that is, they contain no notion of time at all. Such time neutrality can be captured by a third key process-theoretic structure -- \emph{cups and caps} -- which ensures that there is no essential distinction between inputs and outputs to a process. As with the dagger, the tension between cups and caps and discarding provides a clean process-theoretic characterisation of the tension between time neutrality and time asymmetry.

Next we consider three different ways in which this tension can be resolved. The first of these is closely related to recent work by Lucien Hardy, in which the tension is resolved by adding a time-reversed version of discarding together with a suitable consistency condition.  The second, to our knowledge, is a new approach. Here the tension is resolved by adding new systems that propagate backwards in time, and imposing a consistency condition to avoid running into well-known time-travel paradoxes. 
The final approach that we explore is closely related to work by Oreshkov and Cerf, in which the tension is resolved by removing the constraint associated with discarding. We show two equivalent ways in which this can be done whilst ensuring that the resulting theory still makes sensible operational predictions.
\end{abstract}

\newpage

\tableofcontents

\newpage

\allowdisplaybreaks

\section{Introduction}

It is a commonly held belief that quantum theory, or indeed nature itself, is time symmetric. This claim, however, requires qualification -- whilst the unitary dynamics of the theory may be time symmetric \cite{wigner1931gruppentheorie,luders1954equivalence,bell1955time}, the theory as a whole certainly is not \cite{aharonov1964time,watanabe1955symmetry,holster2003criterion}.  In fact, from this broader perspective, we find that time symmetry is broken in a rather dramatic way. This is exemplified by the fact that while there are many pure states in the theory -- corresponding to rays in a Hilbert space -- there are no pure effects that can be realised without invoking post-selection, i.e.,~conditioning on the outcomes when performing measurements. Beyond the pure theory, we find that there is a single effect -- discarding (i.e., the (partial) trace) -- that can be implemented without post-selection whilst every density matrix corresponds to a preparable state. In process-theoretic terms, the root of this asymmetry is known as causality \cite{Chiri1} or terminality \cite{Cnonsig}, and  as demonstrated in  \cite{coecke2017time}, the time reverse of quantum theory is a remarkably different theory: it has only a single state and describes a theory of eternal noise.

There have been several attempts in the literature to reformulate these time-asymmetric theories in an explicitly time-symmetric way, creating a stark tension with the aforementioned process-theoretic causality. There are many different motivations for doing so. Firstly, for philosophical reasons we may believe that nature should not have a preferred direction of time at a fundamental level, and hence quantum theory as a fundamental description of nature should be time symmetric. If this is the case, then the asymmetry that we observe in our experiments must be an emergent phenomenon, perhaps arising due to particular choices of boundary conditions for the universe \cite{oreshkov2015operational}. Secondly, in our attempts to reconcile 
 quantum theory and general relativity, there are hints that modifying the role of time in quantum theory will be essential. In particular, it may be the case that we cannot have a predetermined causal structure and so the `past' and `future' could be inextricably mixed \cite{HardyCausaloid}. If that is the case, formulating quantum theory in a time-neutral way  (a stronger condition than time symmetry as we make precise in section~\ref{sec:PT})  may be essential to making progress on the unification of quantum theory and general relativity. 
Thirdly, if we take the view that quantum theory is a generalisation of Bayesian probability theory, that is, that quantum states are epistemic rather than ontic, then it is natural to want to treat space and time on an equal footing \cite{leifer2013towards}. This requires that we find a time neutral formulation of the theory and that the apparent time asymmetry arises from how we choose to apply the theory to physical scenarios rather than being fundamental.

Without necessarily endorsing any of these motivations, we consider the possible approaches to describing quantum theory in a time-symmetric or time-neutral way. We see that many existing works in this direction \cite{oeckl2008general,
oeckl2016local,
oreshkov2016operational,
oreshkov2015operational,
aharonov2009multiple,
aharonov2008two} have a particularly concise process-theoretic description. Moreover, they all correspond to essentially the same process theory, differing only in convention and philosophical perspective rather than in their underlying mathematics.   Essentially, these approaches can all be seen as particular quotients of a suitably chosen time-neutral theory. 

We also see that there are various other ways to symmetrise quantum theory. We consider two such ways, one of which has close connections to the work in \cite{di2020quantum,hardy2021time} whilst the other is, to our knowledge, a  new approach inspired by the view of antiparticles as being particles travelling back in time.  Specifically, the former can be viewed as a further restriction of a causal theory to also satisfy a `retrocausality' condition, while the latter can be viewed as an operational restriction on a time-neutral theory. 

It is worth commenting further on the process-theoretic nature of our construction. In particular, it means that whilst we are going to focus on quantum theory as a concrete example in this paper, these constructions can be applied more generally, for example, to arbitrary generalised probabilistic theories (GPTs) \cite{HardyAxiom,Barrett}, operational probabilistic theories (OPTs) \cite{Chiri1,Chiri2,d2017quantum},  categorical probabilistic theories \cite{gogioso2017categorical}, or even theories that make no reference to probabilities at all (as have often been considered in the categorical quantum mechanics literature \cite{AC1,cespek,CDKZ2,gogioso2017fantastic}). One interesting feature to highlight is that the symmetrised version of a theory will often be outside of the framework used to describe the original theory. For example, a time-symmetrised GPT is typically not itself a GPT. This showcases the utility of process theories as a flexible and widely applicable framework for exploring conceivable physical theories.  

 To summarise, the main contributions of this paper are threefold. First, we show that several existing approaches to time symmetry and time neutrality in quantum theory \cite{oreshkov2015operational, oeckl2008general, oeckl2016local, oreshkov2016operational, aharonov2009multiple, aharonov2008two, silva2017connecting, silva2014pre} are mathematically equivalent\footnote{Within their common scope of applicability.}, meaning that results and tools developed in any one of them transfer freely to all the others. Second, we introduce two new time-symmetric theories: $\QUnital$, which provides a precise process-theoretic realisation of a theory that is both causal and retrocausal, and $\QPart$, which gives a toy process-theoretic realisation of Feynman's intuition that antiparticles are particles travelling backwards in time, with non-signalling between particles and antiparticles playing the role of a consistency condition that prevents causal paradoxes. Third, we construct $\QSym$, a fully time-neutral formulation of quantum theory that contains all process matrices \cite{oreshkov2012quantum}, connecting our framework to the study of indefinite causal structure.

The remainder of the paper is structured as follows. In section~\ref{sec:PT} we introduce the basics of the framework of process theories. In particular, we define the key concepts which underpin this work namely: causality, time symmetry, and time neutrality. We illustrate these abstract definitions with the concrete example of quantum theory. 
In section~\ref{sec:3TS} we explore  the two aforementioned approaches that one could take towards obtaining a time symmetric version of quantum theory.  
In section~\ref{sec:2TN} we recast various other approaches to time symmetry and neutrality in two equivalent ways. The first is a direct adaptation of the existing literature whilst the second is a more elegant and broadly applicable way to capture the same idea.  We prove their equivalence in theorem~\ref{thm:equiv} and show that the theory can be made deterministic in theorem~\ref{thm:QNeut}. 
We end this section by discussing connections to higher-order processes \cite{chiribella2008quantum,oreshkov2012quantum,bisio2019theoretical,uijlen2019categorical,wilson2022mathematical}. 

This paper is intended as a broad overview of the possibilities of studying time symmetry and neutrality from a process-theoretic perspective. Each approach deserves further study in its own right. As such, we will try to highlight various open problems and research directions throughout.

\section{Process theories}\label{sec:PT}

A process theory \cite{coecke2011universe,CKpaperI,CKbook,reconstruction,selby2017process} is comprised of a collection of \emph{processes}. For example,
\beq\input{figures/process.tikz}\eeq
is a process with \emph{input systems} $A$ and $B$ and \emph{output systems} $A$, $C$, and $C$. Moreover, this collection of processes must be closed under forming \emph{diagrams} by wiring processes together. For example, the diagram,
\beq\input{figures/diagram.tikz},\eeq
corresponds to another process in the theory with inputs $D$ and $A$ and outputs $A$, $E$, and $C$.
This wiring is subject to the following constraints:
\ben\addtolength{\itemsep}{-0.5\baselineskip}
\item[i.] outputs cannot be wired to outputs, and inputs cannot be wired to inputs;
\item[ii.] the wiring is acyclic;
\item[iii.] systems can only be connected if their system types match.
\een
Moreover, we have a notion of diagrammatic equality: two diagrams are equal if they represent the same wiring. For example,
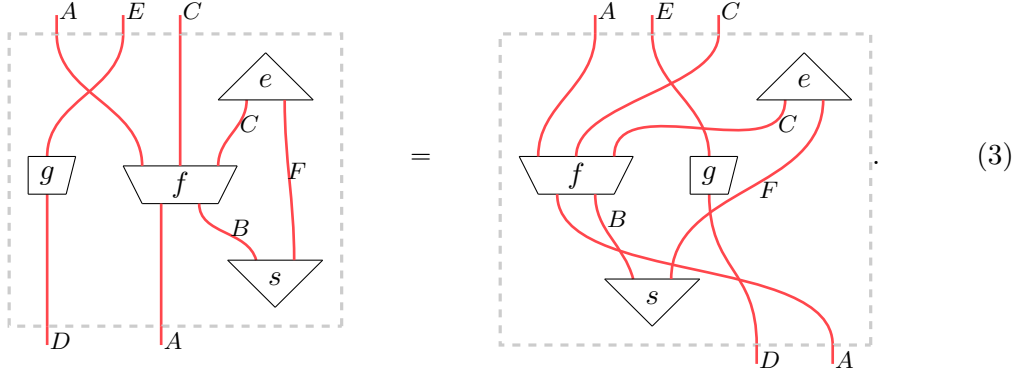
\beq\input{figures/diagram.tikz}\qquad=\qquad\input{figures/diagramRedrawn.tikz}.\eeq
In other words, the precise layout of processes on the page is not important: all that matters is their \emph{connectivity}.

\begin{remark} In categorical quantum mechanics \cite{AC1,AC3}, every symmetric monoidal category (SMC) is taken to be a process theory -- subject to Mac Lane's strictification theorem \cite{MacLane}.  Following \cite{patterson2021wiring}, we also proposed certain operad algebras as process theories \cite{operads}. 
\end{remark}

\subsection{Example 1: quantum physics}

To illustrate this abstract definition, we introduce a process theory that describes finite-dimensional quantum theory, $\QPhys$. There are two different kinds of systems in $\QPhys$: quantum systems, denoted by purple wires and labelled by finite-dimensional Hilbert spaces $\mathcal{H}$, and classical systems, denoted by grey wires and labelled by finite sets, $\mathds{A}$. The quantum systems are the fundamental systems of interest within quantum theory, whilst the classical systems represent how we interact with the quantum world. For example, they represent the control variables on experimental devices and the pointers on measurement apparatus that encode the outcome of the measurement.
This is the setting previously considered in \cite{coecke2016cpstar,selby2017leaks,gogioso2017categorical,tull2018categorical,reconstruction}, among other works.

Following these conventions, a general quantum instrument, $\mathcal{E}$, is denoted by
\beq\input{figures/cqProcessf.tikz}.\eeq
This has a quantum system, $\mathcal{H}$, together with a classical control variable, $\mathds{X}$, as inputs and a quantum system, $\mathcal{K}$, together with a classical outcome variable, $\mathds{A}$, as outputs. Formally, this can be understood as a completely positive trace-preserving (CPTP) map between complex matrix algebras,
\beq
\mathcal{E}:  \mathcal{B}[\mathcal{H}] \otimes \left(\bigoplus_{x\in \mathds{X}} \mathcal{B}[\mathds{C}]\right) \to \mathcal{B}[\mathcal{K}] \otimes \left(\bigoplus_{a\in \mathds{A}} \mathcal{B}[\mathds{C}]\right).
\eeq
 This general definition looks somewhat complicated, but it can be well understood by considering more familiar special cases.

If there is no classical input or output, then we formally treat this as the singleton set $\star:=\{*\}$. In this case we have that $\mathcal{B[H]}\otimes \left(\bigoplus_{*\in\star}\mathcal{B}[\mathds{C}]\right) \cong \mathcal{B[H]}$, the standard noncommutative matrix algebra associated with a quantum system $\mathcal{H}$. In contrast, if there is no quantum input or output, then we formally treat this as the one-dimensional Hilbert space $\mathds{C}$.  In this case we have that $\mathcal{B}[\mathds{C}] \otimes\left(\bigoplus_{x\in \mathds{X}} \mathcal{B}[\mathds{C}]\right) \cong \bigoplus_{x\in \mathds{X}} \mathcal{B}[\mathds{C}] $, a commutative matrix algebra commonly associated with a classical system.

A process with no inputs and only a quantum output,
\beq
\begin{tikzpicture}
	\begin{pgfonlayer}{nodelayer}
		\node [style=point] (5) at (-0.5, 0) {$\rho$};
		\node [style=none] (6) at (-0.5, 1.5) {};
		\node [style=right label] (15) at (-0.5, 1.25) {$\mathcal{H}$};
	\end{pgfonlayer}
	\begin{pgfonlayer}{edgelayer}
		\draw [QWire, in=90, out=-90] (6.center) to (5);
	\end{pgfonlayer}
\end{tikzpicture}
,
\eeq
therefore corresponds to a CPTP map, $\rho:\mathcal{B}[\mathds{C}] \to \mathcal{B}[\mathcal{H}]$. These are in one-to-one correspondence with positive elements of $\mathcal{B}[\mathcal{H}]$ of unit trace, i.e., quantum states.

A process with a quantum input and a classical output,
\beq
\input{figures/quantMeas.tikz},
\eeq
corresponds to a CPTP map, $M:\mathcal{B}[\mathcal{H}] \to \bigoplus_{a\in\mathds{A}}\mathcal{B}[\mathds{C}]$. These are in one-to-one correspondence with sets of positive operators in $\mathcal{B}[\mathcal{H}]$ indexed by $a\in \mathds{A}$, $\{M_a\}_{a\in \mathds{A}}$, such that $\sum_{a\in\mathds{A}} M_a =\mathds{1}_\mathcal{H}$. In other words, they are in one-to-one correspondence with destructive POVM measurements. 

When we compose a quantum state with a quantum measurement, we obtain a process which has only classical outputs,
\beq
\begin{tikzpicture}
	\begin{pgfonlayer}{nodelayer}
		\node [style=point] (7) at (0, -0.5) {$p$};
		\node [style=none] (8) at (0, 1) {};
		\node [style=right label] (16) at (0, 0.75) {$\mathds{A}$};
	\end{pgfonlayer}
	\begin{pgfonlayer}{edgelayer}
		\draw [cWire] (8.center) to (7);
	\end{pgfonlayer}
\end{tikzpicture}
\quad :=\quad \input{figures/quantProbDist.tikz},
\eeq
corresponding to a CPTP map $p: \mathcal{B}[\mathds{C}] \to \bigoplus_{a\in\mathds{A}}\mathcal{B}[\mathds{C}]$. This is simply the sequential composition of the CPTP maps $\rho$ and $M$, i.e., $p=M\circ \rho$. Processes of this form are in one-to-one correspondence with probability distributions over the set $\mathds{A}$ and, in particular, with the probability distribution defined by $p(a)=\mathsf{tr}(M_a \rho)$ for all $a\in \mathds{A}$. Hence, we see that the Born rule is encoded as a special case of how CPTP maps compose in sequence.

More generally, a process with only classical inputs and outputs,
\beq
\input{figures/classProc.tikz},
\eeq
corresponds to a CPTP map, $S:\bigoplus_{x\in \mathds{X}}\mathcal{B}[\mathds{C}]\to \bigoplus_{a\in \mathds{A}}\mathcal{B}[\mathds{C}]$. These are in one-to-one correspondence with stochastic maps from $\mathds{X}$ to $\mathds{A}$. They therefore map probability distributions over $\mathds{X}$ to probability distributions over $\mathds{A}$ via
\beq
\begin{tikzpicture}
	\begin{pgfonlayer}{nodelayer}
		\node [style=point] (17) at (0, -0.75) {$p$};
		\node [style=none] (18) at (0, 0.25) {};
		\node [style=right label] (19) at (0, 0) {$\mathds{X}$};
	\end{pgfonlayer}
	\begin{pgfonlayer}{edgelayer}
		\draw [cWire] (18.center) to (17);
	\end{pgfonlayer}
\end{tikzpicture}
\quad \mapsto\quad \input{figures/classProc1.tikz}.
\eeq

Some other special cases which are common in the literature are
\beq
\input{figures/quantNDM.tikz}\ ,\quad \input{figures/quantCCP.tikz}\ , \quad \text{and}\quad \input{figures/quantCCT.tikz},
\eeq
which describe non-destructive measurements, classically controlled state preparations, and classically controlled CPTP maps, respectively.

Finally, let us consider the set of processes with no outputs,
\beq
\input{figures/quantEff.tikz}.
\eeq
They correspond to CPTP maps $E:\mathcal{B}[\mathcal{H}] \otimes \left(\bigoplus_{x\in \mathds{X}} \mathcal{B}[\mathds{C}]\right)\to \mathcal{B}[\mathds{C}]$. Note, however, that due to the trace-preservation condition, these processes are unique. On the quantum system they correspond to the (partial) trace whilst on the classical system they correspond to marginalisation. Since they are unique, we introduce a special symbol to denote them:
\beq
\input{figures/quantDiscard.tikz}.
\eeq
A special case of the above process is when there is also no input system. Then, there is a unique process with neither input nor output which corresponds to the unique CPTP map from $\mathcal{B}[\mathds{C}]$ to itself. This can be thought of as the scalar $1$ which maps, for example, $\rho \mapsto 1\cdot \rho = \rho$. We diagrammatically denote the scalar $1$ by the empty diagram,
\beq
\input{figures/empty.tikz}\ .
\eeq

The process-theoretic presentation of quantum theory is a time-asymmetric theory: every system has a large number of states, i.e., density matrices, but only a single effect, i.e., the (partial) trace. Indeed, as demonstrated in an earlier paper,   \cite{coecke2017time}, the time reverse of $\QPhys$ is a remarkably different theory: it has only a single state and describes a theory of eternal noise. We abstractly characterise theories that have a trace-preservation property in section~\ref{sec:causalPT}, and formalise the notion of time symmetry which they (typically) violate in section~\ref{sec:timesymmetric}.

\subsection{Deterministic and causal process theories}\label{sec:causalPT}

The fact that quantum theory has a single scalar, namely the empty diagram, means that $\QPhys$ is a deterministic process theory:

\begin{definition}[Determinism]\label{def:determinism}
We say that a process theory which has only a single scalar, the empty diagram
\beq
\input{figures/empty.tikz}\ ,
\eeq
is a deterministic process theory.
\end{definition}
The reason for this terminology is that such theories describe only processes that can be realised with certainty. That is, they do not capture processes that occur only as one possibility amongst many. 

For example, in $\QPhys$ determinism means that we must consider measurements as channels from a quantum to a classical system, as this is something that we can choose to implement in the lab. One can also think of measurements as being represented by a collection of POVM elements. The individual POVM elements, however, are not part of $\QPhys$ as they are not things that we can choose to implement -- we have no control over which possibility will occur.   This should be contrasted with $\QCalc$, as defined in section ~\ref{sec:QCalc}, which includes (amongst other things) individual POVM elements as effects, leading to scalars in $[0,1]$.

The trace-preservation condition for quantum processes in $\QPhys$ can also be elegantly expressed
in process-theoretic terms. In short, it says that there is a unique way to go from something to nothing. More formally,  for any system $A$ there is a unique process with $A$ as an input with no outputs. These are often known as \emph{discarding} maps as they often correspond to simply throwing away or ignoring a system. As in the case of $\QPhys$, we denote these by:
\beq\label{eq:discardA}
\begin{tikzpicture}
	\begin{pgfonlayer}{nodelayer}
		\node [style=upground] (0) at (0, 0.75) {};
		\node [style=none] (1) at (0, -0.5) {};
		\node [style=none] (2) at (0, 0.5) {};
		\node [style=right label] (3) at (0, -0.25) {$A$};
	\end{pgfonlayer}
	\begin{pgfonlayer}{edgelayer}
		\draw [qWire] (2.center) to (1.center);
	\end{pgfonlayer}
\end{tikzpicture}
.
\eeq
\begin{definition}[Causality \cite{Chiri1, CRCaucat,Cnonsig}]\label{def:causal}
We say that any process theory which, for each system $A$, has a unique process from $A$ to nothing is a \emph{causal} process theory. 
\end{definition}

Many important conditions immediately follow for any causal process theory:
\begin{itemize}
\item[i.] Discarding a composite system is the same as discarding the components:
\beq
\begin{tikzpicture}
	\begin{pgfonlayer}{nodelayer}
		\node [style=upground] (0) at (0, 0.75) {};
		\node [style=none] (1) at (0, -0.5) {};
		\node [style=none] (2) at (0, 0.5) {};
		\node [style=right label] (3) at (0, -0.25) {$AB$};
	\end{pgfonlayer}
	\begin{pgfonlayer}{edgelayer}
		\draw [qWire] (2.center) to (1.center);
	\end{pgfonlayer}
\end{tikzpicture}
\quad =\quad \begin{tikzpicture}
	\begin{pgfonlayer}{nodelayer}
		\node [style=upground] (0) at (0, 0.75) {};
		\node [style=none] (1) at (0, -0.5) {};
		\node [style=none] (2) at (0, 0.5) {};
		\node [style=right label] (3) at (0, -0.25) {$A$};
	\end{pgfonlayer}
	\begin{pgfonlayer}{edgelayer}
		\draw [qWire] (2.center) to (1.center);
	\end{pgfonlayer}
\end{tikzpicture}
 \begin{tikzpicture}
	\begin{pgfonlayer}{nodelayer}
		\node [style=upground] (0) at (0, 0.75) {};
		\node [style=none] (1) at (0, -0.5) {};
		\node [style=none] (2) at (0, 0.5) {};
		\node [style=right label] (3) at (0, -0.25) {$B$};
	\end{pgfonlayer}
	\begin{pgfonlayer}{edgelayer}
		\draw [qWire] (2.center) to (1.center);
	\end{pgfonlayer}
\end{tikzpicture}
.
\eeq
This means that our theory has a well-defined subsystem structure, because, for example, it provides a unique way to define local states given a global state.
\item[ii.] Discarding the outputs of any process is the same as discarding the inputs:
\beq \label{eq:causality}
\input{figures/causalityEquation.tikz}
\eeq
This is known as \emph{discard preservation} and is the abstract characterisation of trace preservation of quantum theory.
\item[iii.] There is a unique scalar, and hence the theory is \emph{deterministic}.
This follows by considering the special case of the (unique) discarding effect with a trivial input. 
\end{itemize}

Whilst we have just seen that causality implies determinism, the converse is not necessarily true. For instance, the opposite category $\QPhys^{\text{op}}$\footnote{Denoted by $\QPhys_R$ in section~\ref{sec:timereversed}.} has a single scalar $1$ but multiple effects for each system. Moreover,  in section~\ref{sec:determinism} we define a process theory that is deterministic but not causal.

The terminology `causality' is used because the above condition ensures that the theory is non-signalling between causally separated regions \cite{Cnonsig} and, more generally, is compatible with relativistic causal structure \cite{kissinger2017equivalence}. For example, suppose we have two spacelike-separated parties who therefore cannot directly signal to one another. They may, however, have a common past and future in which their light cones intersect. This would be represented by a diagram of the shape:
\beq
\input{figures/dok.tikz}.
\eeq
Then, to describe the local physics from the perspective of the left-hand party, we discard the right-hand output as it is inaccessible to them:
\beq
\input{figures/dok1.tikz}.
\eeq
Now, using the fact that the only effect is the discarding map, and that $f_\beta$ satisfies Eq.~\eqref{eq:causality}, we find that:
\beq
\input{figures/dok1.tikz}\ \ =\ \ \input{figures/dok2.tikz} \ \ =\ \  \input{figures/dok3.tikz}.
\eeq
That is, the right-hand input is simply discarded and is entirely disconnected from the left-hand side. In particular, any control that the right-hand party may have had over the input cannot have any influence over what the left-hand party could observe at their output.

\begin{remark}
In more category-theoretic language, a causal process theory can be thought of as a symmetric monoidal category in which the monoidal unit is terminal \cite{Cnonsig}.
\end{remark}

\subsection{Time reversed process theories}\label{sec:timereversed}

To understand time symmetry we must first consider what it means to time reverse a process theory. That is, we must understand what transformation we are demanding the theory to be invariant under. 

For any process theory, $\mathbf{Proc}$, there is a very simple way to reverse the arrow of time. We simply read the diagrams in the theory from top to bottom rather than bottom to top. The inputs of a process are then considered outputs and the outputs are considered inputs. It is often, however, more convenient to take an active view of time reversal -- that is, we keep the bottom-to-top reading of the diagrams but we flip all of the diagrams upside down. This defines a new process theory, which we denote by $\mathbf{Proc}_{R}$. This has the same systems and processes as $\mathbf{Proc}$, but all of these processes go in the opposite direction. That is,
\beq
\input{figures/processDagger.tikz} \in \mathbf{Proc}_{R}\quad \iff\quad 
\input{figures/process.tikz} \in \mathbf{Proc}.
\eeq
Moreover, composition in $\mathbf{Proc}_R$ can be defined as the time reverse of the composition of the original processes in $\mathbf{Proc}$. This means, for example, that
\beq
\input{figures/reversedComp1.tikz} =\ \ \input{figures/reversedComp2.tikz}\ \  \in \mathbf{Proc}_R \quad \iff \quad \input{figures/reversedComp3.tikz} =\ \ \input{figures/reversedComp4.tikz}\ \  \in \mathbf{Proc}.
\eeq 
It is straightforward to see that time-reversing a theory twice leaves it invariant. That is, ${\mathbf{Proc}_R}_R = \mathbf{Proc}$.

\begin{remark} Categorically, time reversal is simply the standard contravariant functor from a category to the opposite category $R:\mathbf{Proc}\to \mathbf{Proc}_R$. 
\end{remark}

If the process theory, $\mathbf{Proc}$, that we start with is a causal process theory, then, as noted in   \cite{coecke2017time}, the time-reversed theory will, typically, be remarkably different. It describes a theory with a single state for every system, which is left invariant by every transformation -- it is a theory of eternal noise. 

\begin{definition}[Retrocausality]\label{def:retrocausal}
We say that any process theory which, for each system $A$, has a unique process from nothing to $A$ is a \emph{retrocausal} process theory. 
\end{definition}
In this somewhat less poetic language, the result of  \cite{coecke2017time} states that the time reverse of a causal theory is a retrocausal theory. In particular, in any retrocausal theory there will be a unique state for each system, which we denote by:
\beq
\begin{tikzpicture}
	\begin{pgfonlayer}{nodelayer}
		\node [style=downground] (0) at (0, -0.5) {};
		\node [style=none] (1) at (0, 0.75) {};
		\node [style=none] (2) at (0, -0.25) {};
		\node [style=right label] (3) at (0, 0.5) {$A$};
	\end{pgfonlayer}
	\begin{pgfonlayer}{edgelayer}
		\draw [qWire] (2.center) to (1.center);
	\end{pgfonlayer}
\end{tikzpicture}
.
\eeq
This, in many theories, can be thought of as a state of uniform noise. Uniqueness of this state implies that they compose as:
\beq
\begin{tikzpicture}
	\begin{pgfonlayer}{nodelayer}
		\node [style=downground] (0) at (0, -0.5) {};
		\node [style=none] (1) at (0, 0.75) {};
		\node [style=none] (2) at (0, -0.25) {};
		\node [style=right label] (3) at (0, 0.5) {$AB$};
	\end{pgfonlayer}
	\begin{pgfonlayer}{edgelayer}
		\draw [qWire] (2.center) to (1.center);
	\end{pgfonlayer}
\end{tikzpicture}
 \quad = \quad \begin{tikzpicture}
	\begin{pgfonlayer}{nodelayer}
		\node [style=downground] (0) at (0, -0.5) {};
		\node [style=none] (1) at (0, 0.75) {};
		\node [style=none] (2) at (0, -0.25) {};
		\node [style=right label] (3) at (0, 0.5) {$A$};
	\end{pgfonlayer}
	\begin{pgfonlayer}{edgelayer}
		\draw [qWire] (2.center) to (1.center);
	\end{pgfonlayer}
\end{tikzpicture}
\begin{tikzpicture}
	\begin{pgfonlayer}{nodelayer}
		\node [style=downground] (0) at (0, -0.5) {};
		\node [style=none] (1) at (0, 0.75) {};
		\node [style=none] (2) at (0, -0.25) {};
		\node [style=right label] (3) at (0, 0.5) {$B$};
	\end{pgfonlayer}
	\begin{pgfonlayer}{edgelayer}
		\draw [qWire] (2.center) to (1.center);
	\end{pgfonlayer}
\end{tikzpicture}
,
\eeq
and, moreover, that every other process will satisfy a \emph{retrocausality} constraint:
\beq\label{eq:retrocausality}
\input{figures/retrocausalityEquation.tikz}.
\eeq
Finally, note that, like causal process theories, retrocausal process theories are necessarily deterministic (in the sense of Def.~\ref{def:determinism}).

\subsection{Time symmetric process theories}\label{sec:timesymmetric}
 Refs.~\cite{oeckl2008general,
oeckl2016local,
oreshkov2016operational,
oreshkov2015operational,
aharonov2009multiple,
aharonov2008two} formulate quantum theory in a time-symmetric way. From a process-theoretic perspective, this means that the process theory is the same as the time-reversed theory. Equivalently, this means that time-reversal is \emph{internal} to the process theory. Formally, we capture this by saying that a process theory is time symmetric if and only if it permits a \emph{dagger} \cite{AC1,AC2,SelingerCPM}. 

\begin{definition}[Dagger]\label{def:dagger}
A dagger is a map, $\dagger$, from the process theory to itself that \emph{reflects diagrams}. Specifically, it acts on processes as
\beq\input{figures/process.tikz} \ \underset{\dagger}{\stackrel{\dagger}{\leftrightarrows}} \ \  \input{figures/processDagger.tikz},\eeq
and, moreover, on diagrams as
\beq\input{figures/diagram.tikz} \ \ \ \underset{\dagger}{\stackrel{\dagger}{\leftrightarrows}} \ \ \ \  \input{figures/diagramDagger.tikz}.\eeq
\end{definition}

To understand the connection to time symmetry, note that if a process theory has a dagger, then for every process $f:A\to B$ in the theory there exists a process $f^\dagger:B\to A$ (that is, the symbolic notation for the upside-down $f$). The process, $f^\dagger$, can be interpreted as the time-reversed version of the process, $f$. Moreover, it is such that wiring together the time-reversed processes in the reverse order is the same as taking the time reverse of the process representing the diagram. Hence, we say that process theories with daggers are time symmetric.

Note that this is not the same as demanding that every process in the theory is time symmetric\footnote{That would be the case if we additionally required that $f=f^\dagger$ and hence that there were no processes whose inputs differed from their outputs.}. Instead, we are imposing the time symmetry at the level of the theory as a whole. Indeed, it is in this sense that the theories of \cite{aharonov2008two,aharonov2009multiple,oeckl2008general,oreshkov2015operational,silva2017connecting,silva2014pre,oreshkov2016operational,oeckl2016local} can be considered to be time symmetric.

\begin{remark} If the process theory represents an SMC, then any involutive contravariant strong monoidal endofunctor $\dagger:\mathbf{Proc}\to\mathbf{Proc}$, which acts as the identity on objects, provides a dagger for the process theory. This allows us to see that the existence of a dagger implies that $\mathbf{Proc}$ and $\mathbf{Proc}_R$ are covariantly isomorphic as we can define the covariant isomorphism simply by $R\circ \dagger : \mathbf{Proc} \to \mathbf{Proc}_R$, which is covariant as $\dagger$ and $R$ are each contravariant. Note that $\mathbf{Proc}$ and $\mathbf{Proc}_R$ are, by definition, contravariantly isomorphic, so the covariance here is the key to having a time symmetric theory.
\end{remark}

\begin{proposition}{\cite[Thm. 3]{coecke2017time}}\label{thm:eternalnoise}
A time symmetric theory is causal if and only if it is retrocausal.
\end{proposition}
This is because there is a unique state for a system given by the dagger of the discarding map:
\beq
\begin{tikzpicture}
	\begin{pgfonlayer}{nodelayer}
		\node [style=downground] (0) at (0, -0.5) {};
		\node [style=none] (1) at (0, 0.75) {};
		\node [style=none] (2) at (0, -0.25) {};
		\node [style=right label] (3) at (0, 0.5) {$A$};
	\end{pgfonlayer}
	\begin{pgfonlayer}{edgelayer}
		\draw [qWire] (2.center) to (1.center);
	\end{pgfonlayer}
\end{tikzpicture}
 \ \ :=\ \ \dagger\left( \begin{tikzpicture}
	\begin{pgfonlayer}{nodelayer}
		\node [style=upground] (0) at (0, 0.75) {};
		\node [style=none] (1) at (0, -0.5) {};
		\node [style=none] (2) at (0, 0.5) {};
		\node [style=right label] (3) at (0, -0.25) {$A$};
	\end{pgfonlayer}
	\begin{pgfonlayer}{edgelayer}
		\draw [qWire] (2.center) to (1.center);
	\end{pgfonlayer}
\end{tikzpicture}
 \right).  
\eeq

It is therefore immediately clear, as noted in  \cite{coecke2017time}, that $\QPhys$ is not a time symmetric theory -- it is causal but not retrocausal. To see this, it suffices to note that there are multiple distinct states for any (non-trivial) system, e.g., the computational basis states of a qubit:
\beq
\begin{tikzpicture}
	\begin{pgfonlayer}{nodelayer}
		\node [style=point] (5) at (-0.5, 0) {$0$};
		\node [style=none] (6) at (-0.5, 1.5) {};
		\node [style=right label] (15) at (-0.5, 1.25) {$\mathds{C}^2$};
	\end{pgfonlayer}
	\begin{pgfonlayer}{edgelayer}
		\draw [QWire, in=90, out=-90] (6.center) to (5);
	\end{pgfonlayer}
\end{tikzpicture}
\quad \neq \quad \begin{tikzpicture}
	\begin{pgfonlayer}{nodelayer}
		\node [style=point] (5) at (-0.5, 0) {$1$};
		\node [style=none] (6) at (-0.5, 1.5) {};
		\node [style=right label] (15) at (-0.5, 1.25) {$\mathds{C}^2$};
	\end{pgfonlayer}
	\begin{pgfonlayer}{edgelayer}
		\draw [QWire, in=90, out=-90] (6.center) to (5);
	\end{pgfonlayer}
\end{tikzpicture}
.
\eeq

\subsection{Time neutral process theories}

In the work of  \cite{oreshkov2016operational}, the authors aim to go a step beyond the time symmetry that we have just described and instead want to have a version of quantum theory that is \emph{time neutral}. In process-theoretic terms, this means that we want to forget about the distinction between inputs and outputs entirely. This is possible if the theory has \emph{cups \& caps} \cite{AC1, Kindergarten}.

\begin{definition}[Cups \& caps]\label{def:CupsAndCaps} A \emph{cup} is a process, $\bigcup$, in a process theory that can be represented as a bent piece of wire,
\beq\input{figures/cupProcess.tikz}\ \ =:\ \ \input{figures/cupWire.tikz},\eeq
which allows us to connect inputs to inputs. Similarly, a \emph{cap} is a process, $\bigcap$, which can be represented as a wire bent in the other direction,
\beq\input{figures/capProcess.tikz}\ \ =:\ \ \input{figures/capWire.tikz},\eeq
which allows us to connect outputs to outputs.  For our representation of these processes as bent wires to actually make sense, that is, for it to be compatible with the idea that only the connectivity of diagrams matters, these processes must satisfy the following conditions:
\beq\label{eq:snake}\input{figures/snake.tikz}\quad\qquad\&\quad\qquad\input{figures/cupSymmetry.tikz}.\eeq
\end{definition}

\begin{remark}
There are other conditions that one would expect these two to satisfy such as:
\beq\input{figures/capSymmetry.tikz},\eeq
However, these conditions can be derived from the two above and therefore need not be imposed separately.
\end{remark}

Time neutrality goes a step further than time symmetry as it allows us to freely interchange some of the inputs with outputs and vice versa. This is in contrast to time symmetry which requires flipping all of them at once. What this means is that time-neutral theories allow a freer notion of wiring that neglects the input-output structure. For example, the following diagram is permissible:
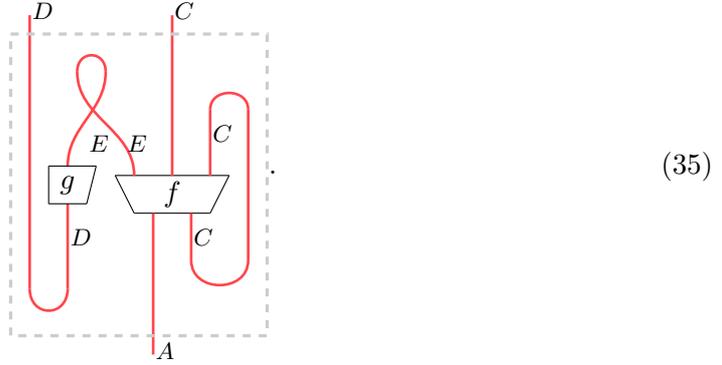
\beq\input{figures/timeNeutralDiagram.tikz}.\eeq
Theories with cups \& caps are therefore said to be time neutral as there is no real distinction between the inputs and outputs of a process\footnote{Recall from the definition of a process theory in Section~\ref{sec:PT} that the only distinction between the inputs and the outputs of a process is that outputs cannot be wired to outputs and inputs cannot be wired to inputs, and so cups \& caps allow us to circumvent this restriction.}. Moreover, there is no meaningful causal order that can be assigned to the processes within general diagrams because this freer notion of wiring allows for cycles. For example, in the diagram,
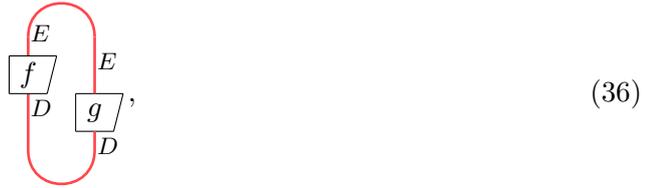
\beq
\input{figures/timeNeutralCycle.tikz},
\eeq
we have the situation in which $f$ is both in the `causal future' and the `causal past' of $g$. 

This process-theoretic notion of a time-neutral theory is stronger than the notion of a time-symmetric theory. That is, any time neutral theory is also necessarily time symmetric, as we can define a dagger using cups \& caps,
\beq
\dagger \left( \input{figures/daggerFromCompact.tikz}\right) \quad :=\quad  \input{figures/daggerFromCompact1.tikz}.
\eeq
It can easily be seen that this does indeed satisfy the constraints of Def.~\ref{def:dagger} \cite{CKbook}.

\begin{proposition}
If a time-neutral theory is causal, then there can be only a single process between any two systems.
\end{proposition}
\begin{proof}
As time neutrality implies time symmetry, we immediately have that causality implies retrocausality. Hence, for every system we have a unique state and effect. In particular, this means that,
\beq
\input{figures/capWire.tikz}\ \ =\ \  \input{figures/doubleDiscardA.tikz}\qquad\text{and}\qquad \input{figures/cupWire.tikz} \ \ = \ \ \input{figures/doubleNoiseA.tikz}.
\eeq
Equation \eqref{eq:snake} then implies that:
\beq
\begin{tikzpicture}
	\begin{pgfonlayer}{nodelayer}
		\node [style=none] (5) at (0.5, -1.25) {};
		\node [style=none] (6) at (-0.5, 1.25) {};
	\end{pgfonlayer}
	\begin{pgfonlayer}{edgelayer}
		\draw [qWire, in=90, out=-90] (6.center) to (5.center);
	\end{pgfonlayer}
\end{tikzpicture}
 \ \ = \ \ \input{figures/identityFactorise.tikz} \ \ = \ \ \input{figures/identityFactorise1.tikz} \ \ = \ \ \input{figures/identityFactorise2.tikz}.
\eeq
From the above we find that for any $f$ we have:
\beq
\input{figures/uniqueProcess.tikz} \ \ = \ \ \input{figures/uniqueProcess1.tikz} \ \ = \ \ \input{figures/uniqueProcess2.tikz} \ \ = \ \ \input{figures/uniqueProcess3.tikz} \ \ = \ \ \input{figures/uniqueProcess4.tikz}.
\eeq
As this holds for any process $f$, we have therefore shown that there is a unique process per pair of systems in the theory, namely:
\beq
\input{figures/uniqueProcess4.tikz}.
\eeq
\end{proof}

It is therefore clear that the process theory $\QPhys$ is not a time-neutral theory, since there are multiple distinct processes with the same inputs and outputs. For example:
\beq
\input{figures/quantIdent.tikz} \quad \neq \quad \input{figures/quantNoiseMap.tikz}.
\eeq

\begin{remark} Categorically, cups \& caps correspond to the unit and counit in a compact closed category in which the objects are equal to their dual.\end{remark}

One could argue that process theories with cups \& caps are not truly time neutral as the individual processes still have a distinction between input and output systems, even if we can now freely interchange them. In  \cite{operads} we show how to construct fully time-neutral process theories, where there is no distinction between inputs and outputs even at the level of a single process,  and then show that these are in fact equivalent to process theories with cups \& caps.

\subsection{Example 2: quantum calculations}\label{sec:QCalc}

We have seen that our first key example, $\QPhys$, is a causal process theory that fails to be either time neutral or time symmetric. There is, however, a closely related theory $\QCalc$, which is a supertheory of $\QPhys$ that is both time symmetric and time neutral. This is the theory in which we often perform calculations about quantum physics -- for example, when computing the probability of a measurement outcome.

The systems in $\QCalc$ are the same as those of $\QPhys$ and so we will use the same diagrammatic notation for them as we did in $\QPhys$. The entire difference between the two theories is a simple modification to the definition of the processes. That is, in $\QPhys$ processes were defined as CPTP maps, whilst in $\QCalc$ we drop the trace-preservation condition and allow for arbitrary completely positive maps. In particular, the `states' of quantum systems, $\mathcal{H}$, are given by arbitrary positive operators and so need not have unit trace, and the `states' of classical systems, $\mathds{X}$, are arbitrary $\mathds{R}^+$-valued functions on $\mathds{X}$ rather than probability distributions over $\mathds{X}$.

Notably, this theory is neither a causal nor a deterministic process theory. Indeed, there are many processes which have a system $\mathcal{H}$ as an input and no output,
\beq
\begin{tikzpicture}
	\begin{pgfonlayer}{nodelayer}
		\node [style=none] (0) at (0, -1) {};
		\node [style=copoint] (1) at (0, 0.5) {$e$};
		\node [style=right label] (2) at (0, -0.75) {$\mathcal{H}$};
	\end{pgfonlayer}
	\begin{pgfonlayer}{edgelayer}
		\draw [style=QWire] (1) to (0.center);
	\end{pgfonlayer}
\end{tikzpicture}
.
\eeq
These correspond to CP maps $e:\mathcal{B[H]}\to \mathcal{B}[\mathds{C}]$, which can be thought of as positive linear functionals on the cone of positive operators on $\mathcal{H}$. These include, but are not limited to, the quantum effects $\mathsf{tr}(\rho\ \cdot\ )$. Similarly, many processes have a system $\mathds{X}$ as an input and no output,
\beq
\begin{tikzpicture}
	\begin{pgfonlayer}{nodelayer}
		\node [style=none] (0) at (0, -1) {};
		\node [style=copoint] (1) at (0, 0.5) {$r$};
		\node [style=right label] (2) at (0, -0.75) {$\mathds{X}$};
	\end{pgfonlayer}
	\begin{pgfonlayer}{edgelayer}
		\draw [style=cWire] (1) to (0.center);
	\end{pgfonlayer}
\end{tikzpicture}
.
\eeq	
These correspond to CP maps $r:\bigoplus_{x\in \mathds{X}} \mathcal{B}[\mathds{C}] \to \mathcal{B}[\mathds{C}]$, which can be thought of as positive linear functionals on the cone of $\mathds{R}^+$-valued functions over $\mathds{X}$. 

We leveraged the causality of $\QPhys$ to demonstrate that it is neither time neutral nor time symmetric. Because $\QCalc$ is free of this causality condition, it is both time-symmetric and time-neutral. That is, we will show that it has a dagger as well as cups \& caps.

A suitable dagger for $\QCalc$ is given by the Hermitian adjoint, $\dagger_H$. There are other daggers, such as the transpose; however, the Hermitian adjoint is special because it has the property that, when applied to a state, it defines an effect that tests for that state ~\cite{selby2017diagrammatic,reconstruction}. That is,
\beq
\dagger_H :: \rho \mapsto \mathsf{tr}(\rho\ \cdot\ ).
\eeq
Moreover, it inverts reversible dynamics: for any unitary supermap $\mathcal{U}$,
\beq
\dagger_H :: \mathcal{U} \mapsto \mathcal{U}^{-1}.
\eeq
These two conditions mean that the Hermitian adjoint is well suited to be interpreted as time-reversal for $\QCalc$.

We can now turn to cups \& caps. In the case of quantum systems, $\mathcal{H}$, these can be expressed as
\beq
\input{figures/cupquantum.tikz}\ \  \sim \ \ \sum_{ij} \ketbra{ii}{jj} \ \ \sim \ \ \input{figures/capquantum.tikz},
\eeq
for some basis $\ket{i} \in \mathcal{H}$. That is, the cup is a supernormalised version of the Bell state and the cap is a supernormalised Bell effect. For classical systems, $\mathds{A}$, we express the cup and cap as
\beq
\input{figures/cupclassical.tikz}\ \  \sim \ \ \sum_{a\in \mathds{A}} \ketbra{a}{a} \ \ \sim \ \ \input{figures/capclassical.tikz}.
\eeq
That is, they are the supernormalised perfectly correlated probability distribution and response function, respectively.

Note that the dagger given by the Hermitian adjoint and cups \& caps interact in the way that one would expect, namely,
\beq\input{figures/cupquantum.tikz} \ \underset{\dagger_H}{\stackrel{\dagger_H}{\leftrightarrows}} \ \  \input{figures/capquantum.tikz}.\eeq
This implies that we can indeed view the Hermitian adjoint as a reflection of the diagrams.

Many of the processes within this theory have a clear physical interpretation. The trace-preserving processes have the clearest interpretation, since they are exactly the processes in $\QPhys$; beyond these, the trace-nonincreasing CP maps can also be interpreted straightforwardly. These can be viewed as processes that occur in some branch of a causal process, that is, processes that only occur as one possibility amongst many.

 One may have hoped that one could simply add trace-nonincreasing processes to $\QPhys$ to obtain a time symmetric theory. Ultimately, however, this is not enough and we must also include some trace-increasing CP maps. For example, the dagger of the discarding effect is a supernormalised maximally mixed state, and hence is a trace-increasing CP map. 
If we were to try to give these trace-increasing processes a physical interpretation, then we would run into problems. In particular, the theory cannot be given a sensible operational interpretation; for example, it permits `probabilities' that are greater than $1$. A simple example of this comes from composing the cups \& caps themselves. In particular, we find that:
\beq\label{eq:badprob}
\input{figures/CupCapIsDimension.tikz}\ \quad \text{and} \quad \ \input{figures/CupCapIsDimensionC.tikz}.
\eeq
 As scalars that are greater than one do not have any physical interpretation, $\QCalc$ cannot be a good description of nature. Nonetheless, as we mentioned earlier, it is extremely useful as a theory in which we perform calculations relevant to quantum physics. For example, if we want to compute the probability of some measurement outcome given a state, then we can simply compose the associated effect with the state:
\beq
\begin{tikzpicture}
	\begin{pgfonlayer}{nodelayer}
		\node [style=point] (0) at (0, -0.75) {$\rho$};
		\node [style=copoint] (1) at (0, 0.75) {$\sigma$};
		\node [style=right label] (2) at (0, 0) {$\mathcal{H}$};
	\end{pgfonlayer}
	\begin{pgfonlayer}{edgelayer}
		\draw [style=QWire] (1) to (0);
	\end{pgfonlayer}
\end{tikzpicture}
\quad = \quad \mathsf{tr}(\sigma \rho).
\eeq 
This necessarily results in a sensible probability if the state and effect are both trace-nonincreasing.

\subsection{Process theories with dual systems}\label{sec:duals}

In the previous subsections, we have assumed that systems are invariant under time reversal. While this is a natural assumption to make, it is also interesting to consider situations in which this is not the case. For example, in particle physics it is often taken to be the case that time-reversed particles are in fact antiparticles. Indeed, in section~\ref{sec:QPart} we define a new kind of time-symmetric quantum theory that captures this idea.

To investigate such situations, we introduce time-directed arrows to our systems,
\beq
\input{figures/CqWire.tikz}\quad\text{,}\quad
\begin{tikzpicture}
	\begin{pgfonlayer}{nodelayer}
		\node [style=none] (0) at (0, 1) {};
		\node [style=none] (1) at (0, -1) {};
		\node [style=right label] (2) at (0, -0.5) {$A$};
	\end{pgfonlayer}
	\begin{pgfonlayer}{edgelayer}
		\draw [style=RqWire] (1.center) to (0.center);
	\end{pgfonlayer}
\end{tikzpicture}
, 
\eeq
and say that these systems are \emph{dual} to one another. We symbolically denote these by $A^\uparrow$ and $A^\downarrow$, respectively. 

\begin{remark}
At this point, categorically, all we are doing is introducing a special kind of relabelling for the systems in the theory. This can be obtained by first defining an involutive function $r$ on the set of atomic systems (i.e., those with no nontrivial decomposition in terms of $\otimes$), using it to pair atomic systems as $(A,r(A))$, and then assigning each pair arbitrarily to $(A^\uparrow,A^\downarrow)$. Note that trivial pairings with $A=r(A)$ are also possible, in which case we would not introduce an arrow. Most choices of involutive function $r$ will not lead to anything particularly meaningful. Meaningful choices are those that allow for a dagger or cups and caps, as we now define.
\end{remark}

It is then clear that when we time reverse such a process theory, that is, by reading the diagrams from top to bottom rather than bottom to top, systems will get mapped to their duals. For example, 
\beq
\input{figures/processDaggerDual.tikz} \in \mathbf{Proc}_{R}\quad \iff\quad 
\input{figures/processDual.tikz} \in \mathbf{Proc}.
\eeq

This means that we must also suitably refine the notion of the dagger, and, hence, the condition for time symmetry for such theories. 

\begin{definition}[Daggers for process theories with duals]\label{def:daggerdual}
A dagger is a map, $\dagger$, from the process theory to itself that \emph{reflects diagrams}. In particular, it acts on processes as
\beq\input{figures/processDual.tikz} \ \underset{\dagger}{\stackrel{\dagger}{\leftrightarrows}} \ \  \input{figures/processDaggerDual.tikz},\eeq
and, moreover, on diagrams as
\beq\input{figures/diagramDual.tikz} \ \ \ \underset{\dagger}{\stackrel{\dagger}{\leftrightarrows}} \ \ \ \  \input{figures/diagramDaggerDual.tikz}.\eeq
\end{definition}
We can then say that a process theory with duals is time symmetric if and only if it has a dagger.

The definition of cups \& caps can also be suitably modified in this context to the bent wires
\beq
\input{figures/cupdual1.tikz}\ \ , \quad \ \input{figures/cupdual2.tikz}\ \ , \quad \ \input{figures/capdual1.tikz}\ \ , \quad \ \input{figures/capdual2.tikz}\ \ . 
\eeq
These bent wires satisfy the obvious diagrammatic equations.
A process theory with duals is then said to be time neutral if and only if it has such cups \& caps.

\subsection{Example 3: quantum symmetries}

An example of a category with duals is the category of representations of a group $G$ within quantum theory, denoted by $\QRep$. Its objects are representations of $G$ on objects of $\QCalc$, and its morphisms are the special processes in $\QCalc$ known as intertwiners.
An extensive discussion of this category in the context of categorical quantum dynamics can be found in \cite{gogioso2017thesis,gogioso2019dynamics}.
This process theory was previously considered in the context of resource theories of asymmetry as a way to define the free set of processes \cite[Example 3.9]{coecke2016mathematical}. The perspective we take here, however, is that this is a mixed state version of the ideas presented in \cite[Penrose (1971)]{baez2011prehistory}. These two perspectives are, however, equivalent to one another thanks to a covariant version of Stinespring's dilation theorem \cite{SCUTARU197979,verdon2021covariant}.

In section~\ref{sec:QPart} we motivate $\QRep$ physically by creating a toy model for particle physics, in which a particle is identified with a representation of $G$ on a quantum system $\mathcal{B(H)}$. Generally speaking, a particle is defined in the literature as a particular irreducible representation of the Poincar\'e group, which is infinite-dimensional. In our case, however, we construct a toy model by assuming that we have only finite-dimensional representations of some other group $G$.
However, it is known that the approach also generalises to (certain) infinite groups \cite{gogioso2019dynamics}.

This toy model has the potential to shed light on particle physics since it provides a neat description of the interaction of particles with classical systems. For instance, it allows for a clear definition of measurements within particle physics. In addition, it provides a passage from the standard pure-state formulation to a mixed-state formulation of particle physics, for example, via the CPM construction \cite{SelingerCPM}. 

We now define the category $\QRep$ using a diagrammatic notation for groups and their representations that we describe in App.~\ref{app:Groups}.

\begin{definition}
Consider a group $G$. The category $\QRep$ consists of the following data:
\begin{itemize}
\item Objects are pairs, 
\beq
\left( Q ,\ \input{GRep.tikz} \right),
\eeq
 where  $Q$ is an object in $\QCalc$ and $\pi_Q$  is a causal representation (see Eq.~\eqref{eq:CausGRep}) of $G$ on $Q$. 
\item The monoidal product of 
$(Q,\pi_Q)$ and $(Q',\pi'_{Q'})$ is given by 
\beq
\left(Q\otimes Q',\input{GCompRep1.tikz}\right).\eeq
\item The monoidal unit is the pair $(\mathds{C}, \pi_{\mathds{C}})$, where $\pi_{\mathds{C}}$ is the trivial representation:
\beq
\input{GRepTrivial.tikz}\ \ =\ \ \begin{tikzpicture}
	\begin{pgfonlayer}{nodelayer}
		\node [style=none] (0) at (0, -1) {};
		\node [style=right label] (1) at (0, -0.75) {$G$};
		\node [style=none] (5) at (0, 0) {};
		\node [style=upground] (18) at (0, 0.25) {};
	\end{pgfonlayer}
	\begin{pgfonlayer}{edgelayer}
		\draw [GWire, in=270, out=90] (0.center) to (5.center);
	\end{pgfonlayer}
\end{tikzpicture}
.
\eeq
\item Morphisms from $(Q,\pi_Q)$ to $(Q',\pi'_{Q'})$ are intertwiners in $\QCalc$, that is, they are CP maps, $\mathcal{E}$, satisfying the covariance condition:
\beq
\input{GIntertwiner.tikz} = \quad \input{GIntertwiner1.tikz}.
\eeq
\item Composition is inherited from $\QCalc$; the composite of intertwiners is again an intertwiner.
\item  The identity is the identity process which can also be seen to be an intertwiner.
\item Finally, we take each system and its conjugate, $(Q,\pi_Q)$ and $(Q^*,\pi^*_{Q^*})$, to form a dual pair, for example,
\beq
\input{GRepSysUp.tikz},\ \text{and} \ \input{GRepSysDown.tikz},
\eeq
where $\pi^*$ is the conjugate representation defined in Eq.~\eqref{eq:congRep}. Note that the choice of which is $\uparrow$ and which is $\downarrow$ is arbitrary as conjugation is involutive. 
\end{itemize}
\end{definition}

Note that one can show that cups \& caps are intertwiners (Eq.~\eqref{eq:capIntertwiner}) and hence $\QRep$ is a time neutral process theory with duals. Like $\QCalc$, however, $\QRep$ cannot be directly interpreted as a theory of physics -- that is, it does not necessarily make sensible probabilistic predictions. We therefore need to find a physicality condition on the processes in $\QRep$ akin to the restriction of $\QCalc$ to $\QPhys$ via the causality condition. In section~\ref{sec:QPart}, we propose an implementation that preserves the time symmetry of $\QRep$.

 \subsection{Summary of approaches}

We have now seen that causality in $\QPhys$ serves as the main obstacle towards time symmetry and time neutrality. In particular, time symmetry together with causality implies that there should be a single state per system, and time neutrality together with causality imply that there should be a single transformation between any pair of systems, neither of which is true within $\QPhys$. However, the time symmetric and time neutral theories that we have introduced, $\QCalc$ and $\QRep$, both have the problem that they do not have a coherent operational interpretation, as they have scalars that are greater than $1$.

Having identified the root of the time asymmetry within quantum theory, we can then ask what we can do to obtain a symmetric theory. We have identified several possible approaches:
\begin{enumerate}
\item[1.] We can restrict $\QPhys$ to a subtheory that additionally satisfies the retrocausality constraint -- that is, every system is both causal and retrocausal. This is related to the work in Refs.~\cite{hardy2021time} and \cite{di2020quantum}.
\item[2.]  We work within $\QRep$, a theory built on $\QCalc$ which is the supertheory of $\QPhys$. $\QRep$'s systems come in dual pairs, and we restrict the theory to processes satisfying non-signalling between causal and retrocausal systems (Eq.~\eqref{eq:QuantNS}). This gives a theory, $\QPart$, in which every system is either causal or retrocausal, serving as a toy model of particle physics. Crucially, this is not a direct extension of $\QPhys$ but rather a restriction of the time-neutral theory $\QRep$.
\item[3.] We can start with the supertheory of $\QPhys$, $\QCalc$, and then, to avoid unphysical predictions, we can:
\begin{enumerate}
\item Modify the composition rule. This is closely related to the works of Refs.~\cite{aharonov2008two,
aharonov2009multiple,
oeckl2008general,
oreshkov2015operational,
silva2017connecting,
silva2014pre,
oreshkov2016operational,
oeckl2016local}.
\item Modify the processes. This is equivalent to (a) but is a more elegant and adaptable presentation of the theory.
\end{enumerate}
\end{enumerate}
 The first two options can lead to a time symmetric theory  (but not a time neutral theory) and are formally introduced in the following section,  while the third can lead to a time neutral theory which  we formally introduce in section~\ref{sec:2TN}.

\section{Two approaches to time symmetry}\label{sec:3TS}

\subsection{Causal and retrocausal}

Given that causality is the main obstacle towards time symmetry within quantum theory, perhaps the most obvious way to symmetrise the theory is to restrict the processes to those that also satisfy a retrocausality condition. 
\begin{definition}[Bicausality]\label{def:bicausal} A process theory that is both causal (Def.~\ref{def:causal}) and retrocausal (Def.~\ref{def:retrocausal}) is said to be \emph{bicausal}\footnote{Known as double causality in \cite{hardy2021time}.}. In particular, this means that it has a unique effect and a unique state for each system.
\end{definition}
In such theories, the simple argument against time symmetry of causal theories (namely, that there are more states than effects) breaks down, and so it is at least plausible that bicausal theories can be time symmetric. It is therefore interesting to explore how to construct bicausal theories out of causal theories.

Given any causal process theory, we can construct a subtheory that additionally satisfies a retrocausality constraint. To do so, however, we must pick a particular state for each system that we will then require be the unique state for the system. Typically there will not be a unique way to choose these states, but they nonetheless must satisfy certain consistency conditions for the resulting theory to be well defined. Let us denote these candidate states by:
\beq
\begin{tikzpicture}
	\begin{pgfonlayer}{nodelayer}
		\node [style=point] (5) at (0, -0.5) {$\mu$};
		\node [style=none] (6) at (0, 1) {};
		\node [style=right label] (15) at (0, 0.75) {$A$};
	\end{pgfonlayer}
	\begin{pgfonlayer}{edgelayer}
		\draw [qWire, in=90, out=-90] (6.center) to (5);
	\end{pgfonlayer}
\end{tikzpicture}
.
\eeq
These must be chosen such that they satisfy
\beq\label{eq:compMaxMixed}
\begin{tikzpicture}
	\begin{pgfonlayer}{nodelayer}
		\node [style=point] (5) at (0, -0.5) {$\mu$};
		\node [style=none] (6) at (0, 1) {};
		\node [style=right label] (15) at (0, 0.75) {$AB$};
	\end{pgfonlayer}
	\begin{pgfonlayer}{edgelayer}
		\draw [qWire, in=90, out=-90] (6.center) to (5);
	\end{pgfonlayer}
\end{tikzpicture}
 \ \ = \ \ \begin{tikzpicture}
	\begin{pgfonlayer}{nodelayer}
		\node [style=point] (5) at (0, -0.5) {$\mu$};
		\node [style=none] (6) at (0, 1) {};
		\node [style=right label] (15) at (0, 0.75) {$A$};
	\end{pgfonlayer}
	\begin{pgfonlayer}{edgelayer}
		\draw [qWire, in=90, out=-90] (6.center) to (5);
	\end{pgfonlayer}
\end{tikzpicture}
\begin{tikzpicture}
	\begin{pgfonlayer}{nodelayer}
		\node [style=point] (5) at (0, -0.5) {$\mu$};
		\node [style=none] (6) at (0, 1) {};
		\node [style=right label] (15) at (0, 0.75) {$B$};
	\end{pgfonlayer}
	\begin{pgfonlayer}{edgelayer}
		\draw [qWire, in=90, out=-90] (6.center) to (5);
	\end{pgfonlayer}
\end{tikzpicture}
,\eeq
in order to ensure that the resulting theory is closed under parallel composition. Note that, because they belong to a causal theory, they also automatically satisfy
\beq
\begin{tikzpicture}
	\begin{pgfonlayer}{nodelayer}
		\node [style=point] (5) at (0, -0.75) {$\mu$};
		\node [style=none] (6) at (0, 0.5) {};
		\node [style=right label] (15) at (0, 0) {$A$};
		\node [style=upground] (16) at (0, 0.75) {};
	\end{pgfonlayer}
	\begin{pgfonlayer}{edgelayer}
		\draw [qWire, in=90, out=-90] (6.center) to (5);
	\end{pgfonlayer}
\end{tikzpicture}
 \ \ = \ \ \input{figures/empty.tikz}.
\eeq

We can then restrict the allowed processes within the theory to those that satisfy:
\beq
\input{figures/maxMixPres.tikz} \ \ = \ \ \begin{tikzpicture}
	\begin{pgfonlayer}{nodelayer}
		\node [style=point] (5) at (0, -0.5) {$\mu$};
		\node [style=none] (6) at (0, 1) {};
		\node [style=right label] (15) at (0, 0.75) {$B$};
	\end{pgfonlayer}
	\begin{pgfonlayer}{edgelayer}
		\draw [qWire, in=90, out=-90] (6.center) to (5);
	\end{pgfonlayer}
\end{tikzpicture}
.
\eeq
Having restricted the theory in this way, we can note that the remaining subtheory satisfies the retrocausality constraint by taking
\beq
\begin{tikzpicture}
	\begin{pgfonlayer}{nodelayer}
		\node [style=downground] (0) at (0, -0.5) {};
		\node [style=none] (1) at (0, 0.75) {};
		\node [style=none] (2) at (0, -0.25) {};
		\node [style=right label] (3) at (0, 0.5) {$A$};
	\end{pgfonlayer}
	\begin{pgfonlayer}{edgelayer}
		\draw [qWire] (2.center) to (1.center);
	\end{pgfonlayer}
\end{tikzpicture}
\ \  := \ \ \begin{tikzpicture}
	\begin{pgfonlayer}{nodelayer}
		\node [style=point] (5) at (0, -0.5) {$\mu$};
		\node [style=none] (6) at (0, 1) {};
		\node [style=right label] (15) at (0, 0.75) {$A$};
	\end{pgfonlayer}
	\begin{pgfonlayer}{edgelayer}
		\draw [qWire, in=90, out=-90] (6.center) to (5);
	\end{pgfonlayer}
\end{tikzpicture}
,
\eeq
for all systems $A$. 

For a general causal process theory and a general choice of these states, there is no reason to believe that this will result in a time symmetric theory. That is, even after ensuring that we only consider processes that are both causal and retrocausal, it may still not be possible to define a dagger. In the case of quantum theory, we show shortly that for a suitable choice of states we do end up with a theory with a dagger. However, there are other generalised probabilistic theories (e.g., those that do not satisfy any notion of self-duality \cite{muller2012structure}) in which this is not possible, at least not without imposing further constraints on the sets of processes.

Returning to the case of $\QPhys$, the natural choice to make for the unique state is $\mu_\mathcal{H} := \frac{1}{|\mathcal{H}|} \mathds{1}_{\mathcal{H}}$, that is, the maximally mixed state for the system:
\beq
\begin{tikzpicture}
	\begin{pgfonlayer}{nodelayer}
		\node [style=none] (5) at (0, -0.25) {};
		\node [style=none] (6) at (0, 1.25) {};
		\node [style=right label] (15) at (0, 1) {$\mathcal{H}$};
		\node [style=downground] (16) at (0, -0.5) {};
	\end{pgfonlayer}
	\begin{pgfonlayer}{edgelayer}
		\draw [QWire, in=90, out=-90] (6.center) to (5.center);
	\end{pgfonlayer}
\end{tikzpicture}
\ \ := \ \ \frac{1}{|\mathcal{H}|}\begin{tikzpicture}
	\begin{pgfonlayer}{nodelayer}
		\node [style=point] (5) at (0, -0.25) {$\mathds{1}$};
		\node [style=none] (6) at (0, 1.25) {};
		\node [style=right label] (15) at (0, 1) {$\mathcal{H}$};
	\end{pgfonlayer}
	\begin{pgfonlayer}{edgelayer}
		\draw [QWire, in=90, out=-90] (6.center) to (5);
	\end{pgfonlayer}
\end{tikzpicture}
.
\eeq
 (Note that the RHS of this equation is actually a diagram living in $\QCalc$ rather than $\QPhys$ as it involves a supernormalised state $\mathds{1}_\mathcal{H}$.) 
It is simple to verify that these indeed satisfy the compositionality condition of Eq.~\eqref{eq:compMaxMixed}. The constraint that is then imposed on the processes of $\QPhys$ in order to define the subtheory is
\beq
\input{figures/unital.tikz} \ \ = \ \ \begin{tikzpicture}
	\begin{pgfonlayer}{nodelayer}
		\node [style=none] (22) at (0, 0) {};
		\node [style=none] (23) at (0, 1.25) {};
		\node [style=right label] (24) at (0, 1) {$\mathcal{K}$};
		\node [style=downground] (25) at (0, -0.25) {};
	\end{pgfonlayer}
	\begin{pgfonlayer}{edgelayer}
		\draw [QWire] (23.center) to (22.center);
	\end{pgfonlayer}
\end{tikzpicture}
.
\eeq
This means that $\mathcal{E}$ maps the maximally mixed state to the maximally mixed state, i.e., that $\mathcal{E}$ is a \emph{unital} CPTP map. In the special case that the inputs and outputs are classical we find that:
\beq
\input{figures/unital2.tikz} \ \ = \ \ \begin{tikzpicture}
	\begin{pgfonlayer}{nodelayer}
		\node [style=none] (22) at (0, 0) {};
		\node [style=none] (23) at (0, 1.25) {};
		\node [style=right label] (24) at (0, 1) {$\mathds{A}$};
		\node [style=downground] (25) at (0, -0.25) {};
	\end{pgfonlayer}
	\begin{pgfonlayer}{edgelayer}
		\draw [cWire] (23.center) to (22.center);
	\end{pgfonlayer}
\end{tikzpicture}
,
\eeq
or, in other words, they are bistochastic maps\footnote{Bistochastic maps are typically taken to be square matrices, but this is the natural notion that also applies to the non-square case \cite{hardy2021time}.}. We will denote the subtheory of unital CPTP maps $\QUnital$.

One may then be tempted to define the dagger using the Hermitian adjoint, $\dagger_{H}$, as we did for $\QCalc$. Note, however, that the Hermitian adjoint maps the discarding map to the supernormalised state, $\mathds{1}_{\mathcal{H}}$, rather than the unique state that we have defined as, $\frac{1}{|\mathcal{H}|} \mathds{1}_{\mathcal{H}}$. To take care of these normalisation issues we must therefore define the dagger as
\beq
\dagger\left(\input{figures/cqProcessf.tikz} \right) := \dagger_H\left(\input{figures/cqProcessf.tikz} \right)\frac{|\mathcal{K}||\mathds{A}|}{|\mathcal{H}||\mathds{X}|}.
\eeq
It is straightforward to show that this defines a dagger: in particular, it maps CP maps to CP maps and trace-preserving maps to unital maps, and vice versa. In other words:
\begin{proposition}
$\QUnital$ is a time symmetric process theory.
\end{proposition}

On the face of it, however,  $\QUnital$ does not seem to be a good candidate to describe our world: there is only a single state for every system. The theory is not entirely trivial, however, as it still has interesting transformations. For example, it still contains unitary evolution and, moreover, makes non-trivial classical `predictions' in the form of the bistochastic matrices. It is not straightforward, however, to see how these bistochastic matrices suffice to explain our everyday experiences. For example, they do not allow us to copy classical information which is, intuitively, an operation that we would expect to be able to do.

Despite these obstacles, this sort of approach is closely related to the approach advocated in Refs.~\cite{hardy2021time} and \cite{di2020quantum}.  The idea is that whilst they may not describe our everyday experiences in the way that we might typically formulate them, they can still be recovered by suitably choosing which classical random variables we condition on. That is, the asymmetry arises as a choice of the direction in which we choose to make inferences, rather than being part of the underlying physics. 
It is therefore an interesting question to ask how the formalisms of Refs.~\cite{hardy2021time} and \cite{di2020quantum} (in particular, the idea of choosing which variables to condition on) can be captured within a process-theoretic context. We leave this as an important direction for future research.

\subsection{Causal or retrocausal} \label{sec:QPart}

In contrast to the previous section, in which we imposed both causality and retrocausality on every system, we now formulate a theory in which each system satisfies either causality or retrocausality while the theory remains time-symmetric.

To do so, we work with process theories with duals, where we view a system $A^\uparrow$ as a causal system and $A^\downarrow$ as its retrocausal counterpart. To enforce this interpretation, we demand that systems $A^\uparrow$ have a unique effect and that systems $A^\downarrow$ have a unique state. Hence, in particular, processes $A^\uparrow \to B^\uparrow$ necessarily satisfy the causality condition (Def.~\ref{def:causal}), whereas those from $A^\downarrow \to B^\downarrow$ satisfy the retrocausality condition (Def.~\ref{def:retrocausal}).

A general process in such a theory, however, has both causal and retrocausal inputs and outputs. Diagrammatically, this is denoted by:
\beq
\input{figures/generalCausalRetrocausalProcess.tikz}.
\eeq
The important question is: what type of ``causality'' condition should this process satisfy so that the (retro)causality conditions hold for the systems $A^\uparrow$ and $A^\downarrow$ individually?

To get to grips with this condition, we work within our example category $\QRep$. In this case, systems $A^\uparrow$ correspond to pairs $(Q, \pi_{Q})$ and systems  $A^\downarrow$ to dual pairs $(Q^*, \pi^*_{Q^*})$.  This correspondence can be thought of as a toy model for particle physics: particles are defined as the (causal) pairs  $(Q, \pi_{Q})$ while antiparticles are defined as (retrocausal) pairs $(Q^*, \pi^*_{Q^*})$. In a sense, this takes seriously  Feynman's interpretation of antiparticles as being particles travelling back in time. The flexibility of our construction allows particles to interact with classical systems. Thus, it can be viewed as a way to include the concept of measurement in particle physics. Furthermore, it allows for particles to be in mixed states rather than pure as in the standard formulation. We could also define this directly for $\QCalc$ using the duality between $Q$ and $Q^*$. Indeed, this can be seen as the special case in which $\pi_Q$ is the trivial representation. However, we work with $\QRep$ to make these tentative connections to particle physics.

As mentioned above, if we only have causal inputs and outputs then the process should be causal:
\beq\label{eq:causalIsCausal}
\input{figures/qCausal.tikz} \ \ = \ \ \input{figures/qCausal1.tikz}.
\eeq
Similarly, the processes with only retrocausal inputs and outputs should be retrocausal processes:
\beq\label{eq:retroIsRetro}
\input{figures/qRetro.tikz} \ \ = \ \ \begin{tikzpicture}
	\begin{pgfonlayer}{nodelayer}
		\node [style=none] (0) at (0, 1.5) {};
		\node [style=none] (1) at (0, 0) {};
		\node [style=right label] (2) at (0, 1.25) {$A$};
		\node [style=downground] (11) at (0, -0.25) {};
	\end{pgfonlayer}
	\begin{pgfonlayer}{edgelayer}
		\draw [style=RQWire] (1.center) to (0.center);
	\end{pgfonlayer}
\end{tikzpicture}
,
\eeq
where, in this case, we take the unique state of the retrocausal system to be the supernormalised maximally mixed state $\mathds{1}$.

We will now show, however, that imposing these conditions alone does not result in a theory that is closed under composition; we then identify precisely what more is needed to obtain a consistent theory.

  To begin, note that Eq.~\eqref{eq:causalIsCausal} tells us that there is a unique effect for a causal system (the discard) and that Eq.~\eqref{eq:retroIsRetro} tells us that there is a unique state for a retrocausal system (the retro-discard). Now, if we have a bipartite effect with both a causal input and a retrocausal input, and we compose it with the retro-discard on the retrocausal system, then we have an effect on the causal system, which, as just argued, must be the discard. Similarly, if we have a bipartite state with both a causal and a retrocausal output, and we compose it with the discard on the causal system, then we are left with a state of the retrocausal system which, as just argued, must be the retro-discard. That is, we find that: 
 \beq\label{eq:rcStateEffect}
\input{figures/qInteraction2.tikz}\quad\text{and}\quad\input{figures/qInteraction3.tikz}.
\eeq
 Cups and caps provide examples of such processes, since they satisfy:
\beq
\input{figures/qcapInteraction2.tikz}\quad\text{and}\quad\input{figures/qcupInteraction3.tikz}\quad\text{respectively.}
\eeq
This example, however, indicates the need for constraints on states and effects beyond those in Eq.~\eqref{eq:rcStateEffect} as we can compose these cups \& caps to achieve violations of Eqs.~\eqref{eq:causalIsCausal} and \eqref{eq:retroIsRetro}. Specifically, as noted in Eq.~\eqref{eq:badprob}, if we compose a cup with a cap then we end up with a scalar other than the empty diagram:
\beq
\input{figures/qInteractionCompoition.tikz} \quad \neq \quad \input{figures/empty.tikz} \quad \neq \quad \input{figures/qInteractionCompositionClassical.tikz}.
\eeq
Composing either of these with a causal process, $f$, on a causal system will then give us a process that violates the causality condition, i.e.,
\beq
\input{figures/qInteractionCompositionClassical.tikz}\input{figures/qCausal.tikz} \ \ = \ \ \input{figures/qInteractionCompositionClassical.tikz}\input{figures/qCausal1.tikz}\ \ \neq \ \ \input{figures/qCausal1.tikz}.
\eeq

We have therefore seen that a straightforward application of Eqs.~\eqref{eq:causalIsCausal} and \eqref{eq:retroIsRetro} does not lead to a theory that is closed under composition. 
This example, however, does give us a hint as to what necessary and sufficient conditions must be imposed on general processes to ensure that Eqs.~\eqref{eq:causalIsCausal} and \eqref{eq:retroIsRetro} always hold. 

Specifically, note that the classical cup in this context can be viewed as perfect signalling from a retrocausal to a causal system, whilst the classical cap can be viewed as perfect signalling from a causal to a retrocausal system:
\beq
\input{figures/signal1.tikz}\quad \text{and} \quad \input{figures/signal2.tikz} \quad \forall x\in \mathds{X}.
\eeq
We can therefore conjecture that it is the signalling between the causal and retrocausal systems that leads us to problems. This is perhaps not surprising as many well-known paradoxes arise from the closed time loops that we could construct via:
\beq
\input{figures/timeloop.tikz}.
\eeq

To prevent such loops from occurring, we must therefore impose non-signalling conditions on our processes.  Formally, what this means is that every process $F$ must satisfy both
\beq\label{eq:QuantNS}
\input{figures/qNoSignaling1.tikz}\quad\text{and}\quad\input{figures/qNoSignaling2.tikz},
\eeq
for some retrocausal process, $F_r$, and some causal process, $F_c$. We recover Eqs.~\eqref{eq:causalIsCausal} and \eqref{eq:retroIsRetro} as special cases by considering only causal (or retrocausal) inputs and outputs. Moreover, these non-signalling conditions imply that a general process in $\QRep$ 
 must satisfy:
\beq\label{eq:dualcausal}
\input{figures/qCR.tikz},
\eeq
which rules out interactions that directly send information from a particle to an antiparticle as it implies that:
\beq
\input{figures/qInteraction1.tikz}.
\eeq
Note also that only one scalar satisfies this constraint:
\beq\label{eq:rcScalar}
\input{figures/empty.tikz}
\eeq
which implies that the theory is deterministic.

Finally, we can be sure that imposing this constraint leads to a consistent theory as it is straightforward to see that this condition is closed under composition. For example, if $F$ and $G$ are both non-signalling, then so is their sequential composite:
\beq
\input{figures/noSignallingComposition.tikz}\ \ \text{and} \ \ \input{figures/noSignallingComposiion1.tikz}.
\eeq

We have therefore identified the constraint that our toy model of particle physics -- with particles represented by finite-dimensional representations of a group $G$ and their antiparticle counterparts by conjugate representations -- must satisfy to be time-symmetric while making sensible operational predictions. This leads us to the following definition:

\begin{definition}\label{def:QPart}
We define the process theory $\QPart$ as the restriction of $\QRep$ to the subtheory of processes satisfying non-signalling from particles to antiparticles and vice versa (i.e., Eq.~\eqref{eq:QuantNS}).
\end{definition}

Note that whilst Def.~\ref{def:QPart}  imposes that there is non-signalling between causal and retrocausal systems in $\QPart$, it does not mean that the interactions between these types of systems are necessarily trivial. In particular, it is well known in quantum and classical theory that there can be processes that do not permit signalling but that are not of product form, e.g., PR-boxes \cite{popescu1994quantum}. 

The dagger for $\QRep$ is also a dagger for $\QPart$. Hence:
\begin{proposition}$\QPart$ is a time symmetric process theory.
\end{proposition}

Returning now to general process theories, we take the non-signalling conditions (from causal to retrocausal and vice versa) as the definition of a well-behaved theory with causal and retrocausal systems: 

\begin{definition}[Dual-causal]
A process theory with dual systems is said to be dual-causal if it is non-signalling from causal to retrocausal and vice versa, namely, if
\beq
\input{figures/NoSignaling1.tikz}\quad\text{and}\quad\input{figures/NoSignaling2.tikz}
\eeq
for all processes $F$.
\end{definition}

Given any causal process theory, there is a straightforward way to construct such a time-symmetric theory by forbidding any interactions between the causal and retrocausal systems. In particular, we demand that:
\beq \label{eq:rcProuct}
\forall F \ \exists F_c \text{ and } F_r \text{ such that } \qquad \input{figures/NonInteracting.tikz}
\eeq
where $F_c$ belongs to the original causal process theory $\mathbf{Proc}$ and $F_r$ is from the associated retrocausal theory $\mathbf{Proc}_R$. We can then define a dagger by simply using the time reversal map $R$  as
\beq
\input{figures/rcDagger1.tikz} \ \ \underset{\dagger}{\stackrel{\dagger}{\leftrightarrows}} \ \  \input{figures/rcDagger.tikz},
\eeq
where now $F_c \in \mathbf{Proc}_R$ is a retrocausal process and $F_r\in\mathbf{Proc}$ is a causal process.

\begin{remark}Categorically, what we are defining is simply $\mathbf{Proc}\times\mathbf{Proc}_R$ and noting that we can use the $R$ functor to define a dagger as $\mathbf{Proc}\times\mathbf{Proc}_R \cong \mathbf{Proc}_R\times\mathbf{Proc}$. \end{remark}

This, however, does not seem like a particularly useful or insightful theory. For it to be of interest, we need non-trivial interactions between the causal and retrocausal systems. Hence, an important direction for further research is to ask whether there is a generic construction that takes a causal process theory to a dual-causal process theory such that there are non-trivial interactions between causal and retrocausal systems.

\subsection{Discussion}\label{sec:ProcessMatrixDiscussion}

In this section, we have seen two possible ways to try to construct time-symmetric theories, with a focus on quantum theory as our key example. Our third approach, which can moreover lead to time-neutral theories, will be discussed in the next section. All of these examples, however, show the utility of the process-theoretic approach to this problem. It provides an intuitive definition of time symmetry (in terms of the existence of a dagger) and hence a clear target to aim for when developing time-symmetric theories. Moreover, it provides a flexible framework in which such theories can be developed and compared.

The most important directions for future work stemming from these approaches are:
\begin{itemize}
\item To explore the connections between our first approach, i.e., $\QUnital$, and Refs.~\cite{hardy2021time} and \cite{di2020quantum}. In particular, these connections should be formalised process-theoretically so that they are straightforward to identify.
\item To explore whether there are physical consequences to our second approach, i.e., $\QPart$, in particular, whether it can be generalised to include physically relevant infinite-dimensional representations, such as those of the Poincar\'e group.
\end{itemize}

Another interesting question is this: if $\QPart$ were a faithful representation of nature, how could we distinguish it from $\QPhys$? Is there some experimental test that could shed light on which is the correct theory?

If we were agents composed only of causal systems, the non-signalling conditions would always prevent us from directly probing the retrocausal systems. That is, some mechanism akin to spontaneous symmetry breaking can be seen to occur -- we have a symmetric theory, but any instantiation of the theory will necessarily appear asymmetric!

 Direct observation may therefore be impossible; however, there remains the possibility that the existence of these retrocausal counterparts could be probed more indirectly.  An analogy can be drawn with spontaneous collapse theories -- the spontaneous collapse cannot be directly observed, but its presence could manifest as an additional decoherence term which has no alternative explanation  \cite{goldwater2016testing}. Similarly, if retrocausal systems were interacting with us, then their presence could be detected by otherwise inexplicable decoherence in experiments. 

A more direct answer could come from the study of higher-order processes. The framework of higher-order processes can be viewed as another attempt to modify the notion of causality and causal structure, in particular, with the motivation of studying (quantumly) indefinite causal structure. The main tools in this study are higher-order processes, particularly those known as \emph{process matrices} \cite{oreshkov2012quantum}. These are defined as maps, $W$, from pairs of channels to channels: 
\beq
\input{figures/PM.tikz} :: \left(\input{figures/PMNew2.tikz},\input{figures/PMNew3.tikz}\right) \mapsto \input{figures/PMNew1.tikz}. 
\eeq
We know that quantum circuits can only realise a subset of logically possible process matrices, that is, those that can be written as \cite{oreshkov2012quantum}:
\beq
\input{figures/PM.tikz}\quad = \quad \input{figures/PM1.tikz}.
\eeq

Now suppose that the theory also contains retrocausal systems. This would provide access to circuit realisable process matrices of the form:
\beq
\input{figures/PM.tikz}\quad = \quad \input{figures/PM2.tikz}.
\eeq
This leads us to our first open problem:
\begin{problem}
Do there exist process matrices that are circuit realisable in $\QPart$ but not in $\QPhys$?
\end{problem}

\section{Two equivalent approaches to time neutrality}\label{sec:2TN}

The above approaches may lead to interesting time symmetric theories. However, they do not address time neutrality. The only time neutral theory that we have encountered so far is the process theory $\QCalc$. Nevertheless, we noted that this is not a good physical theory as it does not always make valid probabilistic predictions. This section, therefore, asks whether it is possible to adapt $\QCalc$ such that it remains time neutral and at the same time makes sensible probabilistic predictions.

One candidate approach would be to restrict to trace-nonincreasing maps, as this has more processes than $\QPhys$ (which assumed trace preservation) but fewer than $\QCalc$ (which does not impose a trace condition at all). Although this may seem promising, it quickly fails, as for any (non-trivial) system the set of subnormalised states is not isomorphic to the set of effects in the theory. See Fig.~\ref{fig:StateEffectBit} for a simple example in the case of a classical bit.
\begin{figure}
     \centering
     \begin{subfigure}[t]{0.3\textwidth}
         \centering
         \ctikzfig{figures/bitStateEffectNormalised}
         \caption{The trace-preserving states (red line segment) and effects (blue point) for a bit. There is a clear asymmetry here as we can see that there is a unique effect, corresponding to marginalisation, whilst there are many states, corresponding to probability distributions over a two-element set.}
     \end{subfigure}
     \hfill
     \begin{subfigure}[t]{0.3\textwidth}
         \centering
         \ctikzfig{figures/bitStateEffectSpace}
         \caption{The trace-nonincreasing states (red triangle) and effects (blue square) for a bit. We can again see that there is an asymmetry, this time with more effects than states.}
     \end{subfigure}
     \hfill
     \begin{subfigure}[t]{0.3\textwidth}
         \centering
         \ctikzfig{figures/bitStateEffectCone}
         \caption{Without any trace condition, the states and effects both form the same convex cone (grey octant), so there is no obstacle to time symmetry here. However, we can now obtain `probabilities' greater than $1$ by composing states and effects and so this theory does not make sensible predictions.}
     \end{subfigure}
        \caption{The states and effects for a classical bit.}
        \label{fig:StateEffectBit}
\end{figure}

This has led to various other approaches in the literature that have, directly or indirectly, proposed time-neutral formulations of quantum theory. For example, see Refs.~\cite{aharonov2008two,
aharonov2009multiple,
oeckl2008general,
oreshkov2015operational,
silva2017connecting,
silva2014pre,oreshkov2016operational,oeckl2016local}. Whilst there are many philosophical differences between them, all of these are -- up to choices of convention -- captured by the formalism we present below. 

These approaches are typically presented as a modification of the measurement postulate of quantum theory.  They begin by extending the set of measurements to allow for measurements, $M$, in which the POVM elements do not necessarily sum to the identity:
\beq
M = \{ \{M_a\}_{a\in\mathds{A}} | M_a \geq 0  \}.
\eeq
They then ensure that the probabilistic predictions are sensible via a suitable modification of the Born rule. That is, they are defined so that, if the ``probability distribution'' over measurement outcomes predicted by the standard quantum formalism is not normalised, i.e., if
 \beq
 N_M(\rho) := \sum_{a\in \mathds{A}} \mathsf{tr}( M_a \rho) \neq 1,
 \eeq
  then this distribution is renormalised by dividing by $N_M(\rho)$. This gives a new rule for computing probabilities:
  \beq
  \mathrm{Prob}(a|\rho) := \frac{1}{N_M(\rho)}\mathsf{tr}(M_a \rho)
  \eeq
  which differs from the standard Born rule whenever $N_M(\rho)\neq 1$.
This new probability rule ensures that the theory makes valid probabilistic predictions. (At least, aside from the case of $N_M(\rho)=0$ which is treated as a special case which we return to later.)

Simply modifying the probability rule is not, however, particularly satisfying from a process-theoretic perspective, since it is not manifestly compositional. What we show in this section is that we can recover this modified Born rule in a manifestly compositional way. We do so via the construction of two (equivalent) new process theories which naturally capture this modified Born rule in the appropriate situations.

The first process theory that we construct to achieve this can be defined by starting from $\QCalc$ and modifying composition so that it reduces to the modified Born rule when a state is composed with a measurement. The second process theory that we construct achieves this in a simpler, more elegant, and more adaptable way by defining an appropriate quotient of $\QCalc$. We will then show how this process-theoretic viewpoint naturally leads to a new way to deal with probability-zero events, leading to a deterministic time neutral theory. Finally, we investigate the relationship between this theory and higher-order quantum processes.

\subsection{Modified composition rule}

Within $\QCalc$ the modified Born rule can be presented as:
\beq
{\color{gray} \text{Standard rule}}\qquad\input{figures/ModifiedBorn1.tikz}\quad   \mapsto \quad \left( \input{figures/ModifiedBorn2.tikz}\right)^{-1} \input{figures/ModifiedBorn1.tikz} \qquad {\color{gray} \text{Modified rule}}
\eeq
for some CP map $M$ from a quantum to a classical system and some quantum state $\rho$. Provided that this normalisation factor is non-zero, this defines a valid probability distribution over $\mathds{A}$. 

As mentioned above, we want to define a process theory with the same processes and systems as $\QCalc$ but in which composition is redefined. In particular, it needs to be such that when we compose a state $\rho$ with a measurement $M$ the renormalisation ``automatically'' occurs. More formally, we want to define a new composition rule $\bullet$ such that
\beq
M\bullet \rho := \frac{1}{\mathsf{tr}(M\circ \rho)} M\circ \rho .
\eeq
Here, $\circ$ is the composition rule in $\QCalc$.
We will therefore denote the new process theory by $\QCalc^\bullet$.

To obtain a consistent process theory, however, we cannot simply redefine composition for these special cases. Instead, we must redefine composition in general, and obtain the above as a particular instance of the new rule. Explicitly, we define $\bullet$ via:
\beq\label{eq:bulletComp}
\input{figures/ModifiedBorn4.tikz} \bullet \input{figures/ModifiedBorn3.tikz} := \left\{ \begin{array}{cc} \left( \input{figures/ModifiedBorn6.tikz}\right)^{-1} \input{figures/ModifiedBorn5.tikz}  & \text{if} \quad\input{figures/ModifiedBorn6.tikz} \neq 0 \\ \ & \ \\
\mathbf{0} & \text{otherwise,} \end{array} \right.
\eeq
where now we are explicitly dealing with the probability-zero case.
This is the same rule introduced in Refs.~\cite{oreshkov2015operational,oreshkov2016operational} and used in Ref.~\cite{pinzani2019timetravel} in the context of the P-CTC model for closed timelike curves.
As an instance of this rule, we have:
\beq
\input{figures/ModifiedBorn8.tikz} \bullet \begin{tikzpicture}
	\begin{pgfonlayer}{nodelayer}
		\node [style=none] (9) at (0, 0.75) {};
		\node [style=point] (10) at (0, -0.75) {$\rho$};
		\node [style=right label] (13) at (0, 0) {$\mathcal{H}$};
	\end{pgfonlayer}
	\begin{pgfonlayer}{edgelayer}
		\draw [QWire] (9.center) to (10);
	\end{pgfonlayer}
\end{tikzpicture}
 := \left\{ \begin{array}{cc} \left( \input{figures/ModifiedBorn2.tikz}\right)^{-1} \input{figures/ModifiedBorn1.tikz}  & \text{if}\quad\input{figures/ModifiedBorn2.tikz} \neq 0 \\ \ & \ \\
\mathbf{0} & \text{otherwise,} \end{array} \right.
\eeq
which reproduces the probabilities that we expect from the modified Born rule used in these time-neutral theories.

For this to define a valid process theory, various conditions must be satisfied. In particular, one can show by direct computation that $\bullet$ is associative and that it interacts suitably with parallel composition.
The more interesting case, however, comes from considering the identity processes. In any process theory it must be the case that $\mathds{1}_B \circ f = f = f\circ \mathds{1}_A$ for every process $f:A\to B$. In our case, however, if we try to impose this condition with $\bullet$ we find that:
\beq
\input{figures/ModifiedBorn3.tikz}\quad =\quad \input{figures/ModifiedBorn10.tikz} \bullet \input{figures/ModifiedBorn3.tikz}\quad =\quad \left\{ \begin{array}{cc} \left( \input{figures/ModifiedBorn11.tikz}\right)^{-1} \input{figures/ModifiedBorn3.tikz}  & \text{if} \quad\input{figures/ModifiedBorn11.tikz} \neq 0 \\ \ & \ \\
\mathbf{0} & \text{otherwise,} \end{array} \right.
\eeq
which only holds in the special cases that
\beq\label{eq:repElement}
\input{figures/ModifiedBorn11.tikz} = 1 \qquad  \text{or} \qquad \input{figures/ModifiedBorn3.tikz} = \mathbf{0}.
\eeq
Therefore, to define the process theory $\QCalc^{\bullet}$, we must both modify the composition rule \emph{and} restrict the set of allowed processes to the above special cases. For example, for a classical bit, we obtain the states and effects as shown in Fig.~\ref{fig:bulletStateEffect}.
\begin{figure}
\[\input{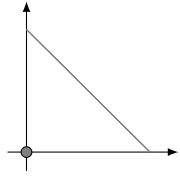} \]
\caption{The states and effects for a classical bit in $\QCalc^\bullet$ are the same and therefore do not present an obstruction to time symmetry in this theory. They include the origin and the space of probability distributions. Note that this means that the state and effect spaces are non-convex.}
\label{fig:bulletStateEffect}
\end{figure}
\begin{definition}[$\QCalc^\bullet$] This is the process theory which has the same objects as $\QCalc$, and, as processes, the subset of processes in $\QCalc$ satisfying Eq.~\eqref{eq:repElement}. Sequential composition is defined as in Eq.~\eqref{eq:bulletComp} whilst parallel composition is as in $\QCalc$.
\end{definition}

Having modified $\QCalc$ in such a way as to obtain $\QCalc^{\bullet}$, it is no longer immediately clear that time neutrality or even time symmetry of $\QCalc$ have been preserved. It turns out that this is the case, but we will return to this in the following section once we have the more elegant characterisation of this process theory. 

\subsection{Modified processes}

The construction of $\QCalc^\bullet$ within the previous section via a modified composition rule (and restriction of the allowed processes) can be more elegantly captured by a particular quotienting of $\QCalc$. Specifically, we define a new theory in which the processes correspond to particular equivalence classes of processes in $\QCalc$.

Specifically, we define two processes to be equivalent if they are equal up to a non-zero scalar:
\beq\label{eq:equivClass}
\input{figures/cqProcessf.tikz}\sim \input{figures/cqProcessg.tikz} \quad \iff \quad \exists r>0 \text{ s.t. } \input{figures/cqProcessf.tikz} = r\ \input{figures/cqProcessg.tikz}.
\eeq
We denote the equivalence class of a process $\mathcal{E}$ by $\tilde{\mathcal{E}}$. It is then straightforward to note that we can construct the quotient process theory $\CPMTN$. This has the same systems as $\QCalc$ but has processes that correspond to equivalence classes of processes in $\QCalc$ under the above equivalence relation. Composition within $\CPMTN$ can then be defined, for instance, as
\beq\label{eq:compEquiv}
\input{figures/composingEquivalenceClasses2.tikz}\quad :=\quad \widetilde{\input{figures/composingEquivalenceClasses.tikz}}
\eeq 
where $\mathcal{E}$ is an arbitrary element of $\tilde{\mathcal{E}}$ and $\mathcal{F}$ is an arbitrary element of $\tilde{\mathcal{F}}$. It is straightforward to verify that this is well defined as this definition is independent of the choice of these elements.

Note that the scalars of $\CPMTN$ are severely restricted: the scalars in $\QCalc$ are $\mathbb{R}^+$ whilst in $\CPMTN$ they are equivalent to $\mathbb{Z}_2$ since we have two equivalence classes $\{\tilde{0},\tilde{1}\}$, where $\tilde{0}=\{0\}$ and $\tilde{1}=(0,\infty)$. This may at first glance seem problematic as we still want our time symmetric theory to make probabilistic predictions, which within $\QCalc$ are encoded in the scalars. The resolution to this is (as we saw with $\QPhys$ which is a causal theory and hence has only a single scalar) that the probabilistic predictions of a theory are encoded into classical states
\beq\input{figures/probabilityDistribution.tikz}\eeq
rather than into individual scalars. We can therefore ask what predictions can be made by the theory $\CPMTN$, i.e., what are the processes of the form:
\beq\input{figures/probabilityDistributionEquiv.tikz}\ \ ?\eeq
It is not hard to see that these will be in one-to-one correspondence with probability distributions with one extra classical state left over, namely, the zero-state:
\beq\input{figures/zeroDist.tikz}\ .\eeq
See Fig.~\ref{fig:quotientStates} for the example where $\mathds{X}$ is a two-element set.
This theory, therefore, makes predictions that can be interpreted probabilistically most of the time: We simply view the equivalence class containing a probability distribution $p$ as describing the same prediction as the probability distribution itself.

\begin{figure}
\[\input{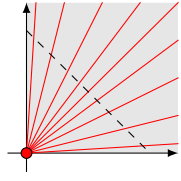}\]
\caption{A schematic depiction of the states in $\CPMTN$ for a classical bit. The red dot at the origin represents the zero state whilst the red rays going out (which do not include the origin) are examples of equivalence classes. The union of these equivalence classes fills the grey shaded cone. The dashed line indicates normalised probability distributions. Note that each of the rays intersects the dashed line at a single unique point, hence can be unambiguously identified with this probability distribution.}
\label{fig:quotientStates}
\end{figure}

The exception to this interpretation is the zero state for which there is no obvious interpretation as a probability distribution. If we want the process theory to represent physical processes, then this serves as a challenge as it is difficult to give a physical interpretation of a zero-process, which, by definition, is something that does not occur. In other words, the existence of the zero state means that the process theory is not deterministic. We return to this shortly. Before that, however, we will show that this theory is time-symmetric and time-neutral:

\begin{proposition}\label{prop:TimeSym} The Hermitian adjoint $\dagger_H$ is a dagger for $\CPMTN$.
\end{proposition}
\proof It is simple to see that $\dagger_H$ preserves equivalence classes as $\dagger_H(r \mathcal{F})=r\dagger_H(\mathcal{F})$ for all processes $\mathcal{F}$ and scalars $r$.
\endproof

\begin{proposition}\label{prop:TimeNeut} $\CPMTN$ is time neutral.
\end{proposition}
\proof
The equivalence class containing the cup and the equivalence class containing the cap will define cups \& caps for $\CPMTN$. These satisfy Eq.~\eqref{eq:snake}, as follows immediately from the definition of composition of equivalence classes in Eq.~\eqref{eq:compEquiv}.
\endproof

\subsection{Equivalence of these approaches}

We now compare the process theory $\CPMTN$ (that we have defined via quotienting) to the process theory $\QCalc^{\bullet}$ (that we defined in the previous subsection via a modification of the composition rule). In particular, we show that these are equivalent process theories.

To begin, note that any process in $\QCalc$ is in the same equivalence class as a process satisfying one of the conditions in Eq.~\eqref{eq:repElement}, and, moreover, that each equivalence class will contain a unique such element.
The crux of this is the following property of processes in $\QCalc$,
\beq
\input{figures/ModifiedBorn11.tikz} \ \ = \ 0 \quad \iff \quad \input{figures/cqProcessf.tikz} \ \ = \ 0,
\eeq 
 which follows from the fact that $\QCalc$ is locally tomographic, and that the trace and maximally mixed states are internal to the cone of effects and states respectively:
\begin{align}
\input{figures/ModifiedBorn11.tikz} \ \ = \ 0 \quad &\implies \quad \input{figures/MBProof1.tikz} \ \ + \cdots + \ \ \input{figures/MBProof2.tikz}\ \ = \ 0\ \ \forall \sigma, r, \rho, p \\
&\implies \quad \input{figures/MBProof1.tikz}\ \  = \ 0\ \  \forall \sigma, r, \rho, p \\
&\implies \quad \input{figures/cqProcessf.tikz} \ \ = \ 0,
\end{align}
where in the first step we use that the maximally mixed state and the discarding effect are internal, in the second step we use that if a sum of non-negative terms is zero then each term is zero, and in the final step we use that $\QCalc$ is locally tomographic.
  
The processes in $\QCalc^{\bullet}$ can therefore be viewed as a particular conventional choice of representative elements for the equivalence classes. The modified composition rule can then be derived by: i) replacing each process with its equivalence class, ii) composing the equivalence classes, and iii) choosing the representative element for the equivalence class of the composite. 

What is therefore clear from this is that:
\begin{theorem}\label{thm:equiv}
 These two time neutral theories are in fact equivalent to one another, 
\beq
\CPMTN\  \cong\  \QCalc^{\bullet}.
\eeq
  Specifically,  $\QCalc^\bullet$ is simply a way to describe $\CPMTN$ using representative elements of the equivalence classes, rather than just working with the equivalence classes themselves.
\end{theorem}
 The inelegant nature of $\QCalc^{\bullet}$ can be viewed as a consequence of the somewhat arbitrary way in which the representative elements are chosen. Picking some other convention for how to pick a representative element would lead to a distinct composition rule (and hence a different probability rule) but would ultimately be describing the same process theory.
 
 This equivalence, together with Props.~\ref{prop:TimeSym} and \ref{prop:TimeNeut}, immediately tells us that $\QCalc^\bullet$ is also time symmetric and time neutral as we claimed earlier.

Note that, in practice, if we want to compute anything within $\CPMTN$ it will be much more convenient to work with representative elements rather than the equivalence classes themselves. In this case, one may be tempted to work within $\QCalc^\bullet$; however, there is no real benefit to choosing these particular representative elements. What is usually simplest is i) to pick \emph{arbitrary} representatives from each equivalence class, ii) compose them within $\QCalc$, and iii) quotient at the end to obtain the equivalence class again. 

\subsection{Restoring determinism} \label{sec:determinism}

As mentioned earlier, in $\CPMTN$ (or equivalently $\QCalc^{\bullet}$) we have a problem with determinism  -- we have zero-processes that describe things that cannot occur. In other approaches this has been dealt with in a fairly arbitrary way by simply stating that when $N_M(\rho) = 0$ then $\mathrm{Prob}(a|\rho)=0$ for all $a \in \mathds{A}$. However, how are we to operationally understand a measurement in which all of the possible outcomes occur with probability zero? 

This manifests in our approach by the fact that we have a pair of scalars in $\CPMTN$ rather than the single scalar that we would expect in a deterministic theory (Def.~\ref{def:determinism}). Furthermore, our constructions have a somewhat arbitrary nature to them: i) the modified composition rule has to have two cases depending on whether or not a zero appears, and ii) when we quotient, we do so with respect to non-zero scalars rather than arbitrary scalars.

Here we present an alternative solution that will result in a deterministic theory in which the zero-cases naturally do not arise. We do so by moving away from $\QCalc$ as the starting point for our constructions.

To begin, we observe that in any real-world experiment, we never manage to completely suppress all sources of noise. That is, in the lab we never actually prepare a pure state or perform a projective measurement (at least not on the system of interest). More generally, we can view the experimentally realisable quantum processes as those of the form:
\beq\label{eq:noisyprocess}
\input{figures/noisyProcess1.tikz} \quad := \quad \input{figures/noisyProcess.tikz} \qquad  0 < \epsilon \leq 1 
\eeq
that is, processes that have some non-zero epsilon of noise. Indeed, we can define a restriction of $\QCalc$ denoted $\QCalc|_{noise}$ which has only processes of this form. It is straightforward to verify that this indeed defines a valid process theory as these are closed under composition.

 $\QCalc|_{noise}$ and $\QCalc$, at least from an operational point of view, are equivalent to one another. No real-world experiment could distinguish between some process $\mathcal{E}$ and the noisy version $\mathcal{E}_\epsilon$ provided that $\epsilon$ is small enough -- we can approximate $\mathcal{E}$ arbitrarily well by using processes $\mathcal{E}_\epsilon$ and taking the limit of $\epsilon \to 0$. In some sense then, choosing between $\QCalc$  (as is standard in the literature) and $\QCalc|_{noise}$ is purely a matter of convenience.
However, when we move over to the respective quotiented theories, $\CPMTN$ and $\QCalc|_{noise}/\!\!\sim$, we obtain strikingly different theories. This is highlighted by considering the scalars of the theory. Note that the scalars in $\QCalc$ are $[0,\infty)$ whilst in $\QCalc|_{noise}$ they are $(0,\infty)$. As noted before, this means that in $\CPMTN$ we have two scalars $\{\tilde{0},\tilde{1}\}$, whereas in $\QCalc|_{noise}/\!\!\sim$ we have only the scalar $\tilde{1}$ (as $0$ is not a scalar in $\QCalc|_{noise}$). This means that $\QCalc|_{noise}/\!\!\sim$ is a deterministic process theory, and so every process has a valid operational interpretation. To see this more clearly, note that the states for a classical system are now only those of the form:
\beq\input{figures/probabilityDistributionEquiv.tikz}\eeq
as the zero state,
\beq\input{figures/zeroDist.tikz},\eeq
is no longer part of the theory. This means that \emph{every} classical state can be thought of as a probability distribution. For the case of a classical bit this is illustrated in Fig.~\ref{fig:noisyQuotient}.
\begin{figure}
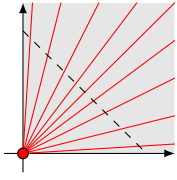
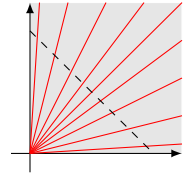

	\centering
     \begin{subfigure}[b]{0.45\textwidth}
         \centering
         \ctikzfig{figures/TNStateEffect} 
         \caption{The space of states for $\CPMTN$}
     \end{subfigure}
     \hfill
     \begin{subfigure}[b]{0.45\textwidth}
         \centering
         \ctikzfig{figures/TNStateEffectNoise}
         \caption{The space of states for $\QCalc|_{noise}/\!\!\sim$}
     \end{subfigure}
      \caption{Note that the only difference is whether or not the zero state appears. Every state (i.e., ray) on the right-hand side can be uniquely identified with the probability distribution at the point where it intersects the dashed line.}
        \label{fig:noisyQuotient}
\end{figure}

This theory remains time symmetric as it obtains a dagger from the Hermitian adjoint of $\QCalc$:
\begin{proposition} $\QCalc|_{noise}/\!\!\sim$ is time symmetric.
\end{proposition}

As per our current definition, however, it is not time neutral as it does not have cups \& caps, nor does it even have identities or swaps! It is simple to see this as the cups \& caps, identities, and swaps in $\QCalc$ are not processes of the form of Eq.~\eqref{eq:noisyprocess}.
\begin{remark}
As discussed above, in $\QCalc|_{noise}$ we do not have  identities, swaps, and other processes which we could construct from them. From a categorical point of view, this would mean that we do not have identity morphisms and symmetries.  This highlights a subtle difference between process theories (as we define them here) and SMCs which previously have been taken to be equivalent. We conjecture, however, that these identities and swaps can be freely reintroduced without changing the theory's semantics. If true, this would restore the equivalence between process theories and SMCs.
\end{remark}
 In this case, we can freely add all of these \emph{wiring processes} back into the theory without changing anything that we have discussed so far.
 \begin{definition}[$\QSym$]\label{def:QNeut}
 We define $\QSym$ as the process theory generated by 
 \beq
 \left\{\input{noisyProcess1New.tikz},\ \  \input{wirings.tikz} \right\}
 \eeq
 and quotiented by arbitrary scalars.
 \end{definition}
Note that we do not need to restrict to quotienting by scalars $r>0$ because, by construction, all scalars will satisfy this condition as we see in the following.
\begin{theorem}\label{thm:QNeut}
$\QSym$ is a deterministic and time neutral process theory.
\end{theorem}
\proof
Time neutrality follows immediately from the existence of the cups \& caps that we just added. We can then note that the set of noisy processes (i.e., those of the form in Eq.~\eqref{eq:noisyprocess}) is closed under composition with wiring processes, e.g.:
\begin{align}
\input{figures/noisyProcessCup.tikz}\quad &= \quad \input{figures/noisyProcessCup2.tikz} \\
 &= \quad \input{figures/noisyProcessCup3.tikz}.
\end{align}
If we have some generic diagram, we can therefore always absorb these wiring processes into noisy processes unless the wiring processes are disconnected, as in:
\beq
\input{figures/noisyProcess1.tikz}\quad\input{figures/wiring.tikz}.
\eeq
A general scalar (i.e., closed diagram) can therefore be written as some scalar $r_\epsilon$ in $\QCalc|_{noise}$ composed in parallel with some closed wiring, for example:
\beq
\begin{tikzpicture}
	\begin{pgfonlayer}{nodelayer}
		\node [style=scalar] (0) at (0, 0) {$r_\epsilon$};
	\end{pgfonlayer}
\end{tikzpicture}
\quad \input{figures/closedWiring.tikz}
\eeq
One can then note that such closed wiring is simply the parallel composition of closed loops and, because each such closed loop equals the dimension of its system, the result is a strictly positive number $w$. Hence, $r_\epsilon w > 0$ and, when we quotient, we obtain the scalar $\tilde{1}$. We have therefore shown that there is a unique scalar in the theory.
\endproof

\subsection{Connections to higher-order processes}\label{sec:HigherOrder}

 Recent work has connected the study of process matrices \cite{oreshkov2012quantum}  (briefly introduced in section~\ref{sec:ProcessMatrixDiscussion})  to the time-symmetric formulation of quantum theory \cite{silva2017connecting} and the time-neutral formulation of quantum theory \cite{oreshkov2016operational}. Here we re-derive part of this in a particularly simple way, using the cups and caps in our time-neutral formalism (Def.~\ref{def:CupsAndCaps}), as instantiated by $\QSym$ (Def.~\ref{def:QNeut}).

\begin{proposition}
 Any time-neutral quantum theory, in the sense of Def.~\ref{def:CupsAndCaps}, that contains all CPTP maps also contains all process matrices \cite{oreshkov2016operational,silva2017connecting}.
\end{proposition}
\begin{proof}
A process matrix is a higher-order linear map, $W$, that acts on a pair of CPTP maps, $A$ and $B$, to give a new CPTP map:
\beq\input{figures/ProcessMatrixDefinition.tikz}\eeq
The question we want to answer is whether such a map can be implemented in our time-neutral formalism. To see that this is indeed possible we will use the cups \& caps that we introduced earlier. Consider the special case in which $A$ and $B$ are both swaps. The definition of the process matrices implies that this must be a CPTP map which we will denote $\mathcal{W}$:
\beq\input{figures/ProcessMatrixSwap.tikz}\eeq
 Since all CPTP maps are in the theory by assumption, $\mathcal{W}$ is also in the theory. It is now simple to show, using cups \& caps, that this means that $W$ itself is in our formalism:
\beq\input{figures/ProcessMatrixInTimeSymmetric.tikz}\eeq 
The LHS is a valid diagram in our theory because it is built from $\mathcal{W}$, which is a CPTP map, and from cups and caps, all of which are valid processes. The first equality follows from the definition of $\mathcal{W}$, and the second from the definition of cups and caps. Hence, the RHS is also valid in our formalism.
Strictly speaking, these kinds of ``diagrams with holes'' are not part of the process-theory framework, so we should not say that $W$ is literally ``in our formalism''. Nevertheless, the action of $W$ on arbitrary CPTP maps is represented within the formalism, since we can use the above idea to write:
\beq
\begin{tikzpicture}
	\begin{pgfonlayer}{nodelayer}
		\node [style=none] (0) at (4.25, 0) {$W$};
		\node [style=none] (1) at (3.5, 1) {};
		\node [style=none] (2) at (3.5, -1) {};
		\node [style=none] (3) at (5, -1) {};
		\node [style=none] (4) at (5, 1) {};
		\node [style=none] (5) at (6.5, 1) {};
		\node [style=none] (6) at (6.5, 1.75) {};
		\node [style=none] (7) at (2, 1.75) {};
		\node [style=none] (8) at (2, 1) {};
		\node [style=none] (9) at (2, -1) {};
		\node [style=none] (10) at (2, -1.75) {};
		\node [style=none] (11) at (6.5, -1.75) {};
		\node [style=none] (12) at (6.5, -1) {};
		\node [style=none] (13) at (2, 0) {$A$};
		\node [style=none] (14) at (6.5, 0) {$B$};
		\node [style=none] (15) at (3, 0.5) {};
		\node [style=none] (16) at (3, -0.5) {};
		\node [style=none] (17) at (1, -0.5) {};
		\node [style=none] (18) at (1, 0.5) {};
		\node [style=none] (19) at (5.5, 0.5) {};
		\node [style=none] (20) at (5.5, -0.5) {};
		\node [style=none] (21) at (7.5, -0.5) {};
		\node [style=none] (22) at (7.5, 0.5) {};
		\node [style=none] (23) at (1.5, 0.5) {};
		\node [style=none] (24) at (2.5, 0.5) {};
		\node [style=none] (25) at (2.5, 1) {};
		\node [style=none] (26) at (2.5, -0.5) {};
		\node [style=none] (27) at (2.5, -1) {};
		\node [style=none] (28) at (1.5, -0.5) {};
		\node [style=none] (29) at (1.5, 2.75) {};
		\node [style=none] (30) at (1.5, -2.75) {};
		\node [style=none] (31) at (6, 1) {};
		\node [style=none] (32) at (6, 0.5) {};
		\node [style=none] (33) at (6, -0.5) {};
		\node [style=none] (34) at (6, -1) {};
		\node [style=none] (35) at (7, 0.5) {};
		\node [style=none] (36) at (7, 2.75) {};
		\node [style=none] (37) at (7, -0.5) {};
		\node [style=none] (38) at (7, -2.75) {};
		\node [style=none] (39) at (0.5, 2.25) {};
		\node [style=none] (40) at (0.5, -2.25) {};
		\node [style=none] (41) at (8, -2.25) {};
		\node [style=none] (42) at (8, 2.25) {};
		\node [style=none] (43) at (4.25, 1.75) {};
		\node [style=none] (44) at (4.25, -1.75) {};
		\node [style=none] (45) at (4.25, 2.75) {};
		\node [style=none] (46) at (4.25, -2.75) {};
		\node [style=none] (80) at (-7.5, 0.5) {};
		\node [style=none] (81) at (-5.5, 0.5) {};
		\node [style=none] (82) at (-5.5, -0.5) {};
		\node [style=none] (83) at (-7.5, -0.5) {};
		\node [style=none] (84) at (-6.5, 0) {$\mathcal{W}$};
		\node [style=none] (86) at (-7, 0.5) {};
		\node [style=none] (87) at (-7, 0.75) {};
		\node [style=none] (88) at (-7, -1) {};
		\node [style=none] (89) at (-7, -0.5) {};
		\node [style=none] (90) at (-8, 0.75) {};
		\node [style=none] (91) at (-8, -1) {};
		\node [style=none] (92) at (-9, -1.5) {};
		\node [style=none] (93) at (-9, 1.5) {};
		\node [style=none] (94) at (-10, 1.5) {};
		\node [style=none] (95) at (-10.25, 0.5) {};
		\node [style=none] (96) at (-10, -1.5) {};
		\node [style=none] (97) at (-10.25, -0.5) {};
		\node [style=none] (98) at (-4, 1.5) {};
		\node [style=none] (99) at (-6, 0.75) {};
		\node [style=none] (100) at (-5, -1) {};
		\node [style=none] (101) at (-6, -0.5) {};
		\node [style=none] (102) at (-6, 0.5) {};
		\node [style=none] (103) at (-2.75, 0.5) {};
		\node [style=none] (104) at (-3, -1.5) {};
		\node [style=none] (105) at (-6, -1) {};
		\node [style=none] (106) at (-5, 0.75) {};
		\node [style=none] (107) at (-3, 1.5) {};
		\node [style=none] (108) at (-2.75, -0.5) {};
		\node [style=none] (109) at (-4, -1.5) {};
		\node [style=none] (156) at (-10.5, 1.5) {};
		\node [style=none] (157) at (-10.5, 2.5) {};
		\node [style=none] (158) at (-2.5, 2.5) {};
		\node [style=none] (159) at (-2.5, 1.5) {};
		\node [style=none] (160) at (-3.5, 1.5) {};
		\node [style=none] (161) at (-3.5, -1.5) {};
		\node [style=none] (162) at (-2.5, -1.5) {};
		\node [style=none] (163) at (-2.5, -2.5) {};
		\node [style=none] (164) at (-10.5, -2.5) {};
		\node [style=none] (165) at (-10.5, -1.5) {};
		\node [style=none] (166) at (-9.5, -1.5) {};
		\node [style=none] (167) at (-9.5, 1.5) {};
		\node [style=none] (176) at (-6.5, 0.5) {};
		\node [style=none] (177) at (-6.5, -0.5) {};
		\node [style=none] (178) at (-6.5, 3) {};
		\node [style=none] (179) at (-6.5, -3) {};
		\node [style=none] (193) at (-10.75, 0) {$A$};
		\node [style=none] (194) at (-2.25, 0) {$B$};
		\node [style=none] (195) at (-9.75, 0.5) {};
		\node [style=none] (196) at (-9.75, -0.5) {};
		\node [style=none] (197) at (-11.75, -0.5) {};
		\node [style=none] (198) at (-11.75, 0.5) {};
		\node [style=none] (199) at (-3.25, 0.5) {};
		\node [style=none] (200) at (-3.25, -0.5) {};
		\node [style=none] (201) at (-1.25, -0.5) {};
		\node [style=none] (202) at (-1.25, 0.5) {};
		\node [style=none] (203) at (-11.25, 0.5) {};
		\node [style=none] (208) at (-11.25, -0.5) {};
		\node [style=none] (209) at (-11.25, 2.75) {};
		\node [style=none] (210) at (-11.25, -2.75) {};
		\node [style=none] (215) at (-1.75, 0.5) {};
		\node [style=none] (216) at (-1.75, 2.75) {};
		\node [style=none] (217) at (-1.75, -0.5) {};
		\node [style=none] (218) at (-1.75, -2.75) {};
		\node [style=none] (219) at (-17.5, 0) {$A$};
		\node [style=none] (220) at (-14.5, 0) {$B$};
		\node [style=none] (221) at (-16.5, 0.5) {};
		\node [style=none] (222) at (-16.5, -0.5) {};
		\node [style=none] (223) at (-18.5, -0.5) {};
		\node [style=none] (224) at (-18.5, 0.5) {};
		\node [style=none] (225) at (-15.5, 0.5) {};
		\node [style=none] (226) at (-15.5, -0.5) {};
		\node [style=none] (227) at (-13.5, -0.5) {};
		\node [style=none] (228) at (-13.5, 0.5) {};
		\node [style=none] (229) at (-18, 0.5) {};
		\node [style=none] (230) at (-17, 0.5) {};
		\node [style=none] (231) at (-17, 1) {};
		\node [style=none] (232) at (-17, -0.5) {};
		\node [style=none] (233) at (-17, -1) {};
		\node [style=none] (234) at (-18, -0.5) {};
		\node [style=none] (235) at (-18, 1) {};
		\node [style=none] (236) at (-18, -1) {};
		\node [style=none] (237) at (-15, 1) {};
		\node [style=none] (238) at (-15, 0.5) {};
		\node [style=none] (239) at (-15, -0.5) {};
		\node [style=none] (240) at (-15, -1) {};
		\node [style=none] (241) at (-14, 0.5) {};
		\node [style=none] (242) at (-14, 1) {};
		\node [style=none] (243) at (-14, -0.5) {};
		\node [style=none] (244) at (-14, -1) {};
		\node [style=none] (245) at (-16, 0) {$,$};
		\node [style=none] (246) at (-12.75, 0) {$\mapsto$};
		\node [style=none] (247) at (-0.5, 0) {$=$};
	\end{pgfonlayer}
	\begin{pgfonlayer}{edgelayer}
		\draw (7.center) to (6.center);
		\draw (6.center) to (5.center);
		\draw (5.center) to (4.center);
		\draw (4.center) to (3.center);
		\draw (3.center) to (12.center);
		\draw (12.center) to (11.center);
		\draw (11.center) to (10.center);
		\draw (10.center) to (9.center);
		\draw (9.center) to (2.center);
		\draw (2.center) to (1.center);
		\draw (1.center) to (8.center);
		\draw (8.center) to (7.center);
		\draw (18.center) to (15.center);
		\draw (15.center) to (16.center);
		\draw (16.center) to (17.center);
		\draw (17.center) to (18.center);
		\draw (19.center) to (20.center);
		\draw (20.center) to (21.center);
		\draw (21.center) to (22.center);
		\draw (22.center) to (19.center);
		\draw [style=QWire] (30.center) to (28.center);
		\draw [style=QWire] (23.center) to (29.center);
		\draw [style=QWire] (27.center) to (26.center);
		\draw [style=QWire] (24.center) to (25.center);
		\draw [style=QWire] (34.center) to (33.center);
		\draw [style=QWire] (32.center) to (31.center);
		\draw [style=QWire] (35.center) to (36.center);
		\draw [style=QWire] (38.center) to (37.center);
		\draw [style=thick gray dashed edge] (39.center) to (42.center);
		\draw [style=thick gray dashed edge] (42.center) to (41.center);
		\draw [style=thick gray dashed edge] (39.center) to (40.center);
		\draw [style=thick gray dashed edge] (40.center) to (41.center);
		\draw [style=QWire] (46.center) to (44.center);
		\draw [style=QWire] (43.center) to (45.center);
		\draw (80.center) to (81.center);
		\draw (81.center) to (82.center);
		\draw (82.center) to (83.center);
		\draw (83.center) to (80.center);
		\draw [style=QWire] (88.center) to (89.center);
		\draw [style=QWire] (86.center) to (87.center);
		\draw [style=QWire, bend right=90, looseness=2.00] (87.center) to (90.center);
		\draw [style=QWire, in=90, out=-90, looseness=1.25] (90.center) to (92.center);
		\draw [style=QWire, bend left=90, looseness=2.00] (92.center) to (96.center);
		\draw [style=QWire, in=270, out=90] (96.center) to (97.center);
		\draw [style=QWire, in=270, out=90] (95.center) to (94.center);
		\draw [style=QWire, bend left=90, looseness=2.00] (94.center) to (93.center);
		\draw [style=QWire, in=90, out=-90, looseness=1.25] (93.center) to (91.center);
		\draw [style=QWire, bend right=90, looseness=2.00] (91.center) to (88.center);
		\draw [style=QWire] (105.center) to (101.center);
		\draw [style=QWire] (102.center) to (99.center);
		\draw [style=QWire, bend left=90, looseness=2.00] (99.center) to (106.center);
		\draw [style=QWire, in=90, out=-90, looseness=1.25] (106.center) to (109.center);
		\draw [style=QWire, bend right=90, looseness=2.00] (109.center) to (104.center);
		\draw [style=QWire, in=270, out=90] (104.center) to (108.center);
		\draw [style=QWire, in=270, out=90] (103.center) to (107.center);
		\draw [style=QWire, bend right=90, looseness=2.00] (107.center) to (98.center);
		\draw [style=QWire, in=90, out=-90, looseness=1.25] (98.center) to (100.center);
		\draw [style=QWire, bend left=90, looseness=2.00] (100.center) to (105.center);
		\draw [style=thick gray dashed edge] (157.center) to (158.center);
		\draw [style=thick gray dashed edge] (158.center) to (159.center);
		\draw [style=thick gray dashed edge] (159.center) to (160.center);
		\draw [style=thick gray dashed edge] (160.center) to (161.center);
		\draw [style=thick gray dashed edge] (161.center) to (162.center);
		\draw [style=thick gray dashed edge] (162.center) to (163.center);
		\draw [style=thick gray dashed edge] (163.center) to (164.center);
		\draw [style=thick gray dashed edge] (164.center) to (165.center);
		\draw [style=thick gray dashed edge] (165.center) to (166.center);
		\draw [style=thick gray dashed edge] (166.center) to (167.center);
		\draw [style=thick gray dashed edge] (167.center) to (156.center);
		\draw [style=thick gray dashed edge] (156.center) to (157.center);
		\draw [style=QWire] (179.center) to (177.center);
		\draw [style=QWire] (176.center) to (178.center);
		\draw (198.center) to (195.center);
		\draw (195.center) to (196.center);
		\draw (196.center) to (197.center);
		\draw (197.center) to (198.center);
		\draw (199.center) to (200.center);
		\draw (200.center) to (201.center);
		\draw (201.center) to (202.center);
		\draw (202.center) to (199.center);
		\draw [style=QWire] (210.center) to (208.center);
		\draw [style=QWire] (203.center) to (209.center);
		\draw [style=QWire] (215.center) to (216.center);
		\draw [style=QWire] (218.center) to (217.center);
		\draw (224.center) to (221.center);
		\draw (221.center) to (222.center);
		\draw (222.center) to (223.center);
		\draw (223.center) to (224.center);
		\draw (225.center) to (226.center);
		\draw (226.center) to (227.center);
		\draw (227.center) to (228.center);
		\draw (228.center) to (225.center);
		\draw [style=QWire] (236.center) to (234.center);
		\draw [style=QWire] (229.center) to (235.center);
		\draw [style=QWire] (233.center) to (232.center);
		\draw [style=QWire] (230.center) to (231.center);
		\draw [style=QWire] (240.center) to (239.center);
		\draw [style=QWire] (238.center) to (237.center);
		\draw [style=QWire] (241.center) to (242.center);
		\draw [style=QWire] (244.center) to (243.center);
	\end{pgfonlayer}
\end{tikzpicture}
,
\eeq
for arbitrary CPTP maps $A$ and $B$.
\end{proof}

Such reasoning also applies to more general higher-order processes \cite{kissinger2017categorical,bisio2019theoretical}. They can all be viewed as particular classes of CP maps between quantum systems; what varies is how they are `normalised' and how they can be composed, particularly while preserving the normalisation constraints. As our formalism does not impose any constraints on normalisation and does not impose any constraints on composition, it seems clear that it should be able to simulate any such higher-order process. However, we will leave the details and explore the consequences of this in future work.
\begin{problem}
Does $\QSym$ subsume the frameworks for higher-order processes of  \cite{kissinger2017categorical,bisio2019theoretical}? Moreover, can $\QSym$ be viewed as a higher-order process theory in the sense of  \cite{wilson2021causality}?
\end{problem}

\subsection{Beyond quantum theory}

So far in this section, we have focused on the example of quantum theory, in particular, with the aim of constructing our deterministic time-neutral quantum theory $\QSym$. It is clear, however, that a great deal of this is much more broadly applicable.

In particular, the key construction of quotienting by non-zero scalars makes sense for any process theory in which the scalars are $\mathds{R}^+$. This includes arbitrary generalised probabilistic theories (GPTs) \cite{HardyAxiom,Barrett}, operational probabilistic theories (OPTs) \cite{Chiri1,Chiri2,d2017quantum}, categorical probabilistic theories \cite{gogioso2017categorical},  and certain kinds of causal-inferential theories \cite{schmid2020unscrambling}. Moreover, the construction can also apply when the scalars belong to an arbitrary commutative semiring with no zero divisors; it therefore applies to many of the more exotic theories found in, for example,   \cite{gogioso2017fantastic}. The only subtlety to watch out for here is that to obtain anything interesting in any of these cases we need to go beyond the causal processes (i.e., the analogue of $\QPhys$), and instead quotient the supertheory that does not impose any causality constraint (i.e., the analogue of $\QCalc$).

Even after this issue is addressed, it is not obvious whether these alternative theories yield a time-neutral theory. It is therefore interesting to investigate what properties a physical theory, $\mathbf{G}$, must satisfy so that $\mathbf{G}/\!\!\sim$ is a time neutral theory.

Exploring the properties of time-neutral physical theories is therefore an important direction for future research. In particular, one can ask whether we can find a set of principles that single out $\QSym$ from the space of all possible time neutral theories.
Or, in other words, can we find a reconstruction of time symmetric quantum theory from well-motivated principles? The reconstruction of  \cite{reconstruction} may well be a useful starting point for this, as it introduces a time symmetric version of the purification postulate which was itself used to great effect in the reconstruction of  \cite{Chiri2}.

\subsection{Discussion}

By taking a process-theoretic approach to the problem of time neutrality we have subsumed seemingly different approaches to the problem, and shown that they amount to different conventional choices within the same underlying physical theory. Moreover, we have shown that this can be elegantly captured via a particular type of quotienting of the theory. However, the approach also yields several further benefits:
\begin{itemize}
\item Arguably, this is a much simpler presentation of the theory, at least, for those who are already familiar with categorical quantum mechanics, and means that tools such as the ZX calculus \cite{CD2, coecke2021kindergarden} can now be applied to the study of time neutral theories.
\item The process-theoretic nature of the construction lends itself to generalisation to alternative physical theories such as GPTs. Although the construction can be applied to any GPT, it need not yield a time-neutral theory. Identifying the GPT principles that imply time neutrality under this construction is therefore an interesting direction for future research.
\item The process-theoretic nature also means that it should be more directly comparable to the frameworks for higher-order processes such as Refs.~\cite{kissinger2017categorical,
bisio2019theoretical,
wilson2021causality}. Section \ref{sec:HigherOrder} provides a simple proof-of-principle application in this direction.
\item This approach may well also more readily extend beyond the finite-dimensional setting which has so far dominated the literature on time symmetric theories, in particular, by considering the categorical formalism of  \cite{StefanoInfinite, shaikh2022feynman}.
\item It can also be considered in the context of the functorial approach to field theories presented in   \cite{gogioso2020functorial}. This may provide a formal way to study how time neutral theories behave in a space-time context, and, in particular, whether they lead to contradictions with the non-signalling principle.
\item Finally, one can ask whether it is possible to find a reconstruction of quantum theory that reproduces the time neutral theory, $\QSym$, rather than the standard causal theory, $\QPhys$. Recent process-theoretic reconstructions such as Refs.~\cite{reconstruction} and \cite{tull2018categorical} provide natural starting points for this question.
\end{itemize}
The biggest issue with the current formalism, however, is that whilst the theory is time neutral, the presentation of the theory still involves an artificial division of the systems associated with a given process into input and output systems. In  \cite{operads} we show how this can be overcome by moving from a categorical formulation of process theories to an operadic formulation based on the work of Refs.~\cite{spivak2017string,patterson2021wiring}. 

\section*{Acknowledgements}
JHS would like to thank Tony Short for an interesting discussion that helped to shape the early stages of the project. JHS and MES would also like to thank Lucien Hardy for insightful discussions about his work on time symmetry in operational theories.  BC thanks Carlo Rovelli and Augustin Vanrietvelde for discussions on time symmetry. The authors also thank Arthur Parzygnat for helpful feedback on the work.

JHS and MES were supported by the Foundation for Polish Science (IRAP project, ICTQT, contract no. MAB/2018/5, co-financed by EU within Smart Growth Operational Programme).  BC and SG were supported by the John Templeton Foundation through grant 61466, The Quantum Information Structure of Spacetime (qiss.fr). The opinions expressed in this publication
are those of the authors and do not necessarily reflect the
views of the John Templeton Foundation.

All of the diagrams within this manuscript were prepared using Aleks Kissinger's TikZit.

\section*{Author contribution statement}

All authors contributed to the first draft of the manuscript. JHS and MES handled journal submission and subsequent revisions. AI tools were only used in the final stage to check for spelling, grammar, and typographic errors -- the scientific content of this paper is all human generated.

\bibliography{main}

\newpage
\appendix
\section{Diagrams for groups and their representations} \label{app:Groups}

Below is a recap of diagrammatic notation for group representations.
The correspondence between strongly complementary structures and finite-dimensional group algebras was introduced in \cite{coecke2012strong} for the Abelian case and in \cite{gogioso2017fully,gogioso2019generalised} for the general case.
Diagrammatic notation for group representations previously appeared in \cite{vicary2013topological,gogioso2015fourier}---in the context of quantum algorithms and the quantum Fourier transform---and was extensively discussed in \cite{gogioso2017thesis,gogioso2019dynamics}---in the context of quantum dynamics.

\subsection{Diagrams for groups}
We denote a group $G$ by a new kind of wire,
\beq
\begin{tikzpicture}
	\begin{pgfonlayer}{nodelayer}
		\node [style=none] (0) at (0, 1.25) {};
		\node [style=none] (1) at (0, -1.25) {};
		\node [style=right label] (2) at (0, -0.25) {$G$};
	\end{pgfonlayer}
	\begin{pgfonlayer}{edgelayer}
		\draw [GWire] (0.center) to (1.center);
	\end{pgfonlayer}
\end{tikzpicture}
,
\eeq
such that the elements of the group, $g\in G$, are states of this wire,
\beq
\begin{tikzpicture}
	\begin{pgfonlayer}{nodelayer}
		\node [style=none] (0) at (0, 0.75) {};
		\node [style=point] (1) at (0, -0.25) {$g$};
		\node [style=right label] (2) at (0, 0.5) {$G$};
	\end{pgfonlayer}
	\begin{pgfonlayer}{edgelayer}
		\draw [GWire] (0.center) to (1);
	\end{pgfonlayer}
\end{tikzpicture}
.
\eeq
In particular, we denote the identity element of the group by:
\beq
\begin{tikzpicture}
	\begin{pgfonlayer}{nodelayer}
		\node [style=none] (0) at (0, 0.75) {};
		\node [style=point] (1) at (0, -0.25) {$I$};
		\node [style=right label] (2) at (0, 0.5) {$G$};
	\end{pgfonlayer}
	\begin{pgfonlayer}{edgelayer}
		\draw [GWire] (0.center) to (1);
	\end{pgfonlayer}
\end{tikzpicture}
.
\eeq
We then introduce the defining group operations, namely group multiplication and inversion:
\beq
\input{GComp.tikz} \quad \text{and} \quad \input{GInverse.tikz}
\eeq
respectively. These can be defined by their action on the group elements via:
\beq
\input{GCompCond1.tikz}=\begin{tikzpicture}
	\begin{pgfonlayer}{nodelayer}
		\node [style=none] (0) at (0, 1.25) {};
		\node [style=point] (7) at (0, -0.5) {$g'\cdot g$};
		\node [style=none] (11) at (0, 1.25) {};
		\node [style=right label] (12) at (0, 1.25) {$G$};
	\end{pgfonlayer}
	\begin{pgfonlayer}{edgelayer}
		\draw [GWire] (0.center) to (7);
	\end{pgfonlayer}
\end{tikzpicture}
 \quad\text{and}\quad \input{GInverseDef.tikz}=\begin{tikzpicture}
	\begin{pgfonlayer}{nodelayer}
		\node [style=point] (7) at (0, -0.75) {$g^{-1}$};
		\node [style=right label] (12) at (0, 1) {$G$};
		\node [style=none] (13) at (0, 1.25) {};
	\end{pgfonlayer}
	\begin{pgfonlayer}{edgelayer}
		\draw [GWire] (13.center) to (7);
	\end{pgfonlayer}
\end{tikzpicture}
.
\eeq
Various important properties of groups, including associativity of multiplication, can be captured elegantly by these diagrams:
\beq
\input{GCompAssoc.tikz}=\input{GCompAssoc1.tikz}, 
\eeq
that the identity element is the unit, 
\beq\input{GCompUnit.tikz}=\begin{tikzpicture}
	\begin{pgfonlayer}{nodelayer}
		\node [style=none] (0) at (0, 1.25) {};
		\node [style=none] (1) at (0, -1.25) {};
		\node [style=right label] (2) at (0, -0.25) {$G$};
	\end{pgfonlayer}
	\begin{pgfonlayer}{edgelayer}
		\draw [GWire] (0.center) to (1.center);
	\end{pgfonlayer}
\end{tikzpicture}
=\input{GCompUnit2.tikz},\eeq
that the group inverse is involutive,
\beq \input{GInverseSelfInverse.tikz}= \begin{tikzpicture}
	\begin{pgfonlayer}{nodelayer}
		\node [style=none] (0) at (0, 1.25) {};
		\node [style=none] (1) at (0, -1.25) {};
		\node [style=right label] (2) at (0, -0.25) {$G$};
	\end{pgfonlayer}
	\begin{pgfonlayer}{edgelayer}
		\draw [GWire] (0.center) to (1.center);
	\end{pgfonlayer}
\end{tikzpicture}
, \eeq
and that multiplication is ``antisymmetric'',
\beq
 \input{GCompAntiSym.tikz}=\input{GCompAntiSym1.tikz}.
\eeq

To capture more of the group structure diagrammatically, we introduce a copy map and a discard map for these new systems, denoted:
\beq
\input{GCopy.tikz}\quad \text{and}\quad \begin{tikzpicture}
	\begin{pgfonlayer}{nodelayer}
		\node [style=none] (0) at (0, -0.75) {};
		\node [style=right label] (1) at (0, -0.75) {$G$};
		\node [style=none] (5) at (0, 0.25) {};
		\node [style=upground] (6) at (0, 0.5) {};
	\end{pgfonlayer}
	\begin{pgfonlayer}{edgelayer}
		\draw [GWire] (0.center) to (5.center);
	\end{pgfonlayer}
\end{tikzpicture}

\eeq
respectively. These, too, can be defined by their action on group elements via:
\beq
\input{GCopyDef.tikz} = \input{GCopyDef1.tikz}\quad \text{and}\quad \begin{tikzpicture}
	\begin{pgfonlayer}{nodelayer}
		\node [style=point] (0) at (0, -0.5) {$g$};
		\node [style=right label] (1) at (0, 0) {$G$};
		\node [style=none] (5) at (0, 0.5) {};
		\node [style=upground] (6) at (0, 0.75) {};
	\end{pgfonlayer}
	\begin{pgfonlayer}{edgelayer}
		\draw [GWire] (0) to (5.center);
	\end{pgfonlayer}
\end{tikzpicture}
=\input{empty.tikz}.
\eeq
Important properties of these processes, including associativity of copying, can also be captured diagrammatically:
\beq
\input{GCopyAssoc1.tikz}=\input{GCopyAssoc2.tikz}, \eeq
that the discard is the counit for copy,
\beq \input{GCopyUnit1.tikz}=\begin{tikzpicture}
	\begin{pgfonlayer}{nodelayer}
		\node [style=none] (0) at (0, 1.25) {};
		\node [style=none] (1) at (0, -1.25) {};
		\node [style=right label] (2) at (0, -0.25) {$G$};
	\end{pgfonlayer}
	\begin{pgfonlayer}{edgelayer}
		\draw [GWire] (0.center) to (1.center);
	\end{pgfonlayer}
\end{tikzpicture}
=\input{GCopyUnit2.tikz},\eeq
and that copying is symmetric,
\beq \input{GCopySym.tikz}=\input{GCopy.tikz},
\eeq
Finally, we can use the interaction of these new processes with the group multiplication and inverse in order to capture more structure of the group. In particular,
\beq
\input{GInteract1.tikz}=\input{GInteract2.tikz},
\eeq
ensures that the group inverse behaves as expected. Moreover, one can show that copying and multiplication form a bialgebra,
\beq
 \quad \input{GInteractBi1.tikz}=\input{GInteractBi2.tikz}, \eeq
 that group multiplication is causal,
 \beq\input{GInteractDisc1.tikz}=\begin{tikzpicture}
	\begin{pgfonlayer}{nodelayer}
		\node [style=none] (0) at (0, -0.75) {};
		\node [style=right label] (1) at (0, -0.75) {$G$};
		\node [style=none] (5) at (0, 0.25) {};
		\node [style=upground] (6) at (0, 0.5) {};
	\end{pgfonlayer}
	\begin{pgfonlayer}{edgelayer}
		\draw [GWire] (0.center) to (5.center);
	\end{pgfonlayer}
\end{tikzpicture}
\begin{tikzpicture}
	\begin{pgfonlayer}{nodelayer}
		\node [style=none] (0) at (0, -0.75) {};
		\node [style=right label] (1) at (0, -0.75) {$G$};
		\node [style=none] (5) at (0, 0.25) {};
		\node [style=upground] (6) at (0, 0.5) {};
	\end{pgfonlayer}
	\begin{pgfonlayer}{edgelayer}
		\draw [GWire] (0.center) to (5.center);
	\end{pgfonlayer}
\end{tikzpicture}
,
 \eeq
 and that the inverse is copied by the copy operation,
 \beq \input{GInteractInv2.tikz} = \input{GInteractInv1.tikz}.
\eeq

\subsection{Group representations}
In this section, we define group representations on systems in $\QCalc$ (or $\QPhys$) as an interaction between the new systems $G$ and systems from $\QCalc$. Formally we can think of the systems $G$ as being particular (potentially infinite-dimensional) classical systems in which the point distributions correspond to the group elements.

A \emph{causal} group representation of $G$ on some system $Q$ in $\QCalc$ is a process, $\pi$, of the form:
\beq
\input{GRep.tikz}
\eeq
such that the following equations are satisfied:
\beq\label{eq:CausGRep}
\input{GRepDef1.tikz}=\input{GRepDef2.tikz} \quad ,\quad \input{GRepDef3.tikz}=\begin{tikzpicture}
	\begin{pgfonlayer}{nodelayer}
		\node [style=none] (11) at (-1.5, 1) {};
		\node [style=none] (12) at (-1.5, -1) {};
		\node [style=right label] (13) at (-1.5, -0.5) {$Q$};
	\end{pgfonlayer}
	\begin{pgfonlayer}{edgelayer}
		\draw [QWire] (11.center) to (12.center);
	\end{pgfonlayer}
\end{tikzpicture}
\quad\text{and}\quad \input{GRepCausal.tikz}=\input{GRepCausal1.tikz},
\eeq
The first two equations guarantee that this is a valid representation and the third guarantees that this is a \emph{causal} representation. Note that, if we are working with $\QPhys$ rather than $\QCalc$, then this last condition is automatically satisfied.

If $Q$ is strictly quantum, that is, a system $\mathcal{H}$, then these causal representations are necessarily unitary; that is, they satisfy
\beq\label{eq:unitRep}
\input{GUnitaryRep.tikz}\ \ =\ \ \input{GUnitaryRep1.tikz}\ \  :: \rho \mapsto U_g\rho U_g^\dagger
\eeq
for all $g\in G$. To see this, simply note that the axioms for group representations imply that the processes on the left-hand side are necessarily reversible  and trace-preserving (i.e., causal),  and reversible CPTP maps are necessarily unitary supermaps.

For finite-dimensional classical systems $\mathds{X}$, the only representation that we will typically use is the trivial representation, namely:
\beq
\input{GTrivialRep.tikz}=\input{GTrivialRep1.tikz}.
\eeq
For many groups, this will in fact be the only possible representation for finite-dimensional classical systems.

We can compose representations of single systems to define representations of composite systems via:
\beq
\input{GCompRep.tikz}=\input{GCompRep1.tikz}.
\eeq
Thus, if we compose a quantum and a classical representation we end up with a representation:
\beq
\input{GCQRep.tikz} = \input{GCQRep1.tikz} = \input{GCQRep2.tikz}.
\eeq
That is, the quantum part of the composite system may transform non-trivially under the action of the group, but the classical part is left invariant. Finally, note that the representation on a trivial system is necessarily trivial, that is:
\beq
\input{GRepTrivial.tikz}=\begin{tikzpicture}
	\begin{pgfonlayer}{nodelayer}
		\node [style=none] (0) at (0, -1) {};
		\node [style=right label] (1) at (0, -0.75) {$G$};
		\node [style=none] (5) at (0, 0) {};
		\node [style=upground] (18) at (0, 0.25) {};
	\end{pgfonlayer}
	\begin{pgfonlayer}{edgelayer}
		\draw [GWire, in=270, out=90] (0.center) to (5.center);
	\end{pgfonlayer}
\end{tikzpicture}
.
\eeq

\subsection{Intertwiners}

Next we introduce \emph{intertwiners}, i.e., processes in $\QCalc$ (or $\QPhys$) that are symmetric with respect to the group action. In particular, they are characterised by the covariance condition:
\beq
\input{GIntertwiner.tikz} = \ \ \input{GIntertwiner1.tikz}.
\eeq
Using the fact that representations on the trivial system and classical systems are trivial representations, one can easily show that states and measurements are intertwiners if and only if they are invariant under the group action, namely:
\beq 
\input{GStateInt.tikz}\ \  =\ \  \input{GStateInt1.tikz}
\qquad \text{and} \qquad
\input{GMeasInt.tikz}\ \  =\ \  \input{GMeasInt1.tikz}.
\eeq

\subsection{Dual systems and conjugate representations}

Finally, we introduce dual systems and their representations. Systems in $\QCalc$  have duals, which we will now denote explicitly by arrows. That is, a generic process is denoted by
\beq
\input{genericQPart.tikz}
\eeq
and given by a CP map 
\beq
\mathcal{E}: \mathcal{B[H]} \otimes \left(\bigoplus_{x\in \mathds{X}} \mathcal{B}[\mathds{C}]\right) \otimes \mathcal{B[H']^*} \otimes \left(\bigoplus_{x'\in \mathds{X'}} \mathcal{B}[\mathds{C}]\right)^* \to \mathcal{B[K]} \otimes \left(\bigoplus_{y\in \mathds{Y}} \mathcal{B}[\mathds{C}]\right) \otimes \mathcal{B[K']^*} \otimes \left(\bigoplus_{y'\in \mathds{Y'}} \mathcal{B}[\mathds{C}]\right)^*.
\eeq
That is, when the arrow points up, we use the primal vector space; when it points down, we use the dual vector space.

It is straightforward to verify that this is time symmetric and time neutral in the sense of section~\ref{sec:duals}.

Given a representation on a system $\mathcal{Q}^\uparrow$, we define the conjugate representation $\pi^*$ on the system $\mathcal{Q}^\downarrow$ as follows:
\beq\label{eq:congRep}
\input{GcongRep1.tikz} \quad := \quad \input{GCongRep.tikz}.
\eeq
Note that if $Q^\uparrow$ is classical, $\mathds{X}^\uparrow$, then just as the representations on $\mathds{X}^\uparrow$ were necessarily trivial, so is the conjugate representation on $\mathds{X}^\downarrow$.

From this definition it is easy to see that cups \& caps are intertwiners for composite representations on the inputs. That is,
\begin{align}\label{eq:capIntertwiner}
\input{GCapInter0.tikz}\quad &=  \quad\input{GCapInter.tikz}\quad =\quad \input{GCapInter1.tikz}\\ &=\quad \input{GCapInter2.tikz}\quad =\quad \input{GCapInter3.tikz}\\ &=\quad \input{GCapInter4.tikz}.
\end{align}
In the first equality, we use the assumption that the representation on the input is a composite of representations together with the definition of the conjugate representation. In the second equality, we use the definition of cups and caps; in the third, the definition of a representation; in the fourth, a key result about groups; and in the final equality, the definition of a representation again.

\end{document}